\pdfoutput=1
\documentclass[a4paper,12pt,reqno]{article}
\usepackage[markup=underlined]{changes}
\RequirePackage{etex}

\usepackage{mathpazo}

\usepackage{etex}
\usepackage{graphicx}
\usepackage{setspace}
\usepackage{amsmath}
\usepackage{amsfonts}
\usepackage{titlesec}
\usepackage{pictex}
\usepackage{comment}
\usepackage{amsthm}
\usepackage{graphics,epsfig,verbatim,bm,latexsym,url,amsbsy}
\usepackage{rotating}
\usepackage[authoryear,round]{natbib}
\usepackage{float}
\usepackage[position=bottom]{subfig}
\usepackage{mathrsfs}
\usepackage{multirow}
\usepackage{array}
\usepackage{bigints}
\usepackage{bbm}
\usepackage{geometry}
\usepackage{bbm}
\usepackage[ruled,vlined]{algorithm2e}
\usepackage{soul}
\usepackage{mathtools}
\usepackage{enumitem}
\usepackage{amssymb}

\DeclarePairedDelimiterX{\inp}[2]{\langle}{\rangle}{#1, #2}

\usepackage{hyperref}
\hypersetup{
	colorlinks=true,
	linkcolor=red!75!black,
	citecolor=blue!75!black,
	filecolor=magenta,      
	urlcolor=cyan,
}

\newcommand{\Ex}{\mathbb{E}}
\newcommand{\N}{\mathbb{N}}

\renewcommand{\Pr}{\mathbb{P}}

\newtheorem{manualtheoreminner}{Theorem}
\newenvironment{manualtheorem}[1]{%
  \renewcommand\themanualtheoreminner{#1}%
  \manualtheoreminner
}{\endmanualtheoreminner}

\theoremstyle{definition} 
\theoremstyle{definition} 
\theoremstyle{definition} 

\theoremstyle{definition} \newtheorem{claim}{Claim}
\theoremstyle{definition} \newtheorem{definition}{Definition}
\theoremstyle{definition} 
\theoremstyle{definition} 
\theoremstyle{definition} \newtheorem{lemma}{Lemma}
\theoremstyle{definition} \newtheorem{theorem}{Theorem}
\theoremstyle{definition} 
\theoremstyle{definition} 
\theoremstyle{definition} 
\theoremstyle{definition} 
\theoremstyle{definition}\newtheorem{proposition}{Proposition}
\theoremstyle{definition} 
\theoremstyle{definition} 
\theoremstyle{definition} 
\theoremstyle{definition} 
\theoremstyle{definition} 
\theoremstyle{definition}

\theoremstyle{definition} 
\usepackage{mathtools}
\titlespacing*{\paragraph}{0pt}{1.25ex plus 1ex minus .2ex}{0.5em}

\titleformat{\section}
		{\bfseries\center}
         {\thesection}
        {0.5em}
        {}
        []

\titleformat{\subsection}[runin]
        {\normalfont\bfseries}
        {\thesubsection}
        {0.5em}
        {\addperiod}
        []
\newcommand{\addperiod}[1]{#1.}
\allowdisplaybreaks

\title{\scshape{\Large{\textbf{Informational Puts}}}}

\author{\makebox[.25\linewidth]{{Andrew Koh}\thanks{MIT Department of Economics; email: \protect\texttt{ajkoh@mit.edu}}}\\{MIT} 
\and 
\makebox[.25\linewidth]{Sivakorn Sanguanmoo\thanks{MIT Department of Economics; email: \protect\texttt{sanguanm@mit.edu}}} \\ {MIT} 
\and \makebox[.25\linewidth]{Kei Uzui\thanks{MIT Department of Economics; email: \protect\texttt{kuzui@mit.edu} 
{{
\\~\\
\emph{First version: December 2023}. We are especially grateful to Drew Fudenberg and Stephen Morris for guidance, support, and many helpful suggestions. We also thank Daron Acemoglu, Matt Elliott, Nobuhiro Kiyotaki, Daniel Luo,  Daisuke Oyama, Parag Pathak, Satoru Takahashi, Iv\'an Werning, Alex Wolitzky, Muhamet Yildiz, as well as audiences at Cambridge University, EC'24, the Econometric Society North American and Asian Meetings, Nuffield College Oxford, and MIT Finance, Macro, and Theory Lunches for helpful comments. A very preliminary version of this paper appeared as an extended abstract in the Proceedings of the 25th ACM Conference on Economics and Computation (EC'24) under the title `Full Dynamic Implementation'.
}} \\
}} \\
{MIT}
}

\begin{document}
\maketitle
\thispagestyle{empty}

\begin{abstract}
We analyze how dynamic information should be provided to uniquely implement the largest equilibrium in binary-action coordination games. The designer offers an \emph{informational put}: she stays silent if players choose her preferred action, but injects asymmetric and inconclusive public information if they lose faith. There is (i) no multiplicity gap: the largest (partially) implementable
equilibrium can be implemented uniquely; and (ii) no commitment gap: the policy is sequentially optimal. Our results have sharp implications for the design of policy in coordination environments.  
\end{abstract}


\clearpage 
\setcounter{page}{1}
\newpage
\section{Introduction}

Many economic environments feature \emph{uncertainty} about a payoff-relevant state, \emph{coordination motives}, and \emph{frictions} in revising one's actions. These components are pervasive across many aspects of social and economic life. In macroeconomics, firms might be uncertain about underlying demand conditions, face complementarities in pricing decisions \citep{nakamura2010monetary}, and set prices in a staggered manner because of menu costs or inattention \citep{calvo1983staggered,reis2006inattentive}. In finance, creditors might be uncertain about the debtor's solvency, have incentives to run if others do the same \citep{diamond1983bank,goldstein2005demand}, but might only be able to withdraw their principal at the staggered expiration of individual contracts \citep{he2012dynamic}. The same applies, mutatis mutandis to industrial organization,\footnote{Consumers might be uncertain about a product's quality, have incentive to adopt the same product as others \citep{farrell1985standardization,ellison2000neo}, and face stochastic adoption opportunities \citep*{biglaiser2022should}.}
industrial policy,\footnote{Firms might only wish to industrialize if others do the same \citep*{rosenstein1943problems,murphy1989industrialization} but might face staggered decisions as to which sector to operate in \cite{matsuyama1991increasing}.} and political revolutions.


Equilibria of such environments are sensitive to the flow of information over time. Consider a player who, at any history of the game, finds herself with the opportunity to re-optimize her action. The fundamental state matters for her flow payoffs, so her decision must depend on her current beliefs. What is more, since she faces switching frictions, her decision also depends on her beliefs about what future agents will do. But those beliefs depend, in turn, on what she expects future players to learn, as well as future players' beliefs about the play of agents even further out into the future. Thus, the evolution of future beliefs---even those arbitrarily distant---multiply back to shape incentives in the present.

We are interested in how dynamic information should be optimally provided by a policymaker who (i) prefers a certain action e.g., for firms not to raise prices, for investors not to run on a bank or currency, for citizens not to protest, and so on; but (ii) is concerned about the equilibrium multiplicity endemic to such environments. Our main result (Theorem \ref{thrm:main}) characterizes the \emph{value}, \emph{form}, and \emph{sequential optimality} of dynamic information policies under {adversarial equilibrium selection}:

\begin{enumerate}[leftmargin = *] 

    \item \textbf{Value: No multiplicity gap.} The optimal policy uniquely implements the upper-bound on the time path of aggregate play. Thus, there is no multiplicity gap: the largest partially implementable equilibrium can also be implemented fully (i.e., as the \emph{unique} equilibrium). This is in sharp contrast to recent work on static implementation via information design in supermodular games which finds there generically exists a gap even with private information \citep*{morris2024implementation}, or with both private information and transfers \citep*{halac2021rank}. 
    
    \item \textbf{Form: Informational puts.}  
    If players play the designer's preferred action, the designer \emph{stays silent}. If, however, the aggregate measure of agents playing the designer's preferred action falls too far short of a dynamic target specified by the policy,\footnote{We emphasize that it is not crucial that information is only delivered off-path; we discuss this in Section \ref{sec:robustness}. In a finite player version of the model developed in Online Appendix \ref{appendix:finite_players}, informational puts delivers information on-path, albeit with small probability.}  the designer injects an \emph{asymmetric} and \emph{inconclusive} public signal---this is the \emph{informational put}.\footnote{This is analogous to the ``Fed put'' in which the Fed's history of intervening to halt market downturns has arguably created the belief that they are insured against downside risk \citep{miller2002moral}. This is \emph{as if} the Fed has offered the market a put option as insurance against downturns. In our setting, the designer steps in to inject information when players start switching to action $0$ which, as we will show, with high probability induces aggregate play to correct---\emph{as if} the designer has offered a put option as insurance against strategic uncertainty about the play of future players.}
    
    The signal is asymmetric such that the probability that agents become a little more confident is far higher than the probability that agents become much more pessimistic. These small but high-probability movements are chosen to be in the direction of the dominance region---at which playing the designer's preferred action is strictly dominant---and are \emph{chained together} such that the unique equilibrium of the subgame is for future players to play the designer-preferred action.\footnote{This is done via a "contagion argument" which can be viewed as the dynamic analog of interim deletion of strictly dominated strategies in static games of incomplete information.} Thus, informational puts deliver \emph{insurance against strategic uncertainty} by inducing confidence that future players will play the designer's preferred action. The signal is inconclusive such that, even if agents turn pessimistic, they do not become excessively so---this will be important for sequential optimality. 
    
    \item \textbf{Sequential optimality: No commitment gap.} Our dynamic information policy is constructed such that at every history, the designer has no incentive to deviate.\footnote{With the caveat that for a small set of histories, deviation incentives can be made arbitrarily small. For such histories, this is simply because optimal information policies continuing from those histories do not exist. Nonetheless, this can be approached via a sequence of policies so that the designer's deviation incentives vanishes along this sequence. This openness property is also typical of static full implementation environments as highlighted by \cite*{morris2024implementation}.}  
    Thus, there is no intertemporal commitment gap: whatever can be implemented with commitment to the dynamic information structure can also be implemented when the sender can continually re-optimize her choice of dynamic information.\footnote{We further emphasize that sequential optimality is not given---we offer examples of policies which are optimal but not sequentially optimal.}
    Sequentially optimality arises through the delicate interaction between properties of our policy: \emph{asymmetry}, \emph{chaining}, and \emph{inconclusiveness}. {Asymmetric} off-path information are {chained} together to obtain full implementation at all states in which the designer-preferred action is not strictly dominated. Then, {inconclusive} off-path information ensures that, even if agents turn pessimistic, full implementation is still guaranteed.\footnote{We will later observe that our policy can often be {implemented via dynamic cheap talk} with no `within period' or `between period' commitment; see Section \ref{sec:robustness}.}
\end{enumerate}

What is distinctive about providing dynamic information vis-a-vis providing transfers or designing the extensive form?  
First, information is \emph{less powerful}: beliefs are martingales and, unlike transfers, do not directly enter players' payoffs---this imposes severe constraints on what incentives can be delivered. Nonetheless, we will show that although information is less powerful on its own, it can be chained together over time to close the gap between full and partial implementation. Second, information is \emph{more flexible}: the designer has the freedom to design any belief martingale and, moreover, delivering information does not necessarily hurt the designer. This flexibility can be leveraged to shape the designer's continuation incentives---ensuring that her counterfactual selves at off-path histories are willing to follow through with the promised information. That dynamic information can be made sequentially optimal stands in sharp contrast to work on subgame implementation where the designer (viewed as a player in the game) does not generally wish to follow-through with off-path threats \citep{chakravorty2006credible},\footnote{See \cite{sjostrom2002implementation}, Section 4.6 for a survey of work where the designer is a player. For example, \cite{chakravorty2006credible} write \emph{``Generally, it is not in the designer's best interest to go through with the reward/punishment in the “subgame” arising from some disequilibrium play.''}}
and on contracting e.g., deposit insurance might not always be credible to prevent bank runs \citep{dybvig2023nobel}.\footnote{See also work on mechanism design with limited commitment \citep{laffont1988dynamics,bester2001contracting,liu2019auctions,doval2022mechanism} and macroeconomics \citep{halac2014fiscal} where time-inconsistency plays a crucial role.}

Our results have simple and sharp implications for policy in coordination environments. It is often held that to eliminate `bad' equilibria, substantial information must be delivered on the path of equilibrium play. We offer an alternative view: as long as the designer's preferred action is not strictly dominated, silence backed by the credible promise to inject asymmetric and inconclusive information suffices.

\paragraph{Related Literature} Our results relate most closely to recent work on full implementation in supermodular games via information design \citep*{morris2024implementation,inostroza2023adversarial,li2023global}. In this literature, information design induces non-degenerate higher-order beliefs, and this is important to obtain uniqueness via a "contagion argument" over the type space. By contrast, our dynamic information is public and higher-order beliefs are degenerate but we leverage a distinct kind of "intertemporal contagion".
A key takeaway from this literature is that there is typically a gap between the designer's value under adversarial equilibrium selection, and under designer-favorable selection (what we call a ``multiplicity gap''); by contrast, we show that for dynamic binary-action supermodular games there is no such gap.

Also related is the important and complementary work of \cite{basak2020diffusing} and \cite{basak2024panics}. We highlight several substantive differences. First, we study different dynamic games: in \cite{basak2020diffusing,basak2024panics} players make a once-and-for-all decision on whether to play the risky action, and they focus on regime change games---both features play a key role in their analysis.\footnote{
\cite{basak2024panics} note: \emph{``...in our
dynamic coordination game, the agents can choose when to attack, and the attack is the only irreversible
choice. Moreover, our principal has a simple objective: she wants to avoid a disaster. We argue that an early warning may not be as effective (if not completely ineffective) in the absence of these features.''} In this regard, our contribution is to offer dynamic information policies which work in a broader class of environments. 
} For instance, \cite{basak2024panics} assume on some states regime change is inevitable so a designer who wishes to preserve the regime still finds it optimal to disclose information. In our environment, agents play a general binary-action supermodular game and our designer's payoff is any increasing functional from the path of aggregate play. Specialized to a dynamic regime change game, this can capture the designer's incentives to prolong the regime or minimize the attack. Importantly, our dynamic information policies---and the reasons they work---are quite distinct. We discuss this connection more thoroughly in Section \ref{section:discussion}. 

Our paper also relates to work on the equilibria of dynamic coordination games. An important paper of \cite{gale1995dynamic} studies a complete information investment game where players can decide when, if ever, to make an irreversible investment and investing is payoff dominant.\footnote{See also \cite*{chamley1999coordinating,dasgupta2007coordination,angeletos2007dynamic,mathevet2013tractable,koh2024inertial} all of which study the equilibria of different dynamic coordination games. Recent work by \cite{jehiel2024power} studies implementation via transfers.} 
The main result is that investment succeeds across all subgame perfect equilibria. Our environment and results differ in several substantive ways. For instance, our policy allows the designer to implement the largest equilibria---irrespective of whether it is payoff dominant.

Our results are also connected to the literature on dynamic implementation. \cite{moore1988subgame} show that arbitrary social choice functions can be achieved with large off-path transfers.\footnote{See also \cite*{aghion2012subgame} for a discussion of the lack of robustness to small amounts of imperfect information.
}  \cite{glazer1996virtual} show that virtual implementation of social choice functions can be achieved by appealing to extensive-form versions of \cite{abreu1992virtual} mechanisms.\footnote{See work by \cite{chen2015full} who exploit the freedom to design the extensive-form. \cite{sato2023robust} designs both the extensive-form and information structure a la \cite{doval2020sequential} and further utilizes the fact the designer can design information about players' past moves; by contrast, we fix the dynamic game and past play is observed.} \cite*{chen2023getting} weaken backward induction to initial rationalizability.\footnote{That is, only imposing sequential rationality and common knowledge of sequential rationality at the beginning of the game, but "anything goes" off-path.
} Different from these papers, our designer is substantially more constrained: (i) there is no freedom to design the extensive-form which we take as given; (ii) the designer only offers dynamic information; and (iii) our policy is sequentially optimal.

Our game is one where players have stochastic switching opportunities. Variants of these models have been studied in macroeconomics \citep{diamond1982aggregate,
calvo1983staggered,diamond1989rational,frankel2000resolving}, industrial policy \citep*{murphy1989industrialization,matsuyama1991increasing}, finance \citep{he2012dynamic}, industrial organization \citep*{biglaiser2022should}, and game theory \citep*{burdzy2001fast, matsui1995approach,oyama2002p,kamada2020revision}.\footnote{See also more recent work by \cite*{guimaraes2018dynamic,guimaraes2020dynamic}. \cite{angeletos2016incomplete} offer an excellent survey.} A common insight from this literature is that switching frictions can generate uniqueness, and the \emph{risk-dominant profile} is selected via a process of backward induction. Our contribution is to show how the \emph{largest equilibrium} can be uniquely implemented by carefully designing dynamic information.

Sequential optimality is an important property of our information policy and thus our work relates to recent work studying the role of (intertemporal) commitment in dynamic information design. \cite*{koh2022attention,koh2024persuasion} show by construction that sequential optimality is generally achievable in single-agent stopping problems. It will turn out that sequentially optimal information policies also exist in our environment, but for quite distinct reasons; we discuss this more thoroughly in Section \ref{sec:mainresult}.

\section{Model}
\paragraph{Environment} There is a finite set of states $\Theta = \{\theta_1,\theta_2\ldots ,\theta_n\}$. There is an interior common prior $\mu_0 \in \Delta(\Theta) \setminus \partial \Delta(\Theta)$, where we use $\Delta(\Theta)$ to denote the set of probability measures over $\Theta$ endowed with the Euclidean metric, and use $\partial \Delta(\Theta)$ to denote the boundary of this set. There is a unit measure of players indexed $i \in I := [0,1]$. Time is continuous and indexed $\mathcal{T} := [0,+\infty)$. The action space is binary: $a_{it} \in \{0,1\}$ where $a_{it}$ is $i$'s action at time $t$. Write $A_t := \int a_{it} di$ to denote the proportion of players playing action $1$ at time $t$. 

Working with a continuum of agents makes our analysis cleaner because randomness from individual switching frictions vanishes in the aggregate. A close analog of our result holds for finite players; we develop this in Online Appendix \ref{appendix:finite_players}. 

\paragraph{Payoffs}
The flow payoff for each player is $u: \{0,1\} \times [0,1] \times \Theta \to \mathbb{R}$. We write $\Delta u(A,\theta) := u(1,A,\theta) - u(0,A,\theta)$ to denote the payoff difference from action $1$ relative to $0$ and assume throughout: 
\begin{itemize}
    \item[(i)] \textbf{Supermodularity.} $\Delta u(A,\theta)$ is continuously differentiable and strictly increasing in $A$. 
    \item[(ii)] \textbf{Dominant state.} There exists $\theta^* \in \Theta$ such that $\Delta u(0,\theta^*) > 0$. 
\end{itemize}

Condition (i) states that the game is one of strategic complements. Condition (ii) is a standard richness assumption on the space of possible payoff structures: there exists \emph{some} state $\theta^*$ under which playing action $1$ is strictly dominant.
The payoff of player $i \in I$ is $\int e^{-rt} u(a_{it},A_t,\theta) dt$ where $r > 0$ is an arbitrary discount rate. Each player is endowed with a personal Poisson clock which ticks at an independent rate $\lambda > 0$.
Players can only re-optimize at the ticks of their clocks \citep*{calvo1983staggered,matsui1995approach,frankel2000resolving,burdzy2001fast,kamada2020revision} and the aggregate measure of ticks in the population is non-random.\footnote{By an appropriate continuum law of large numbers \citep{sun2006exact}.}  

Our specification of flow payoffs is quite general, with the caveat that players are homogeneous.\footnote{A similar assumption has been made in static environments by \cite*{inostroza2023adversarial,li2023global} and was weakened by \cite*{morris2024implementation} who characterize optimal private information for full implementation by focusing on potential games with a convexity requirement, which amounts to there not being ``too much heterogeneity" across players. We discuss the role of heterogeneity in Section \ref{sec:robustness}.} We have chosen to focus on a simple but canonical time-aggregator: the discounted flow payoff. This transparently conveys how our policy works. As will be apparent (and formalized in Appendix \ref{appendix:examples}) our results continue to hold for a wider class of dynamic coordination games analyzed in the literature such as:\footnote{Although such environments are not nested within that of the main text, our results remain unchanged: the informational puts policy closes the multiplicity gap and remains sequentially optimal. Appendix \ref{appendix:examples} formalizes these environments and sketches the (minor) required modifications to the proofs.} 
\begin{enumerate}
    \item \emph{Dynamic regime change games.} Let $\Theta$ be totally ordered reflecting the strength of the regime. Action $0$ corresponds to attacking the regime and $1$ corresponds to not attacking; attacking is associated with a per-unit flow cost of $c > 0$ and the `instantaneous' failure rate of the regime is strictly increasing in the measure attacking and strictly decreasing in the state. Attacking when the regime fails yields a fixed lump sum payoff and the failure of the regime is publicly observed.\footnote{Moreover, our designer's payoff can reflect incentives to delay the regime change, even if it is inevitable i.e., the regime eventually fails almost surely, no matter aggregate play.}
    \item \emph{Stopping games.} Action $1$ corresponds to irreversibly stopping and adopting the action (e.g., investing, attacking, protesting, etc.). Flow payoffs are increasing in the measure of people who have stopped. Action $0$ corresponds to not stopping. 
\end{enumerate}

\paragraph{Dynamic information policies} A history $H_t := \big((\mu_{s})_{s \leq t}, (A_s)_{s \leq t}\big)$ specifies beliefs and aggregate play up to time $t$. Let $\mathcal{H}_t$ be the set of all histories and $\mathcal{H} := \bigcup_{t \geq 0} \mathcal{H}_t$. Write $\mathcal{F}_t$ as the natural filtration generated by histories. A dynamic information policy is a c\`adl\`ag $(\mathcal{F}_t)_t$-martingale. Let 
\[\mathcal{M} := \Big\{\bm{\mu}': \bm{\mu}' \text{ is a c\`adl\`ag $(\mathcal{F}_t)_t$-martingale, $\mu_0 = \mu'_0$ a.s.} \Big\}.
\]
be the set of all dynamic information policies. 

\paragraph{Strategies and Equilibria} A strategy $\sigma_i : \mathcal{H} \to \Delta\{0,1\}$ is a map from histories to a distribution over actions so that if $i$'s clock ticks at time $t$, her choice of action is given by history $H_{t-} := \lim_{t' \uparrow t}H_{t'}$.\footnote{This is well-defined since $(A_t)_t$ is a.s. continuous and $(\mu_t)_t$ has left-limits. Since the measure of agents who act at time $t$ is zero, our game is in effect equivalent to one in which play at time $t$ depends on history $H_t$.}
Thus, each information policy $\bm{\mu}$ induces a stochastic game;\footnote{Note that information is public so all agents share the same beliefs; in Online Appendix \ref{appendix:private} we relax this to show that private information often cannot do better.} let $\Sigma(\bm{\mu},A_0)$ denote the set of subgame perfect equilibria of the stochastic game. We focus on subgame perfection because there is no private information so the game continuing from each history corresponds to a proper subgame.



\begin{figure}[h!]  
\centering
\captionsetup{width=0.9\linewidth}
    \caption{Relationship between beliefs, equilibria, and action paths} \includegraphics[width=0.7\textwidth]{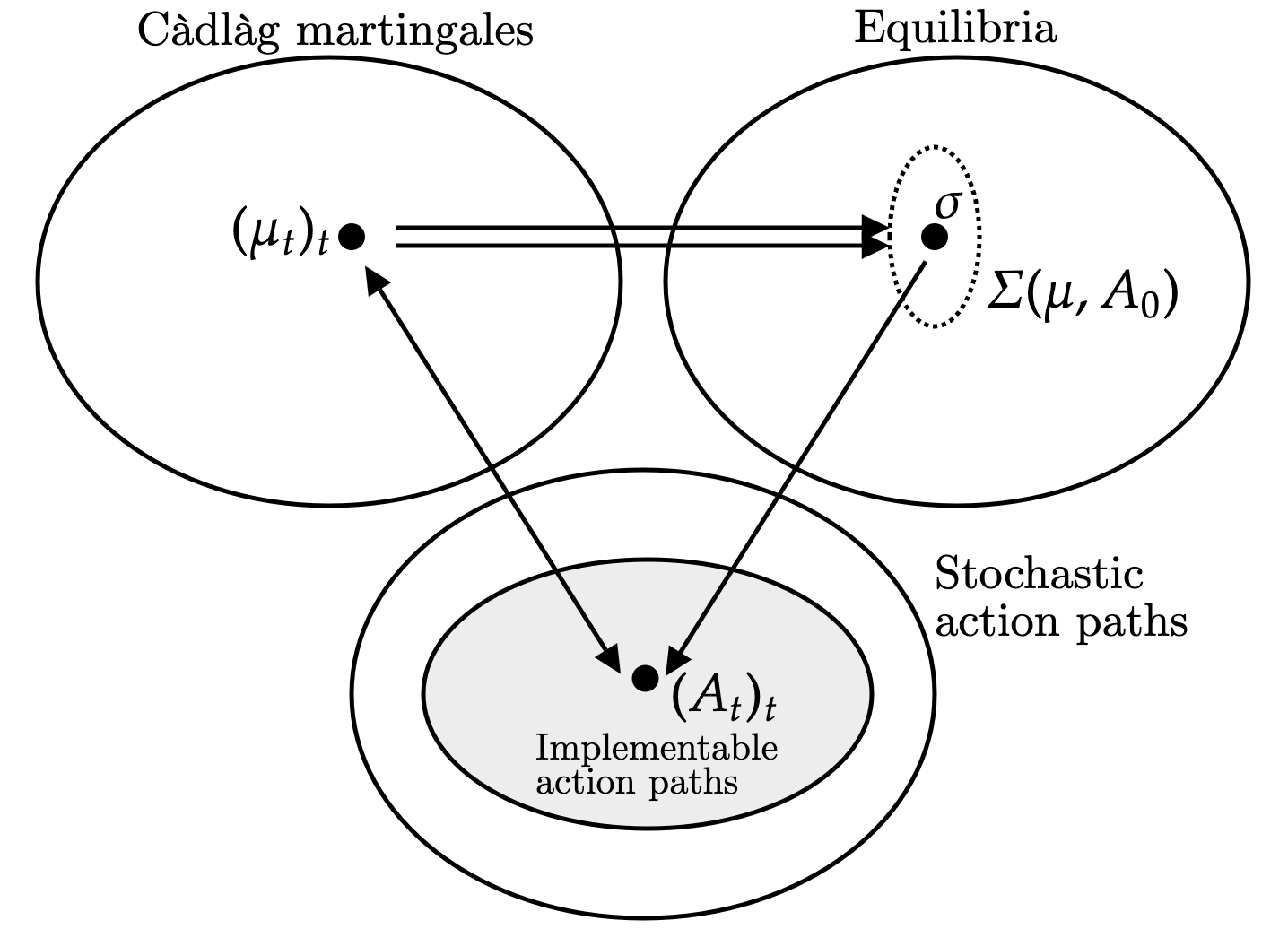}
    \label{fig:model_illust}
\end{figure}

Figure \ref{fig:model_illust} illustrates the connection between dynamic information policies (top left), equilibria (top right), and the path of aggregate actions (bottom). Each information policy $(\mu_t)_t$ specifies a c\`adl\`ag martingale which induces a set of equilibria $\Sigma(\bm{\mu},A_0)$. The realizations of beliefs $(\mu_t)_t$ as well as the selected equilibrium $\sigma \in \Sigma(\bm{\mu},A_0)$ jointly induce a stochastic path of aggregate play $(A_t)_t$. The designer's problem, then, is to choose its dynamic information policy to influence the set of equilbria and thus $(A_t)_t$. 

\paragraph{Designer's problem under adversarial equilibrium selection}  
The designer's problem under commitment when nature is choosing the best equilibrium is 
\[
\sup_{\substack{\bm{\mu} \in \mathcal{M}\\ \bm{\sigma} \in \Sigma(\bm{\mu},A_0)}} \Ex^{\sigma} \Big[\phi\big(\bm{A}\big)\Big]. \tag{OPT} \label{eqn:opt}
\]
Conversely, when nature is choosing the worst equilibrium, the problem is 
\[
\sup_{\bm{\mu} \in \mathcal{M}} \inf_{\bm{\sigma} \in \Sigma(\bm{\mu},A_0)} \Ex^{\sigma} \Big[\phi\big(\bm{A}\big)\Big] \tag{ADV} \label{eqn:adv}
\]
where $\phi: \mathcal{A} \to \mathbb{R}$ is an increasing and bounded functional from the path-space of aggregate play $\mathcal{A}$, and we write $\bm{A} := (A_t)_t \in \mathcal{A}$. For instance, this could be the discounted measure of play i.e., $\phi(\bm{A}) = \int e^{-rt}A_t dt$ with discount rate $r > 0$.

\paragraph{Sequential Optimality} If the designer cannot commit to future information, off-path delivery of information might have no bite in the present. To this end, define the payoff gap at history $H_t$ as the value of the best deviation from the original policy $\bm{\mu}$: 
\[
\inf_{\bm{\sigma} \in \Sigma(\bm{\mu},A_0)} \Ex^{\sigma} \Big[\phi\big(\bm{A}\big) \Big| \mathcal{F}_t\Big]  -  
\sup_{\bm{\mu}' \in \mathcal{M}} \inf_{\bm{\sigma} \in \Sigma(\bm{\mu}',A_0)} \Ex^{\sigma} \Big[\phi(\bm{A})\Big| \mathcal{F}_t\Big] \geq 0 
\]
where $\mathcal{F}_t$ is the filtration corresponding to $H_t$. $\bm{\mu}$ is sequentially optimal if the gap is zero for all histories $H_t \in \mathcal{H}$. Sequential optimality is quite demanding and states that at every history---including off-path ones---the designer still finds it optimal to follow through with her dynamic information policy. 

\section{Optimal dynamic information} \label{sec:mainresult}

We begin with an intuitive description of a sequentially-optimal dynamic information policy for binary states before constructing it formally. With binary states, we set $\Theta = \{0,1\}$ where $1$ is the dominant state on which it is strictly dominant to play action $1$. Since beliefs are one-dimensional, with slight abuse of notation we will write $\mu_t := \mathbb{P}(\theta = 1|\mathcal{F}_t)$. Let $\psi_{UD}(A)$ be the lowest belief such that, if the current aggregate play is $A$, playing action $1$ is strictly dominant. We call the set of belief and aggregate action pairs $(\mu,A)$ such that $\mu \geq \psi_{LD}(A)$ the \emph{upper dominance region}. 

\paragraph{I. Belief and aggregate action near the upper dominance region.} First suppose that at time $t$, the public belief $\mu_t$ and aggregate play $A_t$ are close to the upper dominance region as illustrated by the blue dot labeled $(\mu_t,A_t)$ in Figure \ref{fig:intuition_1} (a). If players switch to action $1$, the designer stays silent. Then, aggregate action progressively increases as illustrated by the upward arrows in Figure \ref{fig:intuition_1} (a).

\begin{figure}[h]
\centering
\caption{Policy near upper dominance region}
    \subfloat[Silence on-path]{\includegraphics[width=0.5\textwidth]{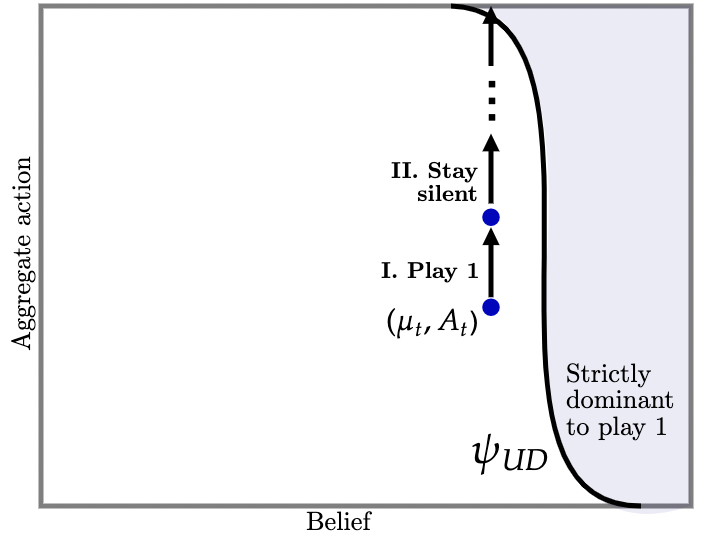}}
    \subfloat[Injection into dominance region]{\includegraphics[width=0.5\textwidth]{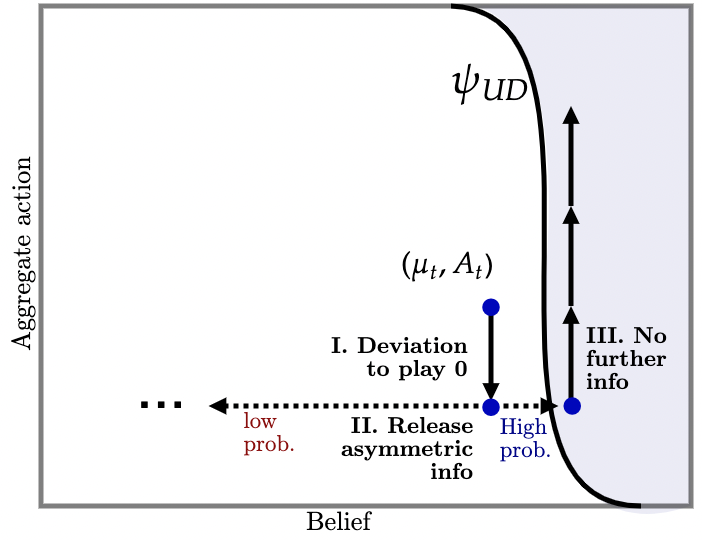}}
    \label{fig:intuition_1}
\end{figure}

Now suppose, instead, that players start playing action $0$ as depicted in Figure \ref{fig:intuition_1} (b) I. Then, the designer injects \emph{asymmetric} information such that it is very likely that agents become slightly more optimistic i.e., public beliefs move up a little and into the upper dominance region, but there is a small chance agents become much more pessimistic (Fig. \ref{fig:intuition_1} (b) II). Suppose that this deviation happened and so this information is injected and, furthermore, that it has made agents a little more optimistic. Then, on this event, future beliefs are in the upper dominance region so it is strictly dominant for future agents to take action $1$. Correspondingly, the designer delivers no further information (Fig. \ref{fig:intuition_1} (b) III) and aggregate play begins to increase thereafter. But, knowing that this sequence of events is likely to take place, and because agents have coordination motives, playing action $0$ in the first place is strictly dominated. 

\paragraph{II. Belief and aggregate action far from upper dominance region.} Next consider Figure \ref{fig:intuition_2} (a) where belief $\mu_t$ and aggregate play $A_t$ are far from the upper dominance region i.e., $\mu_t$ is far below $\psi_{LD}(A_t)$. Our previous argument now breaks down: there is no way for off-path information---no matter how cleverly designed---to ensure beliefs reach the dominance region with a high enough probability as to deter the initial deviation to action $0$. This is the key weakness of off-path information vis-a-vis off-path transfers. What then does the designer do?

\begin{figure}[h]
\centering
\caption{Chaining off-path information}
    \subfloat[Chaining]{\includegraphics[width=0.5\textwidth]{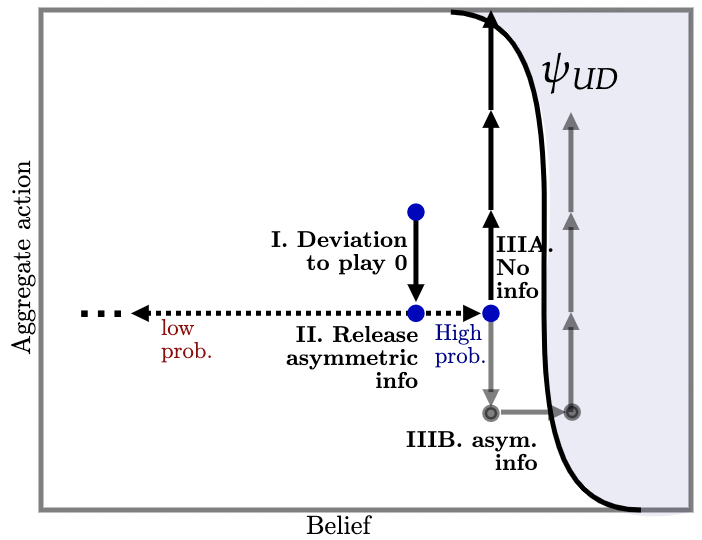}}
    \subfloat[Contagion to lower dominance]{\includegraphics[width=0.5\textwidth]{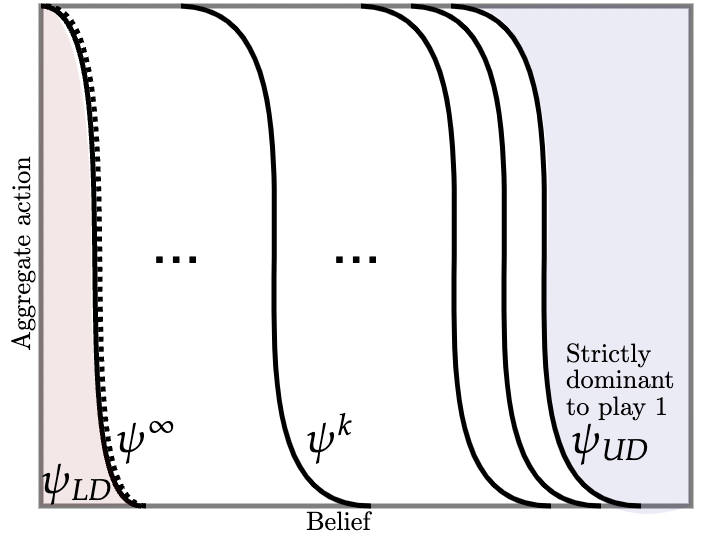}}
    \label{fig:intuition_2}
\end{figure}

If players start playing action $0$, the designer delivers asymmetric information so that, with high probability, agents become a little more confident---but, importantly, not confident enough that action $1$ is strictly dominant. This is depicted in Figure \ref{fig:intuition_2} (a) II. Upon this realization, if future agents continue deviating to $0$, the policy injects yet another bout of asymmetric information which, with high probability, pushes beliefs into the upper dominance region. This is depicted in Figure \ref{fig:intuition_2} (a) IIIB. Knowing this, we have already seen that those future agents strictly prefer to switch to $1$. But knowing that, agents in the present state $(\mu_t,A_t)$, anticipating that upon deviation the injection will, with high probability, induce future agents to play $1$, also strictly prefer to play $1$ in the present. 

What are the limits of this line of reasoning? It turns out that, by choosing our dynamic information policy carefully, we can chain together these injections of off-path information in such a way as to obtain full implementation at \emph{all} belief-aggregate pairs for which action $1$ is not strictly dominated. This is depicted by Figure \ref{fig:intuition_2} (b) where, as before, the blue region represents the upper dominance region, and the pink region represents the lower dominance region at which action $0$ is strictly dominant. The logic is related to the ``contagion arguments" of \cite*{frankel2000resolving,burdzy2001fast,frankel2003equilibrium}. These papers show that the risk-dominant action is typically selected as the limit of some iterated deletion procedure in which \emph{both} the upper dominance region (blue region) and lower dominance region (pink) expand with each iteration and meet in the middle which pins down the unique equilibrium.\footnote{In \cite*{frankel2000resolving,burdzy2001fast} this is also obtained via backward induction, where a symmetric random process governs aggregate incentives. Mapped to our model, this corresponds to public information so that the belief martingale is a time-changed Brownian motion. In \cite*{frankel2003equilibrium}, this is obtained via interim deletion of strictly dominated strategies in many-action global games, though the logic is similar.} 
By contrast, we show how dynamic information can be employed to generate \emph{asymmetric contagion} such that only the upper dominance region expands to engulf the space of \emph{all} belief-aggregate play pairs where action $1$ is not strictly dominated. 

\paragraph{III. Designer-preferred action strictly dominated.} Now suppose beliefs are so pessimistic that $1$ is strictly dominated i.e., $\mu_t \leq \psi_{LD} (A_t)$ where $\psi_{LD}(A_t)$ is the highest belief under which, given aggregtae play $A_t$, action $1$ is strictly dominated. Then, the above policy no longer works: even if players expect all future players to play $1$, they are so pessimistic about the state that playing $0$ is strictly better.  Hence, the designer must offer non-trivial information on-path to push beliefs out of the lower dominance region. How is this optimally done? 

\begin{figure}[h]
\centering
\caption{Escaping the lower dominance region}
    \subfloat[Immediate injection]{\includegraphics[width=0.25\textwidth]{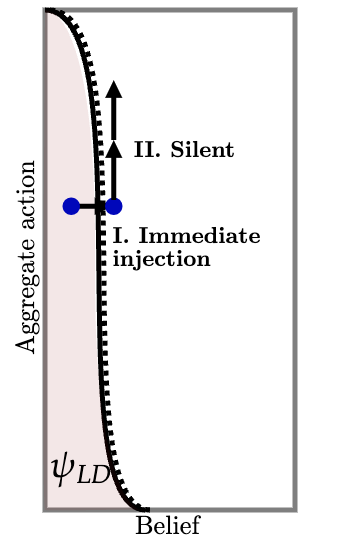}}
    \subfloat[Delayed injection]{\includegraphics[width=0.25\textwidth]{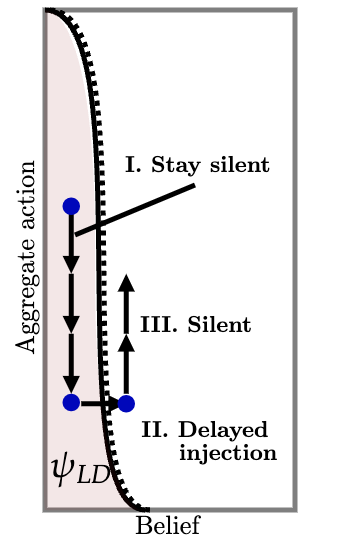}}
    \subfloat[`Smooth' injection]{\includegraphics[width=0.25\textwidth]{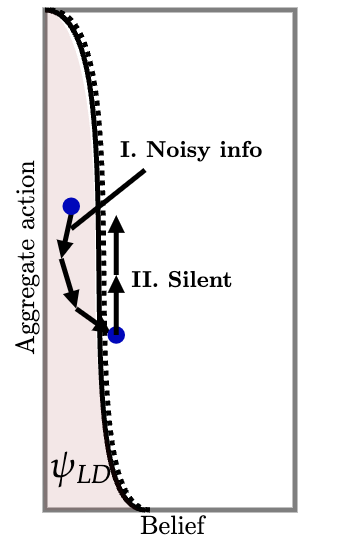}}
    \label{fig:intuition_3}
\end{figure}

Figure \ref{fig:intuition_3} (a) illustrates the optimal policy which consists of an \emph{immediate} and \emph{precise} injection of information such that beliefs jump to either $0$ or (just) out of the lower dominance region. The optimality of such a policy is built on the observation that if the designer does not intervene early to curtail players from progressively switching to $0$, it simply becomes more difficult to escape the lower dominance region down the line. Consider, for instance, the policy in Figure \ref{fig:intuition_3} (b) which also injects precise information to maximize the chance of escaping the lower dominance region, but with a delay. Before this injection, players switch to action $0$ and since $\psi_{LD}(A)$ is strictly decreasing, the probability of escaping the dominance region is strictly smaller. For similar reasons, the policy illustrated in Figure \ref{fig:intuition_3} (c) which induces continuous sample belief paths is also suboptimal.

\paragraph{IV. Sequential optimality.} Our previous discussion specified off-path injections of policies upon deviation away from the action $1$. Of course, if such deviations actually occur, the designer may not have any incentive to follow-through with its policy. For instance, consider Figure \ref{fig:intuition_4} (a) which employs the strategy of injecting `conclusive bad news' that the state is $0$ so that, with high probability beliefs increase a little, and with low probability agents learn conclusively that $\theta = 0$. Information of this form maximizes the chance that beliefs increase\footnote{As in \cite{kamenica2011bayesian} and subsequent work.} and, as we have described, these can be be chained together to achieve full implementation. However, this policy is not sequentially optimal: if agents do deviate and play action $0$, injecting such information is suboptimal because it poses an extra risk: if conclusive bad news does arrive, beliefs become absorbing at $\mu_t = 0$ and further information is powerless to influence beliefs. It then becomes strictly dominant for all agents to play $0$ thereafter. How, then, is sequential optimality obtained?

\begin{figure}[h]
\centering
\caption{Sequential optimality}
    \subfloat[Not sequentially optimal]{\includegraphics[width=0.5\textwidth]{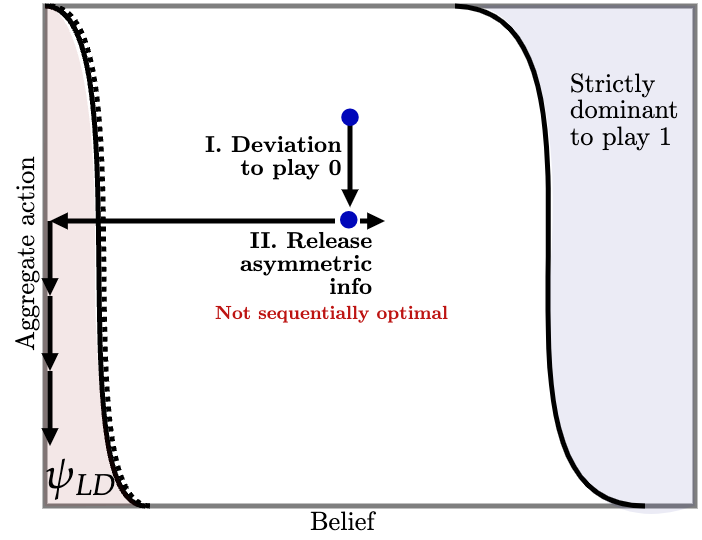}}
    \subfloat[Sequentially optimal]{\includegraphics[width=0.5\textwidth]{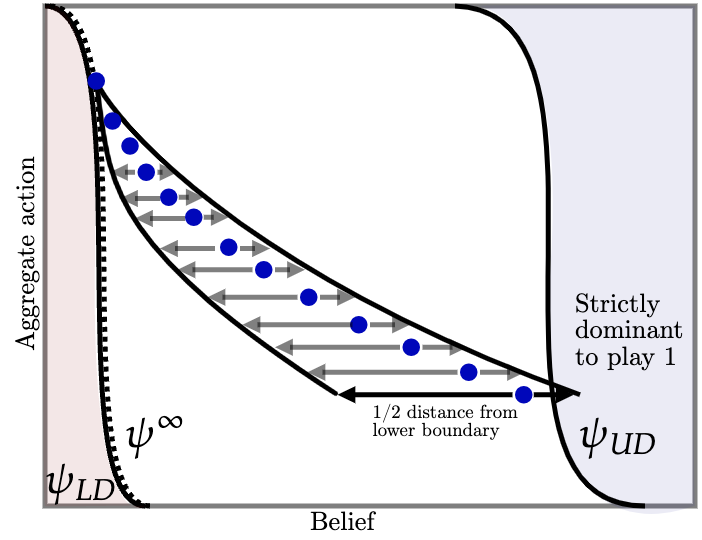}}
    \label{fig:intuition_4}
\end{figure}

Consider, instead, \emph{inconclusive} off-path information as illustrated in Figure \ref{fig:intuition_4} (b) where the arrows corresponds to each blue dot represents a potential injection of off-path information upon players' deviating to action $0$. Each injection induces two kinds of beliefs: upon arrival of a `good' signal (right arrow), agents become a little more optimistic; upon arrival of a `bad' signal, agents become \emph{relatively} more pessimistic (left arrow), but not so much that action $1$ becomes strictly dominated. Figure \ref{fig:intuition_4} illustrates a particular policy in which, upon realization of the bad signal at state $(\mu_{t},A_{t})$, agents' beliefs move halfway toward the lower dominance region i.e., to $[\mu_{t} + \psi_{LD}(A_t)]/{2}$. Conversely, if the good signal arrives, believes move up a little, so that the probability of the former is much higher than the latter. 

By choosing this distribution carefully for each belief-aggregate action pair, we can achieve full implementation via the chaining argument outlined above, which requires that (i) probability of the good signal arriving is sufficiently high as to deter deviations; and (ii) movement in beliefs generated by the good signal is sufficiently large that, when chained together, we obtain full implementation over the whole region. At the same time, this is sequentially optimal since, whenever the designer is faced with the prospect of injecting off-path information, she is willing to do so: with probability $1$ agents' posterior beliefs are such that, in the continuation subgame, full implementation of the largest time-path of play remains possible.\footnote{We emphasize that there is nothing circular about this argument: we iteratively delete switching to action $0$ under the \emph{worst-case} conjecture that, upon the bad signal arriving, all future agents play $1$. This is sufficient to obtain full implementation as long as action $1$ is not strictly dominated.} 

Sequential optimality of dynamic information has been recently studied in single-agent optimal stopping problems \citep*{koh2022attention,koh2024persuasion} who show that optimal dynamic information can always be modified to be sequentially optimal.\footnote{See also \cite{ball2023dynamic} who finds in a different single-agent contracting environment that the optimal dynamic information policy happens to be sequentially optimal.} In such environments, sequential optimality is obtained via an entirely distinct mechanism: the designer progressively delivers more interim information to {raise the agent's outside option} at future histories which, in turn, ties the designer's hands in the future. By contrast, in the present environment our designer chains together off-path information together to \emph{raise her own continuation value} by ensuring that, following any realization of the asymmetric signal, her future self can always fully implement the largest path of play.

\paragraph{Construction of informational puts.} 
We now make our previous discussion precise and general.   We will construct a particular family of belief martingales $\bm{\mu} \in \mathcal{M}$ which is `Markovian' in the sense that the `instantaneous' information at time $t$ depends only on the belief-aggregate play pair $(\mu_t,A_t)$, as well as an auxiliary $(\mathcal{F}_t)_t$-predictable process $(Z_t)_t$ we will define as part of the policy. We begin with several key definitions: 

\begin{definition} [Lower dominance region] Let $\Psi_{LD}: [0,1] \rightrightarrows \Delta(\Theta)$ denote the set of beliefs under which players prefer action $0$ even if all future players choose to play action $1$:
\begin{align*}
    \Psi_{LD}(A_t) \coloneqq  \Big\{ \mu \in \Delta(\Theta): \Ex_{\theta \sim \mu} \Big[ \int_{t}^{t + \tau} e^{-rs} \Delta u (\bar{A}_s,\theta) ds \Big] \leq 0 \Big\},
\end{align*}
where $\bar{A}_s$ solves $d \bar{A}_s = \lambda (1- \bar A_s) ds$ for $s \geq t$ with boundary $\bar{A}_t = A_t$ and $\tau$ is independently distributed according to an exponential distribution with rate $\lambda$.
\end{definition}

\begin{figure}[h]
\centering
\caption{Illustration of $\Psi_{LD}$, $\text{Bd}_{\theta^*}$, and $D$}
    \subfloat[$|\Theta| = 2$]{\includegraphics[width=0.45\textwidth]{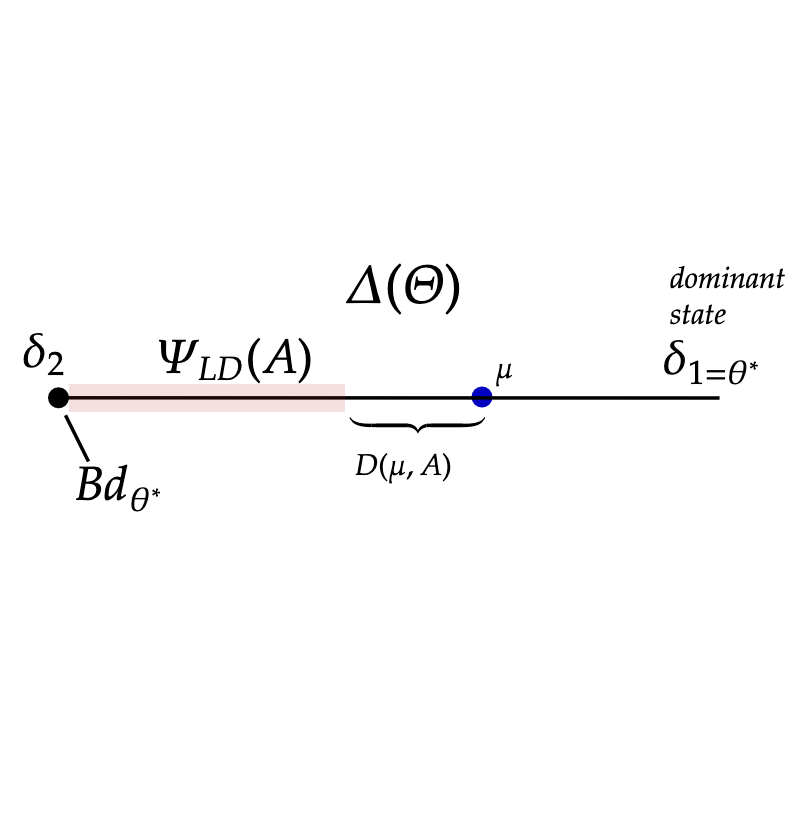}}
    \subfloat[$|\Theta| = 3$]{\includegraphics[width=0.45\textwidth]{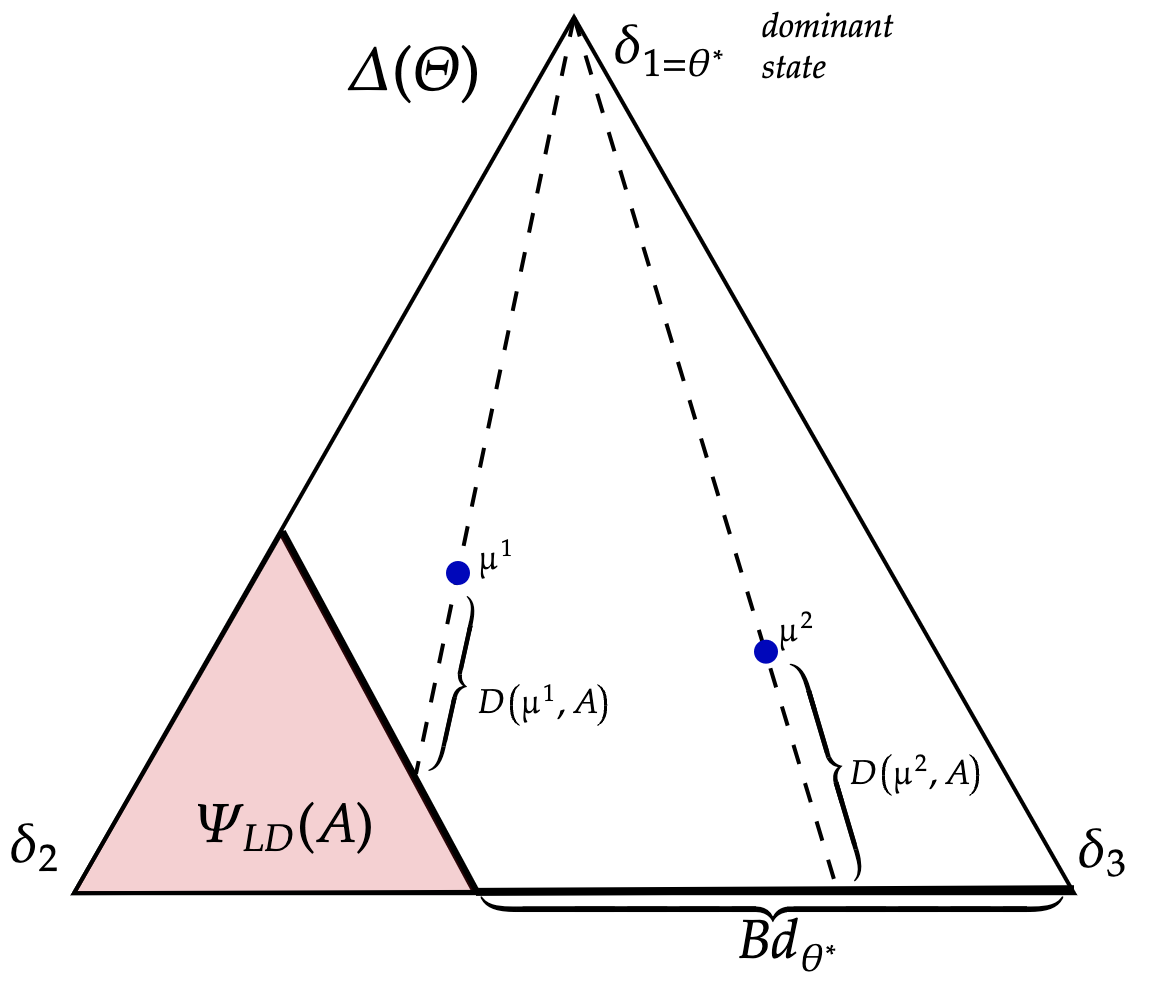}}
    \label{fig:3state_illust}
\end{figure}

\vspace{-1em}
Observe that supermodularity implies $\Psi_{LD}$ is decreasing in $A_t$: $\Psi_{LD}(A_t) \subseteq \Psi_{LD}(A'_t)$ if $A_t \geq A'_t$. 
$\Psi_{LD}$ is illustrated by the pink region of Figure \ref{fig:3state_illust} for the cases where $|\Theta| = 2$ (panel (a)) and $|\Theta| = 3$ (panel (b)).

\begin{definition}
For each $\mu_t \notin \Psi_{LD}(A_t)$ and $A_t \in [0,1]$, define
\[D(\mu_t,A_t) \coloneqq \inf\bigg\{ \alpha \in [0,1] :  \mu_t - \alpha \cdot \frac{\delta_{\theta^*} - \mu_t}{1 - \mu_t(\theta^*)} \in \Psi_{LD}(A_t) \cup \text{Bd}_{\theta^*}  \bigg\}.\]
\end{definition}
\vspace{-1em}
This gives the `distance' from current beliefs $\mu_t$ as it moves along a linear path starting from $\delta_{\theta^*}$ to either (i) the lower dominance region $ \Psi_{LD}(A_t)$; or (ii) the set of beliefs that assign zero probability on state $\theta^*$ which we denote with $\text{Bd}_{\theta^*} \coloneqq \{ \mu \in \Delta(\Theta) : \mu(\theta^*) = 0\}.$ 
This is depicted in Figure \ref{fig:3state_illust}; each blue dot represents a belief.


We next define several variables to describe our policy when beliefs are outside of the lower dominance region. 

\begin{definition}[Tolerance, upward/downward jump sizes, belief direction] \label{defn:outsideLD} Define: 
\begin{figure}[H]
\begin{minipage}[t]{0.5\linewidth}
\begin{itemize}[leftmargin = 2em]
    \item[(i)] \textbf{Tolerance.} $\mathsf{TOL}(D)$ specifies the magnitude of deviation of off-path play vis-a-vis a target $Z_t$. If this is exceeded, the policy injects additional information. 
    \item[(ii)] \textbf{Upward jump size.}  $M\cdot \mathsf{TOL}(D)$ scales the tolerance by $M > 0$. This specifies the upward movement in beliefs if the injected information is positive. 
    \item[(iii)] \textbf{Downward jump size.} $\mathsf{DOWN}(D)$ specifies the downward movement in beliefs if the injected information is negative. 
    \item[(iv)] \textbf{Belief direction.} $\hat{\bm{d}}(\mu) \in \mathbb{R}^{n}$ specifies the direction of belief movements. We set it as the directional vector of $\mu$ towards $\delta_{\theta^*}$.
\end{itemize}
\end{minipage}%
\hfill%
\begin{minipage}[t]{0.47\textwidth}\vspace{0pt}
\centering 
\vspace{0em}
{\includegraphics[width=\textwidth]{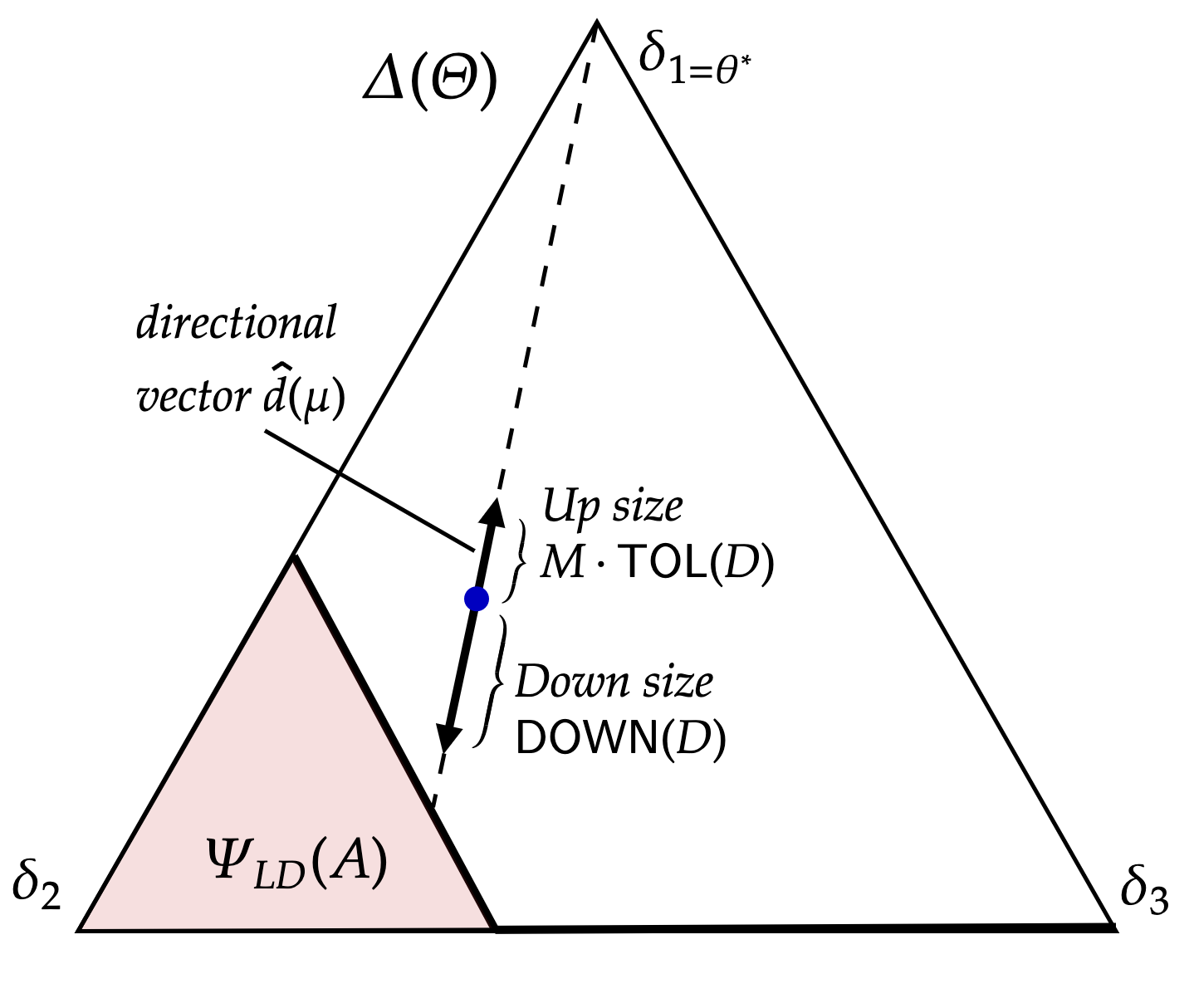}}
\caption{Illustration of up/downward jump sizes and belief directions}
\vspace{-2em}
\end{minipage}
\end{figure}
\end{definition}

Definition \ref{defn:outsideLD} specifies variables to describe our information policy when beliefs lie outside of the lower dominance region $\mu \notin \Psi_{LD}$. We now define variables to describe our information policy when beliefs lie inside the lower dominance region  $\mu \in \Psi_{LD}$. 
\clearpage 

\begin{definition}[Maximal escape probability and beliefs] \label{defn:insideLD} Define: \begin{figure}[H]
\begin{minipage}[t]{0.5\linewidth}
\begin{itemize}[leftmargin = 2em]
    \item[(i)] The set of beliefs which are attainable with probability $p$ from $\mu$ is
    \[F(p,\mu) \coloneqq \{\mu' \in \Delta(\Theta): p\mu' \leq \mu \}\]
    which follows from the martingale property of beliefs. 
    \item[(ii)] \textbf{Maximal escape probability.} $p^*(\mu,A) \coloneqq \max \{p \in [0,1]: F(p,\mu) \subseteq \Psi_{LD}(A)\}$ is a tight upper-bound on the probability that beliefs escape $\Psi_{LD}$. 
    \item[(iii)] The \textbf{maximal escape beliefs} are 
    \vspace{-1em}
    \[\partial({\eta},\mu):= F(p^*(\mu,A)-\eta,\mu) \cap \Psi_{LD}^c(A)\]
    \vspace{0em} 
    where $\Psi^c_{LD}(A) = \Delta(\Theta) \setminus \Psi_{LD}(A)$.
\end{itemize}
\end{minipage}%
\hfill%
\begin{minipage}[t]{0.48\textwidth}\vspace{0pt}
\centering 
\vspace{0em}
{\includegraphics[width=\textwidth]{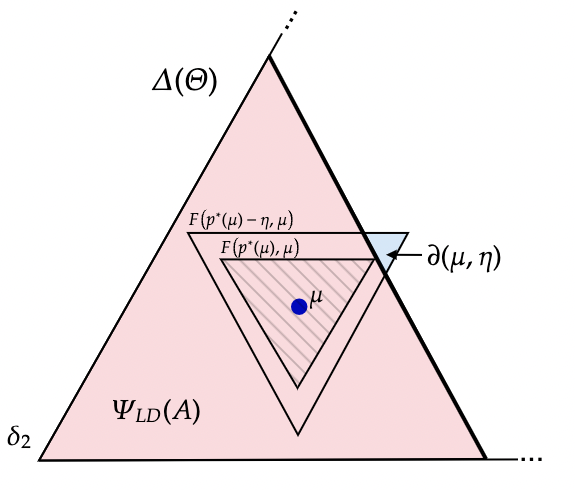}}
\caption{Illustration of $F(p,\mu)$, maximal escape probability, and maximal escape beliefs}
\vspace{-2em}
\end{minipage}
\end{figure}
\end{definition}

We are (finally!) ready to define our family of dynamic information policies. Recall that $\bm{\mu}$ is c\`adl\`ag so has left-limits which we denote with $ \mu_{t-} := \lim_{t' \uparrow t} \mu_t$. We will simultaneously specify the law of $\bm{\mu}$ as well as construct the stochastic process $(Z_t)_t$ which is $(\mathcal{F}_t)_t$-predictable\footnote{That is, $Z_t$ is measurable with respect to the left filtration $\lim_{s \uparrow t} \mathcal{F}_s$.} and initializing $Z_0 = A_0$. $(Z_t)_t$ is interpreted as the \emph{target aggregate play} at each history. Given the tuple $(\mu_{t-},Z_{t-},A_t)$,  define the time-$t$ information structure and law of motion of $Z_t$ as follows:
\begin{enumerate}
    \item \textbf{Silence on-path.} If action $1$ is not strictly dominated i.e., $\mu_{t-} \notin \Psi_{LD}(A_t)$ and play is within the tolerance level i.e., $|A_t - Z_{t-}| < \mathsf{TOL}(D)$ then 
    \begin{center}
        $\mu_t = \mu_{t-}$ almost surely,
    \end{center}
    i.e., no information, and $dZ_t = \lambda(1-Z_{t-}).$
    \item \textbf{Asymmetric and inconclusive off-path injection.} 
    If action $1$ is not strictly dominated i.e., $\mu_{t-} \notin \Psi_{LD}(A_t)$ and play is outside the tolerance level i.e., $|A_t - Z_{t-}| \geq \mathsf{TOL}(D)$ then
    \begin{align*}
        \mu_t = \begin{cases}
            \mu_{t-} + (M \cdot \mathsf{TOL}(D))\cdot \hat{d}(\mu_{t-}) &\text{w.p. $\frac{\mathsf{DOWN}(D)}{\mathsf{DOWN}(D)+ M\cdot\mathsf{TOL}(D)}$} \\
            \mu_{t-} -  \mathsf{DOWN}(D) \cdot \hat{d}(\mu_{t-}) &\text{w.p. $\frac{M \cdot \mathsf{TOL}(D)}{\mathsf{DOWN}(D)+ M\cdot\mathsf{TOL}(D)}$},
        \end{cases}
    \end{align*}
    where $\hat{d}(\mu) := \frac{\delta_{\theta^*} - \mu}{ 1 - \mu(\theta^*) }$ is the (normalized) directional vector of $\mu$ toward $\delta_{\theta^*}$, and reset $Z_t = A_t$.
    \item \textbf{Jump.} If action $1$ is strictly dominated i.e., $\mu_{t-} \in \Psi_{LD}(A_t)$ then beliefs jump to a maximal escape point: pick any $\mu^+ \in \partial(\mu,\eta)$ and 
    \begin{align*}
        \mu_t = \begin{cases}
            \mu^+ &\text{w.p. $p^*(\mu_{t-},A_t) - \eta$} \\
            \mu^- &\text{w.p. $1-(p^*(\mu_{t-},A_t) - \eta)$.\footnotemark}
        \end{cases}
    \end{align*}
    \footnotetext{From the martingale property of beliefs, this can be computed as $\mu^- = \frac{\mu_t - (p^*(\mu_{t-},A_t) - \eta) \mu^+}{ 1-(p^*(\mu_{t-},A_t) - \eta)}.$}
\end{enumerate}


We have defined a family of information structures parametrized by $\mathsf{TOL}(D)$ (tolerance), $M\cdot \mathsf{TOL}(D)$ (upward jump size), $\mathsf{DOWN}(D)$ (downward jump size), and $\eta$ (distance beliefs jump outside the lower dominance region $\Psi_{LD}$). There is some flexibility to choose them: we will set $\mathsf{DOWN}(D) = \frac{1}{2}D$, $\mathsf{TOL}(D) = m \cdot D^2$ where $m > 0$ is a small constant, $M > 0$ is a large constant. 

We choose $m$ small so that $\mathsf{TOL}(D)$, the upward jump size, is much smaller than the downward jump size---this ensures that the probability of becoming (a little) more optimistic is much larger. $M$ is the ratio between the upward jump size and the tolerance---it is large to guarantee that off-path information can push future beliefs into the upper dominance region. The exact choice of $m$ and $M$ will depend on the primitives of the game, but are independent of $\eta$ (how far outside the lower dominance region beliefs jump); a detailed construction is in Appendix \ref{appendix:proofs}. Parameterize this family of policies by $(\bm{\mu}^{\eta})_{\eta > 0}$ which we call \emph{informational puts}.  

\begin{theorem} \label{thrm:main} In the limit $\eta \downarrow 0$, informational puts close the multiplicity gap and is sequentially optimal: 
\begin{itemize}
    \item[(i)] \textbf{Form and value.} 
    \[
    \lim_{\eta \downarrow 0} \inf_{\sigma \in \Sigma(\bm{\mu}^{\eta}, A_0)} \Ex^{\sigma}\Big[\phi(\bm{A})\Big] = \eqref{eqn:adv} = \eqref{eqn:opt}.\]
    \item[(ii)] \textbf{Sequential optimality.} 
    \[
    \lim_{\eta \downarrow 0} \sup_{H_t \in \mathcal{H}} \Bigg| \inf_{\bm{\sigma} \in \Sigma(\bm{\mu}^{\eta},A_0)} \Ex^{\sigma} \Big[\phi\big(\bm{A}\big) \Big| \mathcal{F}_t\Big]  - 
\sup_{\bm{\mu}' \in \mathcal{M}} \inf_{\bm{\sigma} \in \Sigma(\bm{\mu}',A_0)} \Ex^{\sigma} \Big[\phi(\bm{A})\Big| \mathcal{F}_t\Big] \Bigg| = 0.
    \]
\end{itemize}
\end{theorem}

\begin{proof}
    See Appendix \ref{appendix:proofs}. 
\end{proof}

\section{Robustness and generalizations} \label{sec:robustness}
Our dynamic game is quite general in some regards but more specific in others. We now discuss which aspects are crucial, and which can be relaxed. 

\paragraph{Continuum vs finite players} We worked with a continuum of players so there is no aggregate randomness in the time path of agents who can re-optimize their action. This delivers a cleaner analysis since the only source of randomness is fluctuations in beliefs driven by the policy. In Online Appendix \ref{appendix:finite_players} we show that an analog of Theorem \ref{thrm:main} holds in a model with a large but finite number of players and information is---as in the main text---only conditioned on aggregate rather than individual play.\footnote{See \cite{aumann1966} and \cite{fudenberg1986limit,levine1995agents} for a discussion of the subtleties between continuum and finite players.} In the finite player version of our model, the \emph{same} informational puts policy that was optimal in the continuum case continues to solve the adversarial problem and close the multiplicity gap for large but finite number of players whenever initial beliefs are such that action $1$ is not strictly dominated; see Theorem \ref{thrm: finite} in Online Appendix \ref{appendix:finite_players}.\footnote{In particular, informational puts closes the multiplicity gap between adversarial and designer-optimal selection in the finite case at rate $O(N^{-1/9})$ where $N$ is the number of players. The case where the prior is inside the lower dominance region is more subtle and is also analyzed in Online Appendix \ref{appendix:finite_players}.} 

\paragraph{On-path vs off-path information.} An artifact of working with a continuum is that there is no aggregate randomness in the ticks of players' clocks. Our finite player model reintroduces aggregate randomness which can lead aggregate play to---even if all players switch to $1$ as soon as they can---fall short of target play.\footnote{This was previously ruled out by the continuum law of large numbers.} Thus, in the finite player model, our policy now injects information with positive probability (see Online Appendix \ref{appendix:finite_players}) but continues to achieve full implementation. This clarifies that the effectiveness of offering informational puts do not hinge on the delivery of information being zero-probability \emph{per se}. 


\paragraph{Public vs private information.} When the initial condition is such that playing $1$ is not strictly dominated, our policy fully implements the upper-bound on the time path of aggregate play. Thus, private information clearly cannot do better. If initial beliefs are such that the designer's preferred action is strictly dominated, however, this is more subtle.\footnote{It is still an open question as to how to characterize feasible joint paths of higher-order beliefs.} We analyze this in Online Appendix \ref{appendix:private}.

\paragraph{Homogeneous vs heterogeneous players.} Payoffs in our environment are quite general, with the caveat that they were identical across players. It is well-known that introducing heterogeneity typically aids equilibrium uniqueness in supermodular games. Since we have already closed the multiplicity gap under homogeneous payoffs,\footnote{At least when initial beliefs are such that action $1$  is not strictly dominated.} we expect that heterogeneity will make full implementation easier. 

\paragraph{Switching frictions.} Switching frictions are commonly used to model switching costs, inattention, or settings with some staggered structure. They are important in our environment because our policy can then inject information as soon as players begin deviating from the designer's preferred action. This allows off-path information to be chained together by correcting incipient deviations as soon as they occur. By contrast, if players could continually re-optimize their actions, then off-path information is powerless to rule out equilibria of the form ``all simultaneously switch to $0$".\footnote{Indeed, prior work which obtained equilibrium uniqueness (of risk-dominant selection) \citep*{frankel2000resolving,burdzy2001fast} do so via switching frictions. Switching frictions are prevalent in macroeconomics but, as \cite{angeletos2016incomplete} note, \emph{``It is then
somewhat surprising that this approach} [combining aggregate shocks with switching frictions to generate uniqueness] \emph{has not attracted more attention in applied research.''}} 
We emphasize, however, that it would suffice for \emph{some} frictions to exist, but the exact form is not particularly important: the switching rate could vary with aggregate play,  change over time, and can be taken to be arbitrarily quick or slow.

\paragraph{Full implementation via dynamic cheap talk.} 
We have worked within the framework of dynamic information design. Implicit in this framework are two distinct assumptions: (i) `within period' commitment to a statistical experiment as in static information design; and (ii) `between period' intertemporal commitment to follow through with promised information. Theorem \ref{thrm:main} shows that the latter is unnecessary; what of the former? 

If initial beliefs are such that action $1$ is not strictly dominated, the largest path of play can be implemented via dynamic cheap talk with {no commitment---within or between periods}.\footnote{This is driven by the following sequence of simple observations. First, observe that the value of dynamic cheap talk is upper-bounded by the value of dynamic information design. Next, from Theorem \ref{thrm:main} dynamic information design can always uniquely implement the largest possible time-path of play in a sequentially optimal manner. Thus, we can construct \emph{an} equilibrium of a dynamic cheap talk game in which the policymaker randomizes exactly as she would under dynamic information design. In this equilibrium the policymaker achieves her first-best payoff.}
We highlight two caveats. First, dynamic cheap talk reintroduces a different kind of multiplicity since there will always be a babbling equilibrium. Thus, this applies only to environments in which the policymaker \emph{fears miscoordination} but not \emph{miscommunication}.\footnote{Loosely, this require adversarial equilibrium selection to apply only to the dynamic coordination game among agents.} Second, when the initial condition is such that action $1$ is strictly dominated, then some `within period' commitment is indeed required to push beliefs out of the lower dominance region. 

\paragraph{Designer has full control over information.} We have assumed that the designer is the sole provider of information. This simplified our analysis and allowed us to obtain clean results. Nonetheless, informational puts can remain efficacious in environments where agents have access to alternate sources of information. For instance, in the dynamic regime change game analyzed in Appendix \ref{appendix:examples}, agents observe whether the regime has yet survived and so make additional inferences about the state---thus, the designer is now constrained by the introduction of outside information and cannot choose any belief martingale. Nonetheless, we show that a version of Theorem \ref{thrm:main} obtains: informational puts continues to (i) solve the adversarial problem; and (ii) close the multiplicity gap.

\section{Discussion}\label{section:discussion}
We have shown that dynamic information is a powerful tool for full implementation in general binary-action supermodular games. In doing so, we highlighted key properties of sequentially-optimal dynamic information: asymmetric and inconclusive signals are chained together off-path to achieve full implementation while ensuring the policymaker has no incentive to deviate from her promised information. 
We conclude with a brief discussion of implications.

\paragraph{Information as insurance.}
Our policy of offering an informational put ensures that, if future players lose faith and switch to action $0$, information is appropriately injected to (with high probability) correct the path of aggregate play. That is, our information policy serves as \emph{credible insurance} against strategic uncertainty, rather than as {punishment} for playing the designer's less preferred action.\footnote{Indeed, another key distinction from transfers is that information provision \emph{per se} cannot serve as punishment since its value is non-negative in single-agent settings.}

This distinction is seen most clearly by noting that our policy is effective---both with a continuum and finite players---although the  designer's policy is conditioned only on aggregate play and cannot detect and respond to individual deviations, a more demanding requirement \citep{levine1995agents}. Moreover, our analysis in the continuum case clarifies that players do not need to believe that they can individually influence game for informational puts to work. This is in stark contrast to work on dynamic coordination, durable goods monopolist, or public-good provision where it is crucial that each agent's action makes a non-negligible difference.\footnote{For instance \cite{gale1995dynamic} highlights a gap between a continuum and finite number of players in dynamic coordination games. A similar gap emerges in durable goods monopolist \citep*{fudenberg1985infinite,gul1986foundations,bagnoli1989durable}. See also recent work by \cite{battaglini2024dynamic} in public goods context where the fact that each agent can influence the state (by a little) is important.} 
Our key insight that this is not required: a dynamic policy with a \emph{moving target}---such that information is injected if aggregate play comes up short---can deliver full implementation by insuring players against strategic uncertainty. 

This view of information as insurance links informational puts to a well-studied instrument in the context of bank runs: deposit insurance. Both instruments serve as insurance against strategic uncertainty and can uniquely implement the designer-preferred equilibrium in which nobody runs. 
 But information also differs in substantive ways: it applies to environments where transfers might be infeasible or unnatural, and can be credibly implemented without any `backing'.\footnote{For instance, \cite{dybvig2023nobel} discussing an alternate form of insurance---a discount window for troubled banks---writes: \emph{``...the central bank may suffer from a credibility problem. If people
are unsure whether the central bank will always fund the banks... they [the depositors] may still run on the bank.''}}

\paragraph{Implications for implementation via information.} A recent literature on information and mechanism design finds that in static environments, there is generically a multiplicity gap---the designer can do strictly better under partial rather than full implementation \citep*{morris2024implementation,halac2021rank}.\footnote{See also \cite*{inostroza2023adversarial,li2023global,morris2022joint,halac2024pricing}.} We showed that the careful design of dynamic public information can quite generally close this gap in dynamic coordination environments.\footnote{Moreover, information in our environment is public so higher-order beliefs are degenerate; by contrast, optimal static implementation via information typically requires inducing non-degenerate higher-order beliefs.}

But do our results demand more of players' rationality and common knowledge of rationality? Yes and no. On the one hand, it is well-known that in environments like ours, there is a tight connection between the iterated deletion of \emph{interim} strictly dominated strategies (as in \cite*{frankel2003equilibrium}) and backward induction, which can be viewed as the iterated deletion of \emph{intermporally} strictly dominated strategies (as in \cite*{frankel2000resolving,burdzy2001fast}). In this regard, we do not think our results require "more sophistication" of agents than in static environments. On the other hand, it is also known that common knowledge of rationality is delicate in dynamic games and must continue to hold at off-path histories.\footnote{See \cite{aumann1995backward}. \cite{samet2005counterfactuals} offers an entertaining discussion. This motivates implementation in "initial rationalizability" when the designer has freedom to design the extensive form game \citep*{chen2023getting}.} In this regard, our stronger results are obtained at the price of arguably stronger assumptions on common knowledge of rationality.

\paragraph{Implications for coordination policy.} Our results have simple policy implications. It is often held that to prevent agents from playing undesirable equilbiria, policymakers must deliver substantial on-path information. For instance, Mario Draghi's `whatever it takes' speech is often cited as an instance in which an informational injection led to an equilibrium switch.\footnote{See \cite{morris2019crises} for a recent articulation of this idea in static games in which a large (informational) shock can lead to equilibrium switching.}

Our results offer a different view since no information is delivered on the path of equilibrium play. This is in contrast to the important work of \cite{basak2020diffusing,basak2024panics} who analyze dynamic regime change games with an irreversible `attack' decision. They show that a designer wishing to preserve the regime can do so with appealingly simple policies: `frequent stress tests' \citep{basak2020diffusing} or a `timely disaster alert' \citep{basak2024panics}---both of which deliver substantial information on-path. However, these policies are also quite specific to their environment---as \cite{basak2024panics} write: ``\emph{if the principal’s
objective is to minimize the size of the attack or to delay the attack as much as possible even when
a disaster is inevitable... then the timely disaster alert is not a desirable policy.}''\footnote{\cite{basak2024panics} go on to conjecture that \emph{``In such cases, private disclosure could be useful to ensure miscoordination
and therefore delay the crisis''.} Perhaps surprisingly, Theorem \ref{thrm:main} shows that with switching frictions this is not required: public information suffices.} 

Our policy of offering informational puts as insurance against strategic uncertainty, though perhaps more complicated, works for broader class of coordination environments and richer designer objectives. For instance, when our setting is specialized to a dynamic regime change game, this can incorporate the designer's desire to minimize the attack or delay the regime's failure.\footnote{Our environment and those of \cite{basak2020diffusing,basak2024panics} are non-nested since we have switching frictions. Nonetheless they are in a similar spirit; in Appendix \ref{appendix:examples} we show how our results apply to dynamic regime change games and stopping games.} Moreover, it prescribes qualitatively different kinds of information: when public beliefs are so pessimistic that the designer-preferred action is strictly dominated, an \emph{early} and \emph{precise} injection of on-path information is indeed required; waiting only makes implementation harder. But as long as beliefs are not so pessimistic, no additional information is required on the path of equilibrium play: \emph{silence} backed by the \emph{credible} promise to inject asymmetric and inconclusive information suffices. 

\setstretch{0.8}
\setlength{\bibsep}{0pt}
\small 
\bibliography{WorksCited}

\appendix 

\begin{center}
    \large{\textbf{\scshape{Appendix to Informational Puts}}}
\end{center}

\normalsize 
\setstretch{1.15}

Appendix \ref{appendix:proofs} proves Theorem \ref{thrm:main}. Appendix \ref{appendix:examples} shows our results apply to a wider class of dynamic coordination games.

\titleformat{\section}
		{\bfseries\center\scshape}     
         {Appendix \thesection:}
        {0.5em}
        {}
        []

\appendix

\section{Proofs} \label{appendix:proofs}

\paragraph{Preliminaries.} We use the following notation for the time-path of aggregate actions following from $A_t$: for $s \geq t$, $\bar A_s$ solves 
\[d\bar A_s = \lambda (1 - \bar A_s) \cdot ds \quad \text{with boundary $\bar A_t = A_t$.}
\]
Similarly, for $s \geq t$, $\underline A_s$ solves 
\[d\underline A_s = - \lambda \underline A_s \cdot  ds \quad \text{with boundary $\bar A_t = A_t$.}
\]
In words, $\bar A_s$ and $\underline A_s$ denote future paths of aggregate actions when everyone in the future switches to actions $1$ and $0$ as quickly as possible, respectively. Finally, it will be helpful to define the operator $S: \mathcal{H} \to \Delta(\Theta) \times [0,1]$ mapping histories to the most recent pair of belief and aggregate action, i.e., $S((\mu_s,A_s)_{s \leq t}) := (\mu_t,A_t)$.

\paragraph{Outline of proof.} The proof of Theorem \ref{thrm:main} consists of the following steps, as described in Figure~\ref{fig:map_main}:

\begin{figure}[h!]  
\centering
\captionsetup{width=1.0\linewidth}
    \caption{Roadmap for proof of Theorem \ref{thrm:main}} \includegraphics[width=1\textwidth]{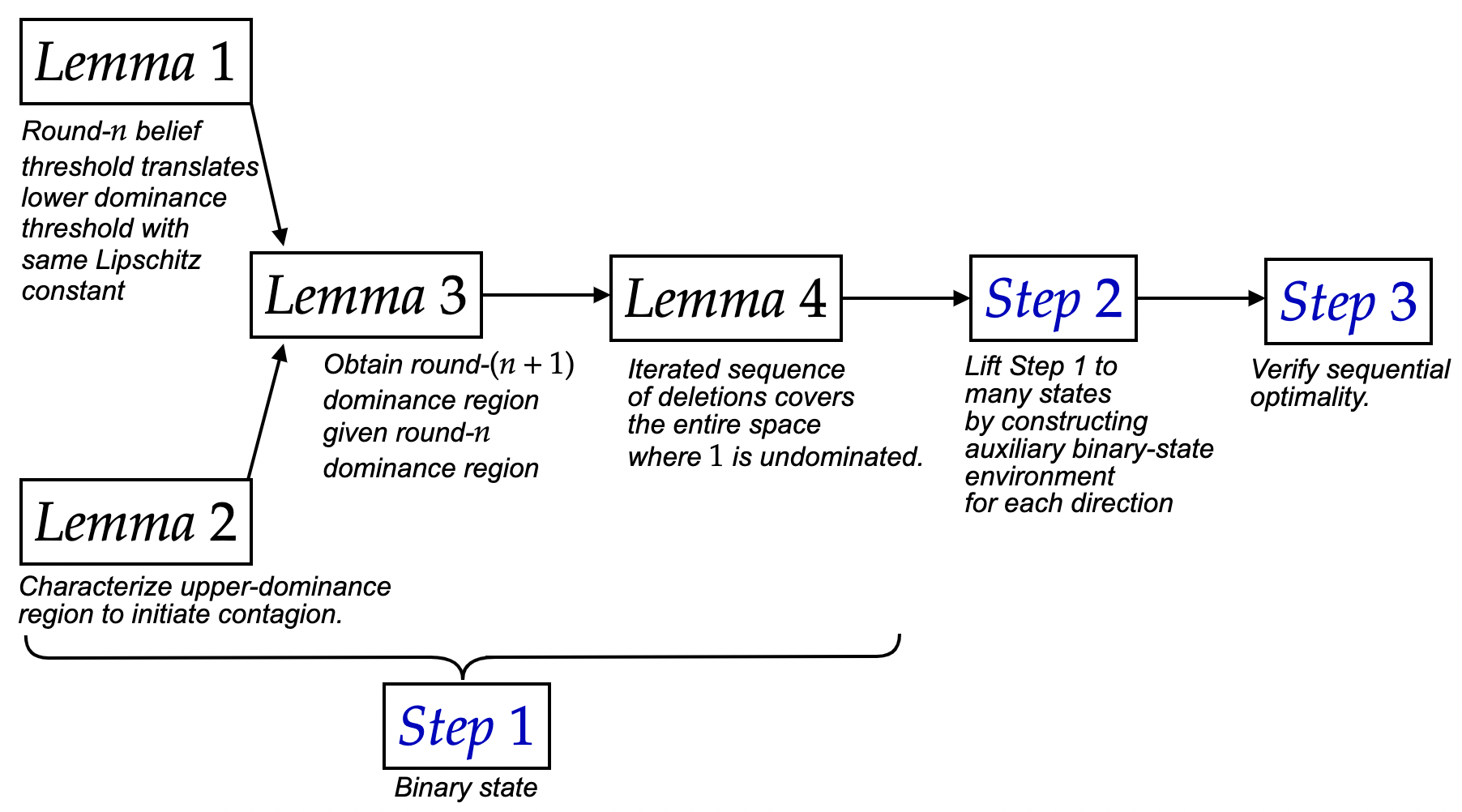}
    \label{fig:map_main}
\end{figure}

\noindent \underline{Step 1}: We first show the result for binary states $\Theta = \{ 0, 1 \}$ with $\theta^* = 1$. With slight abuse of notation, we associate beliefs directly with the probability that the state is $1$: $\mu_t = \mu_t(\theta^*)$. Then, our lower dominance region is one-dimensional and summarized by a threshold belief for each $A$: 
\[ \psi_{LD}(A) := \max_{\mu \in \Psi_{LD}(A)}\mu(\theta^*)\]
We show that $\mu_t > \psi_{LD}(A_t)$ implies switching to $1$ is the unique subgame perfect equilibrium under the information policy $\bm{\mu}^*$. We show this in several sub-steps. 
\begin{itemize}[leftmargin=*]
    \item \underline{Step 1A}: There exists a belief threshold, which is a `rightward' translation of the lower dominance region $\psi_{LD}(A_t)$ such that agents find it strictly dominant to play action $1$ regardless of others' actions if the current belief is above this threshold (Lemma~\ref{lem: dominance region}). We call this threshold $\psi_0(A_t)$.

    \item \underline{Step 1B}: For $n \in \mathbb{N}$, suppose that agents conjecture that all agents will switch to action $1$ at all future histories $H$ such that $S(H) = (\mu_s,A_s)$ fulfills $\mu_s > \psi^n(A_s)$. Under this assumption, we can compute a lower bound \eqref{eq: lower bound} on the expected payoff difference for agents between playing actions $1$ and $0$ for any given current belief $\mu_t \in (\psi_{LD}(A_t), \psi^n(A_t)]$. 
    
    To do so, we will separately consider the future periods before and after the aggregate action deviates from the tolerated distance from the target, at which point new information is provided. Call this time $T^*$. 
    \begin{itemize}
        \item Before $T^*$, we construct the lower bound using the fact that aggregate actions cannot be too far from the target even in the worst-case scenario.

        \item At $T^*$, the designer injects information with binary support. We choose the upward jump size $M \cdot \mathsf{TOL}(D)$ to be sufficiently large so that, whenever the `good signal' realizes beliefs exceed $\psi^n(A_{T^*})$. Whenever the `bad signal' realizes, we conjecture the worst-case that all agents switch to action $0$. 
    \end{itemize} 
    
    \item \underline{Step 1C}: We show that by carefully choosing the information policy, the threshold under which switching to $1$ is strictly dominant, $\psi^{n+1}(A_t)$, is strictly smaller than $\psi^n(A_t)$. 
    The policy has several key features:
    \begin{itemize}
        \item Large $M$: when the aggregate action $A_{T^*}$ falls below the tolerated distance $\mathsf{TOL}(D(\mu_t, A_{T^*}))$ from the target, the high belief after the injection exceeds $\psi^n(A_{T^*})$, which ensures the argument in Step 1B. In particular, we choose $M$ to be large relative to the Lipschitz constant of $\psi^n$. 

        \item Small $\mathsf{TOL}(D(\mu_t, A_{T^*}))$: we should maintain a low tolerance level for deviations from the target. If the designer allowed a large deviation, the aggregate action could drop so low by the time information is injected that agents’ incentives to play action $1$ would be too weak to recover.

        \item Large $\mathsf{DOWN}(D(\mu_t, A_{T^*}))$: the downward jump size should be large relative to the upward jump size $M \cdot \mathsf{TOL}(D(\mu_t, A_{T^*}))$, but not so large that beliefs fall into the lower dominance region. This ensures that the probability of the belief being high after the injection is sufficiently large.

    \end{itemize}
    These three features guarantee that the lower bound \eqref{eq: lower bound} is sufficiently large and remains positive even when the current belief $\mu_t$ is slightly below $\psi^n(A_t)$. Hence $\psi^{n+1}(A_t)$ is strictly smaller than $\psi^n(A_t)$, allowing us to expand the range of beliefs under which action $1$ is uniquely optimal (Lemma~\ref{lem: contagion}).

    \item \underline{Step 1D}: By iterating Step 1C for $n \in \mathbb{N}$, we show that $\psi^n(A_t)$ converges to $\psi_{LD}(A_t)$. Then, if $\mu_t > \psi_{LD}(A_t)$, agents who can switch in period $t$ find it uniquely optimal to choose action $1$.
\end{itemize}

\noindent  \underline{Step 2}: We extend the arguments in Step 1 from binary states to finite states: if $\mu^*_t \notin \Psi_{LD}(A_t) \cup \text{Bd}_{\theta^*},$ then playing action $1$ is the unique subgame perfect equilibrium under the information policy $\bm{\mu}^*$. As described in the main text, our policy is such that beliefs move either in the direction $\hat d(\mu)$ toward $\delta_{\theta^*}$, or in the direction $- \hat d(\mu)$ away from $\delta_{\theta^*}$. 


\noindent \underline{Step 3}: We establish sequential optimality:
\begin{itemize}[leftmargin=*]
    \item \underline{Step 3A:} for any $\epsilon > 0$, $\bm{\mu^*}$ is $\epsilon$-sequentially optimal when $\mu^*_t \in \Psi_{LD}(A_t) \cup \text{Bd}_{\theta^*}$
    \item \underline{Step 3B:}  $\bm{\mu^*}$ is sequentially optimal when $\mu^*_t \notin \Psi_{LD}(A_t) \cup \text{Bd}_{\theta^*}.$
\end{itemize}

\begin{proof}[Proof of Theorem \ref{thrm:main}] We proceed along the steps outlined above. \\~\\
\noindent \underline{\textbf{Step 1.}} Suppose that $\Theta = \{0,1\}$ and $\theta^* = 1$. With slight abuse of notation, we associate beliefs $\mu_t$ with the probability that $\theta = 1$. As in the main text, we let $\psi_{LD}(A_t)$ denote the boundary of the lower dominance region. We will show that as long as action $1$ is not strictly dominated i.e., $\mu^*_t > \psi_{LD}(A_t)$, then action $1$ is played under any subgame perfect equilibrium. 

\begin{definition}\label{def: psi}
For $n \in \N$, we will construct a sequence $(\Psi^n)_n$ where $\Psi^n \subset \Delta(\Theta) \times [0,1]$ is a subset of the round-$n$ dominance region. $\Psi^n$ will satisfy the following conditions: 
\begin{enumerate}
    \item[(i)] \textbf{Contagion.} Action $1$ is strictly preferred under every history $H$ where $S(H) \in \Psi^n$ under the conjecture that action $1$ is played under every history $H'$ such that $S(H') \in  \Psi^{n - 1}$.
    \item[(ii)] \textbf{Translation.} There exists a constant $c_n > 0$ such that $\Psi^n = \{(\mu,A) : D(\mu,A) \geq c_n \},$ where $D(\mu,A) = \mu - \psi_{LD}(A)$.
\end{enumerate}
We initialize $\Psi^0$ as the upper dominance region whereby $1$ is strictly dominant. 
\end{definition}

Observe also that since $\Delta u(\cdot,\theta)$ is continuous and strictly increasing on a compact domain, it is also Lipschitz and we let the constant be $L > 0$. This also implies the lower dominance region (as a function of $A$) is Lipschitz continuous, and we denote the constant with $L_{\psi}$. 
\begin{lemma}
    \label{lem: lipschitz}$\psi_{LD}(\cdot)$ is Lipschitz continuous. 
\end{lemma}

\begin{proof}[Proof of Lemma \ref{lem: lipschitz}]
    Fix any $t$. The expected payoff difference between playing $1$ and $0$ when everyone in the future switches to action $1$ is given by
    \begin{align*}
        \Delta \bar U(\mu_t, A_t) &:= \mu_t \Ex_\tau\bigg[\int_{s=t}^{\tau} e^{-r(s-t)} \Big\{ \Delta u (\bar{A}_s,1) - \Delta u (\bar{A}_s,0) \Big\} ds\bigg] \\
        &\quad\quad\quad\quad + \Ex_\tau\bigg[\int_{s=t}^{\tau} e^{-r(s-t)} \Delta u (\bar{A}_s,0) ds\bigg].
    \end{align*}
    Note that $\Delta \bar U$ is continuously differentiable and strictly increasing in both $\mu_t$ and $A_t$. Since the domain of $\Delta \bar U$ is compact, the following values are well-defined:
    \begin{align*}
        L := \max_{A, \mu} \frac{\partial \Delta \bar U}{\partial A} > 0, \quad l := \min_{A, \mu} \frac{\partial \Delta \bar U}{\partial \mu} > 0.
    \end{align*}
    Then, for any $A_t < A_t'$ and $\mu_t > \mu_t'$, we have
    \[
    \Delta \bar U(\mu_t, A_t) - \Delta \bar U(\mu_t', A_t') \geq -L(A_t' - A_t) + l(\mu_t - \mu_t')
    \]
    because the mean value theorem implies
    \begin{align*}
        \Delta \bar U(\mu_t, A_t) - \Delta \bar U(\mu_t', A_t) &\geq l(\mu_t - \mu_t') \\
        \Delta \bar U(\mu_t', A_t') - \Delta \bar U(\mu_t', A_t) &\leq L(A_t' - A_t).
    \end{align*}
    Substituting $\mu_t = \psi_{LD} (A_t)$ and $\mu_t' = \psi_{LD} (A_t')$ into the above inequality yields
    \begin{align*}
        0 &= \Delta \bar U(\psi_{LD} (A_t), A_t) - \Delta \bar U(\psi_{LD} (A_t'), A_t') \tag{$\Delta \bar U(\psi_{LD}(A), A) = 0$}\\
        &\geq -L(A_t' - A_t) + l(\psi_{LD} (A_t) - \psi_{LD} (A_t')).
    \end{align*}
    Hence, we have
    \[
    \psi_{LD} (A_t) - \psi_{LD} (A_t') \leq \underbrace{\frac{L}{l}}_{=:L_{\psi}} (A_t' - A_t).
    \]
\end{proof}

\noindent \noindent \underline{\textbf{Step 1A.}} Construct $\Psi^0$. 

Define $\Psi^0$ as
    \[
     \Psi^0 = \Big\{(\mu,A) \in \Delta(\Theta) \times [0,1]: D(\mu,A) \geq c_0 \Big\},
    \]
    with $c_0 := \max_A \psi_{UD}(A) - \psi_{LD}(A)$, where $\psi_{UD}(A)$ is defined as
    \[
    \psi_{UD}(A) := \min\Big\{\mu \in \Delta(\Theta): 
    \Ex \Big[ \int_{t} u(1,\underline{A}_{s},\theta) ds \Big] \geq 
    \Ex \Big[ \int_{t} u(0,\underline{A}_{s},\theta) ds \Big]
    \Big\}.
    \]
    $\psi_{UD}(A)$ is the minimum belief under which players prefer action $1$ even if all future players choose to play action $0$.
    
\begin{lemma}
    \label{lem: dominance region}
    Action $1$ is strictly preferred under every history $H$ where $S(H) \in \Psi^0$.
\end{lemma}

\begin{proof}[Proof of Lemma~\ref{lem: dominance region}]
    Fix any history $H$ such that $S(H) \in \Psi^0.$ Then, by the definition of $\Psi_0$, the current $(\mu, A)$ satisfies
    \begin{align*}
        \mu \geq \psi_{LD}(A) + \max_{A'} \left\{\psi_{UD}(A') - \psi_{LD}(A')\right\} \geq \psi_{UD}(A).
    \end{align*}
    Hence, action $1$ is strictly preferred regardless of others' future play.
\end{proof}

\noindent \underline{\textbf{Step 1B.}} Construct a lower bound for the expected payoff difference given $\Psi^n$. 

Suppose that everyone plays action $1$ for any histories $H'$ such that $S(H') = (\mu, A)$ is in the round-$n$ dominance region $\Psi^n$. To obtain $\Psi^{n+1}$ in Step 1C, we derive the lower bound on the expected payoff difference of playing $0$ and $1$ given $\Psi^n$.

To this end, fix any history $H$ with the current target aggregate action $Z_t$ such that $S(H) = (\mu_t,A_t) \notin \Psi^n$ but $\mu_t > \psi_{LD}(A_t)$. From our construction of $Z_t,$ we must have $Z_t - A_t < \mathsf{TOL}(D(\mu_t,A_t)).$\footnote{By construction, if $Z_{t-} - A_t \geq \mathsf{TOL}(D(\mu_t,A_t))$, $Z_t = A_t$ must hold, which implies $Z_t - A_t = 0 < \mathsf{TOL}(D(\mu_t,A_t))$. If $Z_{t-} - A_t < \mathsf{TOL}(D(\mu_t,A_t))$, $Z_t - A_t < \mathsf{TOL}(D(\mu_t,A_t))$ is immediate because $Z_t$ does not jump.} For any continuous path $(A_s)_{s \geq t}$, we define the hitting time $T^*((A_s)_{s \geq t})$ as follows:
    \[T^* = \inf \Big\{s \geq t : Z_s - A_s \geq \mathsf{TOL}(D(\mu_t,A_s)) \text{ or } (\mu_t,A_s) \in \Psi^n \Big\}.\]
$T^*$ represents the first time at which either the designer injects new information, or the pair $(\mu_s, A_s)$ enters the round-$n$ dominance region. We will calculate the agent's expected payoff before time $T^*$ given the continuous path $(A_s)_{s \in [t,T^*]}$ and find a lower bound for this payoff by using the lower bound of $A_s$ for $s \geq t$. 
    
\noindent  \textbf{Before time ${T}^*$.} First, we calculate the agent's payoff before time $T^*$. Given $(A_s)_{s \geq t}$, we have $\mu_s = \mu_t$ and $(\mu_t,A_s ) \notin \Psi^n$ for every $s\in [t, T^*)$ because no information is injected when $Z_s - A_s < \mathsf{TOL}(D(\mu_t,A_s))$. Define $\psi^n(A_t) = \sup \{\mu \in \Delta(\Theta): (\mu,A_t) \notin \Psi^n\}$. This implies
    \begin{align} 
    \mathsf{TOL}(D(\mu_t,A_s)) &= \mathsf{TOL}(\mu_t - \psi_{LD}(A_s))\notag\\
    &\leq \mathsf{TOL}(\psi^n(A_s) - \psi_{LD}(A_s))\tag{$\mathsf{TOL}$ is increasing} \\
    &= \mathsf{TOL}(\psi^n(A_t) - \psi_{LD}(A_t)), \quad\text{(Translation property of $\Psi^n$)} \label{ineq: delta}
    \end{align}
Let $\bar{A}_s = \bar{A}(A_t,s-t)$, which is the aggregate play at $s \geq t$ when everyone will switch to action $1$ as fast as possible. By the definition of $Z$, we must have $Z_s = \bar{A}_s$ for every $s\in [t, T^*)$ because $Z_s - A_s < \mathsf{TOL}(D(\mu_s, A_s))$. Then we can write down the lower bound of $A_s$ when $s \in [t,T^*]$ as follows:
\begin{align}\label{ineq: lb aggregate action}
    A_s \geq \bar{A}_s - \mathsf{TOL}(D(\mu_t,A_s)) \geq \bar{A}_s - \mathsf{TOL}(\psi^n(A_t) - \psi_{LD}(A_t)).\quad \text{(From \eqref{ineq: delta})}
\end{align}
By Lipschitz continuity of $\Delta u(\cdot, \theta),$ we must have  
\begin{align}
    \Delta u (A_s,\theta) \geq \Delta u (\bar{A}_s, \theta) - \mathsf{TOL}(\psi^n(A_t) - \psi_{LD}(A_t))L \label{eq: lb before trigger}
\end{align}
with the Lipschitz constant $L$. Thus, the expected payoff difference of taking action $1$ and $0$ at time $(\mu_t,A_t)$ given a continuous path $(A_s)_{s \in [t,T^*]}$ before time $T^*$
is:
\begin{align}\label{eq: lower bound before T^*}
    &\Ex_\tau\bigg[\int_{s=t}^{\tau \wedge T^*} e^{-r(s-t)} \Delta u(A_s,\theta) ds \Big \lvert (A_s)_{s \in [t,T^*]} \bigg]\notag \\
    &= \Ex_\tau\bigg[\int_{s=t}^{\tau \wedge T^*} e^{-r(s-t)} \Delta u (\bar{A}_s,\theta) ds \bigg] \notag
    \\
    & \quad \quad + \Ex_\tau\bigg[\int_{s=t}^{\tau \wedge T^*} e^{-r(s-t)} \big(\Delta u (A_s,\theta) - \Delta u (\bar{A}_s,\theta) \big) ds \Big \lvert (A_s)_{s \in [t,T^*]}\bigg]\notag \\
    &\geq \Ex_\tau\bigg[\int_{s=t}^{\tau \wedge T^*} e^{-r(s-t)} \Delta u (\bar{A}_s,\theta) ds  \bigg] \notag
    \\ & \quad \quad - \mathsf{TOL}(\psi^n(A_t) - \psi_{LD}(A_t))L \cdot \Ex_\tau\bigg[\int_{s=t}^{\tau \wedge T^*} e^{-r(s-t)}  ds \Big\lvert (A_s)_{s \in [t,T^*]}\bigg] \tag{From (\ref{eq: lb before trigger})}\\
    &\geq \Ex_\tau\bigg[\int_{s=t}^{\tau \wedge T^*} e^{-r(s-t)} \Delta u (\bar{A}_s,\theta) ds \bigg] - \frac{\mathsf{TOL}(\psi^n(A_t) - \psi_{LD}(A_t))L}{\lambda}, \label{eq: lower bound before T^*}
\end{align}
where the last inequality follows from
\[
\Ex_\tau\bigg[\int_{s=t}^{\tau \wedge T^*} e^{-r(s-t)}  ds \Big\lvert (A_s)_{s \in [t,T^*]}\bigg] \leq \Ex_\tau\bigg[\int_{s=0}^{\tau}   ds\bigg] = \frac{1}{\lambda}.
\]

\noindent \textbf{After time ${T^*}$.} We calculate the lower bound of the expected payoff difference after time $T^*$.  We know that $\mu_{T^*-} = \mu_{t}.$ From the definition of $T^*$, we consider the following two cases depending on whether $Z_{T^*} - A_{T^*} < \mathsf{TOL}(D(\mu_{T^*},A_{T^*}))$ holds or not.

\noindent \textbf{Case 1: ${Z_{T^*} - A_{T^*} < \mathsf{TOL}(D(\mu_{T^*},A_{T^*}))}$.} This means $\mu_{T^*} = \mu_t$ because no information has been injected until $T^*$. Then the definition of $T^*$ implies $(\mu_{T^*},A_{T^*}) \in \Psi^n,$ where $\Psi^n$ is the round-$n$ dominance region. This means every agent strictly prefers to take action 1 at $T^*$. This increases $A_{T^*}$, inducing every agent taking action 1 after time $T^*$.\footnote{If $(\mu, A) \in \Psi^n$, then $(\mu, A') \in \Psi^n$ holds for any $A' \geq A.$} Thus, for $s \geq T^*$, we have
\begin{align}
    A_s = \bar{A}(A_{T^*},s-T^*)
    &\geq \bar{A}(\bar{A}_{T^*} - \mathsf{TOL}(\psi^n(A_t) - \psi_{LD}(A_t)), s-T^*)\tag{From \eqref{ineq: lb aggregate action}}\\ 
    &\geq \bar{A}_s - \mathsf{TOL}(\psi^n(A_t) - \psi_{LD}(A_t)),\label{ineq: lower bound for aggregate action}
\end{align}
where the last inequality follows from
\begin{align*}
    &\bar{A}(\bar{A}_{T^*} - \mathsf{TOL}(\psi^n(A_t) - \psi_{LD}(A_t)), s-T^*)\\
    &= 1 - \left( 1 - \bar A_{T^*} + \mathsf{TOL}(\psi^n(A_t) - \psi_{LD}(A_t)) \right) \exp (-\lambda (s - T^*)) \tag{Definition of $\bar A(A, s - t)$} \\
    &= \bar A_s - \mathsf{TOL}(\psi^n(A_t) - \psi_{LD}(A_t))\exp (-\lambda (s - T^*)) \tag{$\bar A_s = \bar A(\bar A_{T^*}, s - T^*)$}\\
    &\geq  \bar A_s - \mathsf{TOL}(\psi^n(A_t) - \psi_{LD}(A_t)).
\end{align*}
Hence, by the Lipschitz continuity of $\Delta u (\cdot,\theta)$, if $(\mu_{T^*},A_{T^*}) \in \Psi^n$, then
\begin{align}
    \Delta u (A_s,\theta) \geq \Delta u (\bar{A}_s, \theta) - \mathsf{TOL}(\psi^n(A_t) - \psi_{LD}(A_t))L. \label{eq: lb after trigger case 1}
\end{align}
The expected payoff difference of taking action $1$ and $0$ at time $(\mu_t,A_t)$ given a path $(A_s)_{s \in [t,T^*]}$ after time $T^*$ is
\begin{align}
    &\Ex_\tau\bigg[\int_{s=\tau \wedge T^*}^{\tau} e^{-r(s-t)} \Delta u(A_s,\theta) ds \Big\lvert (A_s)_{s \in [t,T^*]} \bigg] \notag\\
    &= \Ex_\tau\bigg[\int_{s=\tau \wedge T^*}^{\tau} e^{-r(s-t)} \Delta u (\bar{A}_s,\theta) ds \bigg]\notag\\
    &\quad\quad + \Ex_\tau\bigg[\int_{s=\tau \wedge T^*}^{\tau} e^{-r(s-t)} \big(\Delta u (A_s,\theta) - \Delta u (\bar{A}_s,\theta) \big) ds\Big\lvert (A_s)_{s \in [t,T^*]}\bigg] \notag\\
    &\geq \Ex_\tau\bigg[\int_{s=\tau \wedge T^*}^{\tau} e^{-r(s-t)} \Delta u (\bar{A}_s,\theta) ds\bigg] \notag\\
    &\quad\quad- \mathsf{TOL}(\psi^n(A_t) - \psi_{LD}(A_t))L \cdot \Ex_\tau\bigg[\int_{s=\tau \wedge T^*}^{\tau} e^{-r(s-t)}  ds \Big\lvert (A_s)_{s \in [t,T^*]}\bigg] \tag{From (\ref{eq: lb after trigger case 1})} \\
    &\geq \Ex_\tau\bigg[\int_{s=\tau \wedge T^*}^{\tau} e^{-r(s-t)} \Delta u (\bar{A}_s,\theta) ds \bigg] - \frac{\mathsf{TOL}(\psi^n(A_t) - \psi_{LD}(A_t))L  }{\lambda}. \label{eq: lower bound after T^* case 1}
\end{align}

\noindent \textbf{Case 2: ${Z_{T^*} - A_{T^*} \geq \mathsf{TOL}(D(\mu_{T^*}, A_s))}$.} By the definition of $T^*$, information is injected at $T^*$, and thus the belief at $T^*$ must be
\begin{align*}
    \mu_{T^*} = \begin{cases}
        \mu_{t} + M\cdot\mathsf{TOL}(D(\mu_{t},A_{T^*}))   &\text{w.p. $p_+(\mu_{t},A_{T^*})$}\\
        \mu_{t} - \mathsf{DOWN}(D(\mu_{t},A_{T^*})) &\text{w.p. $p_-(\mu_{t},A_{T^*})$}.
    \end{cases}
\end{align*}
Note that, if $(\mu_{T^*}, A_{T^*}) \in \Psi^n$, then everyone strictly prefers to take action $1$ at $T^*$. This increases $A_{T^*}$ and induces every agent to take action $1$ after time $T^*$ because $(\mu_s, A_s)$ stays in $\Psi^n$ for all $s \geq T^*$. Hence, we can write down the lower bound of $A_s$ when $s>T^*$ as follows:
\begin{align*}
    A_s &\geq 1\{(\mu_{T^*},A_{T^*}) \in \Psi^n\} \bar{A}(A_{T^*}, s-T^*) + 1\{(\mu_{T^*},A_{T^*}) \notin \Psi^n\} \underbar{$A$}(A_{T^*}, s-T^*) \tag{$A_s \geq \underbar{$A$}(A_{T^*}, s-T^*)$}\\
    &\geq 1\{(\mu_{T^*},A_{T^*}) \in \Psi^n\} \{\bar{A}_s - \mathsf{TOL}(\psi^n(A_t) - \psi_{LD}(A_t))\} \\
    &\quad\quad\quad\quad\quad\quad + 1\{(\mu_{T^*},A_{T^*}) \notin \Psi^n\} \underbar{$A$}(A_{T^*}, s-T^*). \tag{From \eqref{ineq: lower bound for aggregate action}}
\end{align*}
By Lipschitz continuity of $\Delta u (\cdot,\theta)$, we must have, if $(\mu_{T^*},A_{T^*}) \in \Psi^n$, then
\begin{align}
    \Delta u (A_s,\theta) \geq \Delta u (\bar{A}_s, \theta) - \mathsf{TOL}(\psi^n(A_t) - \psi_{LD}(A_t))L, \label{eq: lb after trigger case 2}
\end{align}
and if $(\mu_{T^*},A_{T^*}) \notin \Psi^n$, then 
\begin{align}
    \Delta u (A_s,\theta) \geq \Delta u (\bar A_s, \theta) - L (\bar A_s - A_s) \geq \Delta u (\bar A_s, \theta) - L.
    \label{eq: lb after trigger case 2, bad case}
\end{align}

Define $p_n \coloneqq \Pr((\mu_{T^*},A_{T^*})  \in \Psi^n \mid (A_s)_{s \in [t,T^*]})$. The expected payoff difference of taking action $1$ and $0$ at time $(\mu_t,A_t)$ given a path $(A_s)_{s \in [t,T^*]}$ after time $T^*$
is
\begin{align}
    &\Ex_\tau\bigg[\int_{s=\tau \wedge T^*}^{\tau} e^{-r(s-t)} \Delta u(A_s,\theta) ds \Big \lvert (A_s)_{s \in [t,T^*]} \bigg] \notag\\
    &= \Ex_\tau\bigg[\int_{s=\tau \wedge T^*}^{\tau} e^{-r(s-t)} \Delta u (\bar{A}_s,\theta) ds \bigg] \notag\\
    &\quad\quad + \Ex_\tau\bigg[\int_{s=\tau \wedge T^*}^{\tau} e^{-r(s-t)} \big(\Delta u (A_s,\theta) - \Delta u (\bar{A}_s,\theta) \big) ds \Big \lvert (A_s)_{s \in [t,T^*]}\bigg] \notag \\
    &\geq \Ex_\tau\bigg[\int_{s=\tau \wedge T^*}^{\tau} e^{-r(s-t)} \Delta u (\bar{A}_s,\theta) ds \bigg] \notag\\
    &\quad\quad- \big(\mathsf{TOL}(\psi^n(A_t) - \psi_{LD}(A_t))L p_n + L(1-p_n)  \big) \cdot \Ex_\tau\bigg[\int_{s=\tau \wedge T^*}^{\tau} e^{-r(s-t)}  ds \Big \lvert (A_s)_{s \in [t,T^*]}\bigg] \tag{From (\ref{eq: lb after trigger case 2}) and \eqref{eq: lb after trigger case 2, bad case}} \\
    &\geq \Ex_\tau\bigg[\int_{s=\tau \wedge T^*}^{\tau} e^{-r(s-t)} \Delta u (\bar{A}_s,\theta) ds \bigg] - \frac{\mathsf{TOL}(\psi^n(A_t) - \psi_{LD}(A_t))Lp_n + L(1-p_n)  }{\lambda}. \label{eq: lower bound after T^* case 2}
\end{align}

\noindent  \textbf{Combining before and after time ${T^*}$.} We are ready to construct a lower bound of the expected discounted payoff difference. To evaluate $s \geq T^*$, it is sufficient to focus on the case in which information is injected (Case 2) since \eqref{eq: lower bound after T^* case 2} is smaller than \eqref{eq: lower bound after T^* case 1} because $\mathsf{TOL}(\psi^n(A_t) - \psi_{LD}(A_t)) < 1$. By taking the sum of the payoffs before and after time $T^*$, that is \eqref{eq: lower bound before T^*} and \eqref{eq: lower bound after T^* case 2}, the expected payoff difference of taking action 1 and 0 at $(\mu_t,A_t)$ given a path $(A_s)_{s \in [t,T^*]}$ is lower-bounded as follows: 
\begin{align}
    &\Ex\Big[U_1(\mu_t,(A_s)_{s \geq t}) - U_0(\mu_t, (A_s)_{s \geq t}) \Big| (A_s)_{s \in [t,T^*]} \Big]\notag  \\
    &\geq \Ex_\tau\bigg[\int_{s=t}^{\tau} e^{-r(s-t)} \Delta u (\bar{A}_s,\theta) ds\bigg]  - \frac{\mathsf{TOL}(\psi^n(A_t) - \psi_{LD}(A_t))L(1 + p_n) + L(1-p_n)  }{\lambda}  
    \tag{LB} \label{eq: lower bound}.
\end{align}

Intuitively, the expected payoff cannot be too low compared to the case where everyone switches to action $1$ in the future because (i) aggregate actions are close to the target before new information is injected; and (ii) if the belief jumps upward upon injection, everyone will subsequently switch to action $1$.

\noindent \underline{\textbf{Step 1C.}} Finally, we characterize $\Psi^{n+1}$. The following lemma establishes that under $\bm{\mu^*},$ $\Psi^n$ is strictly increasing in the set order.
\begin{lemma}
    \label{lem: contagion}
    Given $\Psi^n$, there exists $\Psi^{n+1} \subsetneq \Psi^n$ such that $\Psi^{n+1}$ satisfies \textbf{Contagion} and \textbf{Translation} in Definition \ref{def: psi}.
\end{lemma}

\begin{proof}[Proof of Lemma~\ref{lem: contagion}]
To characterize $\Psi^{n+1}$, we first show that there exist tolerance level $\delta$, upward jump magnitude $M,$ and downward jump size $\epsilon$ such that if $\mu_t \geq \psi^n(A_t) - M\cdot\mathsf{TOL}(D(\mu_t,A_t))/2$, then $U_1(\mu_t,(A_s)_{s \geq t}) - U_0(\mu_t, (A_s)_{s \geq t}) > 0$. 

Suppose $\mu_t \geq \psi^n(A_t) - M\cdot\mathsf{TOL}(D(\mu_t,A_t))/2$. First, we evaluate the first term of \eqref{eq: lower bound}. We know from the definition of the lower dominance region $\psi_{LD}(A_t)$ that 
\[\psi_{LD}(A_t) \underbrace{\Ex_\tau\bigg[\int_{s=t}^\tau e^{-r(s-t)} \Delta u(\bar{A}_s,1) ds \bigg]}_{\geq 0} + (1-\psi_{LD}(A_t))\Ex_\tau\bigg[\int_{s=t}^\tau e^{-r(s-t)} \Delta u(\bar{A}_s,0) ds \bigg]\geq 0\]
with equality when $\psi_{LD}(A_t) > 0$. Hence,
 if $\psi_{LD}(A_t) > 0$, we have  
\begin{align}
    &\Ex_{\tau,\theta \sim \mu_t}\bigg[\int_{s=t}^{\tau} e^{-r(s-t)} \Delta u (\bar{A}_s,\theta) ds\bigg]\notag\\
    &= (\mu_t - \psi_{LD}(A_t))\Ex_{\tau}\bigg[\int_{s=t}^{\tau} e^{-r(s-t)} (\Delta u (\bar{A}_s,1) - \Delta u (\bar{A}_s,0)) ds \bigg] \notag\\
    &\geq (\mu_t - \psi_{LD}(A_t))\Ex_{\tau}\bigg[\int_{s=t}^{\tau} e^{-r(s-t)} \Delta u (\bar{A}_s,1) ds \bigg] \tag{$\Ex_\tau\big[\int_{s=t}^\tau e^{-r(s-t)} \Delta u(\bar{A}_s,0) ds \big] \leq 0$}\\
    &> C(\mu_t-\psi_{LD}(A_t)) \label{ineq: conctant C}
\end{align}
for some $C > 0$. This constant $C$ exists because
\[
\min_{A_t \in [0,1]} \Ex_{\tau}\bigg[\int_{s=t}^{\tau} e^{-r(s-t)} \Delta u (\bar{A}_s,1) ds \bigg] > 0
\]
since $\Delta u (A,1) > 0 $ for any $A \in [0, 1]$. If $\psi_{LD}(A_t) = 0$, we have $\Ex_\tau\big[\int_{s=t}^\tau e^{-r(s-t)} \Delta u(\bar{A}_s,1) ds \big] \geq 0$, which implies
\begin{align*}
    \Ex_{\tau,\theta \sim \mu_t}\bigg[\int_{s=t}^{\tau} e^{-r(s-t)} \Delta u (\bar{A}_s,\theta) ds\bigg] &\geq \mu_t\Ex_{\tau}\bigg[\int_{s=t}^{\tau} e^{-r(s-t)} \Delta u (\bar{A}_s,1)  ds \bigg] \tag{$\Ex_\tau\big[\int_{s=t}^\tau e^{-r(s-t)} \Delta u(\bar{A}_s,0) ds \big] \geq 0$} \\
    &> C(\mu_t-\psi_{LD}(A_t)).
\end{align*}

Additionally, note that if $\mathsf{TOL}$ satisfies $M\cdot\mathsf{TOL}(D(\mu,A)) \leq D(\mu,A)$ for every $(\mu,A)$, then $\mu_t-\psi_{LD}(A_t) \geq \frac{1}{2}(\psi^n(A_t) - \psi_{LD}(A_t))$. This follows from
\begin{align*}
    &\mu_t-\psi_{LD}(A_t) - \frac{1}{2}(\psi^n(A_t) - \psi_{LD}(A_t)) \\
    \geq& \frac{1}{2}(\psi^n(A_t) - \psi_{LD}(A_t)) - M\cdot\mathsf{TOL}(D(\mu_t,A_t))/2  \tag{$\mu_t \geq \psi^n(A_t) - M\cdot\mathsf{TOL}(D(\mu_t,A_t))/2$}\\
    \geq& M\cdot\mathsf{TOL}(D(\psi^n(A_t),A_t))/2 - M\cdot\mathsf{TOL}(D(\mu_t,A_t))/2 \tag{$M\cdot\mathsf{TOL}(D(\psi^n(A_t),A_t)) \leq D(\psi^n(A_t),A_t)$} \\
    \geq & 0. \tag{$\psi^n(A_t) \geq \mu_t$}
\end{align*}
Thus, if $\mu_t \geq \psi^n(A_t) - M\cdot\mathsf{TOL}(D(\mu_t,A_t))/2$, then
\begin{align}\label{ineq: first term}
    \Ex_{\tau,\theta \sim \mu_t}\bigg[\int_{s=t}^{\tau} e^{-r(s-t)} \Delta u (\bar{A}_s,\theta) ds\bigg]
    &> \frac{C}{2}(\psi^n(A_t) - \psi_{LD}(A_t)).
\end{align}

Next, we evaluate the second term of \eqref{eq: lower bound}. Notice that, if $(\mu_t, A_t) \notin \Psi^n$, then
\begin{align*}
    p_n = \Pr((\mu_{T^*},A_{T^*}) \in \Psi^n) = p_+(\mu_t,A_{T^*})  1\{(\mu_t + M\cdot\mathsf{TOL}(D(\mu_t,A_{T^*})),A_{T^*}) \in \Psi^n\}.
\end{align*}
We will show that $(\mu_t + M\cdot\mathsf{TOL}(D(\mu_t,A_{T^*})),A_{T^*}) \in \Psi^n$ if $\mu_t \geq \psi^n(A_t) - M\cdot\mathsf{TOL}(D(\mu_t,A_t))/2$. Observe that when $(\mu_t, A_t) \notin \Psi^n$, we must have
\begin{align} \label{ineq: lb A_T^*}
    A_{T^*} > A_t - \mathsf{TOL}(D(\mu_t,A_t)).
\end{align}
To see this, suppose for a contradiction that $A_{T^*} \leq A_t - \mathsf{TOL}(D(\mu_t,A_t))$, which implies $A_{T^*} < A_t$. However, since the definition of $T^*$ implies $A_{T^*} = Z_{T^*} - \mathsf{TOL}(D(\mu_t, A_{T^*}))$, we have
\begin{align*}
    A_{T^*} &= Z_{T^*} - \mathsf{TOL}(D(\mu_t, A_{T^*})) \\
    &> A_t - \mathsf{TOL}(D(\mu_t, A_t)), \tag{$Z_{T^*} = \bar A_{T^*} > A_t$ and $A_t > A_{T^*}$}
\end{align*}
which is a contradiction.

Lemma~\ref{lem: lipschitz} shows that $\psi_{LD}$ is a Lipschitz function. Since $\psi^n(A_t)$ is a translation of $\psi_{LD}(A_t)$, $ \psi^n(A_t)$ has the same Lipschitz constant $L_{\psi}$ as $\psi_{LD}(A_t)$. Hence, if $\mu_t \geq \psi^n(A_t) - M\cdot\mathsf{TOL}(D(\mu_t,A_t))/2$, we must have 
\begin{align*}
    \mu_t + M\cdot\mathsf{TOL}(D(\mu_t,A_{T^*})) &\geq  \psi^n(A_t) + M\cdot\mathsf{TOL}(D(\mu_t,A_t))/2 \\
    &> \big( \psi^n(A_{T^*}) - L_{\psi}\mathsf{TOL}(D(\mu_t,A_t))\big) + M\cdot\mathsf{TOL}(D(\mu_t,A_t))/2 \tag{From \eqref{ineq: lb A_T^*} and Lipschitz continuity of $\psi^n$}\\
    &=  \psi^n(A_{T^*}),
\end{align*}
by setting $M = 2L_{\psi}$. Thus, $(\mu_t + M\cdot\mathsf{TOL}(D(\mu_t,A_{T^*})),A_{T^*}) \in \Psi^n$ holds, implying \[p_n = p_{+}(\mu_t,A_{T^*}) = \frac{\mathsf{DOWN}(D(\mu_{t},A_{T^*}))}{\mathsf{DOWN}(D(\mu_{t},A_{T^*})) + M\cdot\mathsf{TOL}(D(\mu_t,A_{T^*}))}.\]
We set \[\mathsf{DOWN}(D(\mu_t,A_t)) = \frac{\mu_t-\psi_{LD}(A_t)}{2} \quad \text{\&} \quad  \mathsf{TOL}(D(\mu_t,A_t)) = \bar\delta \cdot \frac{\lambda C (\mu_t - \psi_{LD}(A_t))}{4L + 4LM (\mu_t-\psi_{LD}(A_t))^{-1}},\]
for a fixed small number $\bar\delta<1$ so that $\mathsf{TOL}(D(\mu,A)) < 1$ and $M\cdot\mathsf{TOL}(D(\mu,A)) \leq D(\mu,A)$ for every $\mu$ and $A$ (e.g., $\bar\delta = \min\{1,\frac{4L}{\lambda C}, \frac{4L}{\lambda C M}\}$). Thus,
\begin{align}
    1-p_n &= \frac{M\cdot\mathsf{TOL}(D(\mu_t,A_{T^*}))}{\mathsf{DOWN}(D(\mu_{t},A_{T^*})) + M\cdot\mathsf{TOL}(D(\mu_t,A_{T^*}))}\notag \\
    &\leq \frac{M\cdot\mathsf{TOL}(D(\mu_t,A_{T^*}))}{\mathsf{DOWN}(D(\mu_{t},A_{T^*}))} \tag{$M\cdot\mathsf{TOL}(D(\mu_t,A_{T^*})) \geq 0$}\\
    &= \frac{\bar\delta\lambda M C}{2L + 2LM (\mu_t-\psi_{LD}(A_{T^*}))^{-1}} \notag\\
    &\leq \frac{\bar\delta\lambda M C}{2L + 2LM (\psi^n(A_t) - \psi_{LD}(A_t))^{-1}}.\quad\quad\quad \text{(From \eqref{ineq: delta} and continuity of $A_s$)} \label{ineq: second term}
\end{align}

Thus, if $\mu_t \geq  \psi^n(A_t) - M\cdot\mathsf{TOL}(D(\mu_t,A_t))/2,$ we have
\begin{align*}
    &\Ex[U_1(\mu_t,(A_s)_{s \geq t}) - U_0(\mu_t, (A_s)_{s \geq t}) \mid (A_s)_{s \in [t,T^*]}] \\
    &> \frac{C}{2}( \psi^n(A_t) - \psi_{LD}(A_t))  - \frac{1}{\lambda} \Big( \mathsf{TOL}(\psi^n(A_t) - \psi_{LD}(A_t))L(1+p_n) +  L(1-p_n)\Big) \tag{From \eqref{eq: lower bound} and \eqref{ineq: first term}} \\
    &> \frac{C}{2}( \psi^n(A_t) - \psi_{LD}(A_t))  - \frac{L}{\lambda} \Big( 2\mathsf{TOL}(\psi^n(A_t) - \psi_{LD}(A_t)) +  (1-p_n)\Big) \\
    &\geq  \frac{C}{2}( \psi^n(A_t) - \psi_{LD}(A_t))  - \frac{\bar\delta L}{\lambda} \cdot \bigg(\frac{\lambda C (\psi^n(A_t) - \psi_{LD}(A_t)) + \lambda MC}{2L + 2LM (\psi^n(A_t) - \psi_{LD}(A_t))^{-1}} \bigg) \tag{From \eqref{ineq: second term}} \\
    &= \frac{C(1-\bar\delta)}{2}(\psi^n(A_t) - \psi_{LD}(A_t)) \\
    &>0,
\end{align*}
for every given path $(A_s)_{s \in [t,T^*]}$.

In conclusion, we found $\delta$, $M$, and $\epsilon$ such that if $\mu_t \geq  \psi^n(A_t) - M\cdot\mathsf{TOL}(D(\mu_t,A_t))/2$, then the agent must choose action $1$. Note that $\delta$ is increasing in $\mu_t$ and increasing in $A_t$. Thus, $\mu_t + M\cdot\mathsf{TOL}(D(\mu_t,A_t))/2$ is increasing in $\mu_t$ and continuous in $\mu_t$ when $\mu_t > \psi_{LD}(A_t)$. Therefore, for each $A_t$, there exists $\mu'(A_t) < \psi^n(A_t)$ such that 
\begin{align*}
    \mu'(A_t) + \frac{M\cdot\mathsf{TOL}(D(\mu'(A_t),A_t))}{2} =  \psi^n(A_t).
\end{align*}
Then we define
\[\psi^{n+1} = \{(\mu_t,A_t) : \mu_t \geq \mu'(A_t)\},\]
which also implies $\psi^{n+1}(A_t) := \sup \{\mu \in \Delta(\Theta): (\mu,A_t) \notin \Psi^{n+1}\}= \mu'(A_t)$. From the argument above, we must have an agent always choosing action $1$ whenever $(\mu_t,A_t) \in \Psi^{n+1}$ (\textbf{Contagion} in Definition~\ref{def: psi}). Moreover, we can rewrite the above equation as follows:
\begin{align*}
    (\psi^{n+1}(A_t)-\psi_{LD}(A_t)) + \frac{M\cdot\mathsf{TOL}(\psi^{n+1}(A_t)-\psi_{LD}(A_t))}{2} = \psi^n(A_t) - \psi_{LD}(A_t),
\end{align*}
where the RHS is constant in $A_t$ by the translation property of $\psi^n$. Thus, $\psi^{n+1}(A_t) - \psi_{LD}(A_t)$ must be also constant in $A_t$ (\textbf{Translation} in Definition~\ref{def: psi}). This concludes that round-$(n+1)$ dominance region $\Psi^{n+1}$ satisfies $\Psi^n \subset \Psi^{n+1}$ because $c_n = \psi^n(A_t) - \psi_{LD}(A_t) > \psi^{n+1}(A_t) - \psi_{LD}(A_t) =: c_{n+1}$.
\end{proof} 

\noindent \underline{\textbf{Step 1D.}} In the limit, the sequence $(\Psi^{n})_n$ covers the $(\mu,A)$ region where action $1$ is not strictly dominated.
\begin{lemma}
    \label{lem: induction}
    \[
    \bigcup_{n \in \mathbb{N}} \Psi^n = \Big\{(\mu,A) \in \Delta(\Theta) \times [0,1]: \mu > \psi_{LD}(A)\Big\}.
    \]
\end{lemma}

\begin{proof}[Proof of Lemma~\ref{lem: induction}]
    Recall $\psi^n(A_t) = \sup\{\mu \in \Delta(\Theta): (\mu,A_t) \notin \Psi^n\}.$ By Lemma~\ref{lem: contagion}, $\psi^n(A_t)$ is decreasing in $n$. Define $\psi^*(A_t) = \lim_{n\to\infty} \psi^n(A_t)$. In limit, we must have
\begin{align*}
    &\psi^*(A_t) + M\cdot\mathsf{TOL}(D(\psi^*(A_t),A_t))/2 = \psi^*(A_t)\\
    &\Rightarrow \mathsf{TOL}(D(\psi^*(A_t),A_t)) = 0 \Rightarrow \psi^*(A_t) = \psi_{LD}(A_t),
\end{align*}
which implies 
\[\bigcup_{n \geq 0} \Psi^n = \Big\{(\mu_t,A_t): \mu_t > \psi_{LD}(A_t)\Big\}\]
as required. 
\end{proof}

\noindent \underline{\textbf{Step 2.}} We have constructed an information policy which uniquely implements an equilibrium achieving \eqref{eqn:opt} for $|\Theta| = 2$. We now lift this to the case with finite states $\Theta = \{\theta_1,\ldots \theta_n\}$ as set out in the main text, where recall we set $\theta^*$ as the dominant state.  

In particular, we show that if $\mu^*_t \notin \Psi_{LD}(A_t) \cup \text{Bd}_{\theta^*},$ then playing action $1$ is the unique subgame perfect equilibrium under the information policy $\bm{\mu}^*$. To apply Step 1, we will construct an auxiliary binary-state environment for each direction from $\delta_{\theta^*}$.

To this end, we call a vector $\hat{\bm d} = (\hat{d}_\theta)_{\theta \in \Theta} \in \mathbb{R}^n$ a \emph{feasible directional vector} if  $\sum_\theta{\hat{d}_\theta} = 0$ and $\hat{d}_{\theta^*} = 1$ but $\hat{d}_{\theta} < 0$ if $\theta^* \ne \theta$. For each feasible directional vector $\hat{\bm d}$, define a function $\bar{\alpha}_{\hat{\bm d}}: [0,1] \to [0,1]$ such that, for every $A \in [0,1],$
    \[\bar{\alpha}_{\hat{\bm d}}(A) = \inf\Big\{ \alpha \in [0,1] : \delta_{\theta^*} - (1-\alpha) \hat{\bm d}  \notin \Psi_{LD}(A) \cup \text{Bd}_{\theta^*} \Big\}.\]
    Note that $\delta_{\theta^*} - (1-\alpha) \hat{\bm d} \in \text{Bd}_{\theta^*}$ if and only if $\alpha = 0$ because $\hat{d}_{\theta^*} = 1$. Observe that
    \begin{align*}
        \big( \Psi_{LD}(A_t) \cup \text{Bd}_{\theta^*} \big) ^c = \bigcup_{\hat{\bm d} \in \mathcal{D}} \big\{\delta_{\theta^*} - (1-\alpha) \hat{\bm d}: \alpha \in (\bar{\alpha}_{\hat{\bm d}}(A_t)),1]\big\},
    \end{align*}
    where $\mathcal{D}$ is the set of all feasible directional vectors. This is true because 1) $\big(\Psi_{LD}(A_t)\big)^c$ is a polygon since the expectation operator is linear; and 2) $\Psi_{LD}(A_t) \cup \text{Bd}_{\theta^*} \Delta(\Theta)$ is closed. Thus, it is equivalent to show that, for every feasible directional vector $\hat{\bm d},$ if $\alpha \in (\bar{\alpha}_{\hat{\bm d}}(A_t),1]$, then playing action $1$ is the unique subgame perfect equilibrium under the information policy $\bm{\mu^*}$.

    Fix a feasible directional vector $\hat{\bm d}.$ Define
    \[\Delta(\Theta)_{\hat{\bm d}} = \big\{\delta_{\theta^*} - (1-\alpha) \hat{\bm d} : \alpha \in [0,1]\big\} \]
    as the set of beliefs whose direction from $\delta_{\theta^*}$ is $\hat{\bm d}$. Consider an auxiliary environment with binary state $\tilde{\Theta} = \{0,1\}.$ Construct a bijection $\psi_{\hat{\bm d}}: \Delta(\Theta)_{\hat{\bm d}} \to \Delta(\tilde{\Theta})$ such that $\psi_{\hat{\bm d}}(\mu) =  \alpha$ if $\mu = \delta_{\theta^*} - (1-\alpha) \hat{\bm d}.$ Denote $\tilde{\mu} \coloneqq \psi_{\hat{\bm d}}(\mu) \in \Delta(\tilde{\Theta})$ for every $\mu \in \Delta(\Theta)_{\hat{\bm d}}$. Note that $\psi_{\hat{\bm d}}(\delta_{\theta^*}) = 1$.

    We define a flow payoff for each player under the new environment $\tilde{u}: \{0,1\} \times [0,1] \times \tilde{\Theta} \to \mathbb{R}$ as follows:
    \[\tilde{u}(a,A,\tilde{\theta}) = u\left(a,A, \psi^{-1}_{\hat{\bm d}}(\tilde{\theta}) \right).  \]
    Define $\Delta \tilde{u} (A,\tilde\theta) := u(1,A,\tilde\theta) - u(0,A,\tilde\theta)$. Since $\psi_{\hat{\bm d}}$ is a linear map, $\Delta  \tilde u (A,\theta)$ is still continuously differentiable and strictly increasing in $A.$ Also, given that $\Delta \tilde{u}(0,1) = \Delta u(0,\theta^*) > 0$, we still have an action-$1$-dominance region under this new environment.

    Then we can similarly define the maximum belief under which players prefer action $0$ even if all future players choose to play action $1:$
    \[\psi_{LD}^{\hat{\bm d}}(A_t) := \max\Big\{\tilde\mu \in \Delta(\tilde\Theta): 
    \Ex \Big[ \int_{t} \tilde u(0,\bar{A}_{s},\tilde\theta) ds \Big] \geq 
    \Ex \Big[ \int_{t} \tilde u(1,\bar{A}_{s},\tilde\theta) ds \Big]
    \Big\}.\]
    We define $\tilde D(\tilde{\mu},A) = \tilde{\mu} - \psi_{LD}^{\hat{\bm d}}(A)$. Then it is easy to see that $\tilde D(\mu,A) = D(\mu,A)$ for every $\mu \in \Delta(\tilde\Theta)$ and $A \in [0,1]$
    
    A key observation is that if $\mu_{t-} \notin \Psi_{LD}(A_t)$ and $\mu_{t-} \in \Delta(\Theta)_{\hat{\bm d}}$, then every future belief must stay in $\Delta(\Theta)_{\hat{\bm d}}$ almost surely with respect to any strategy. We can rewrite the time-$t$ information struture corresponding to the new environment as follows:
    \begin{enumerate}
    \item \textbf{Silence on-path.} If $\tilde{\mu}_{t-} > \psi_{LD}^{\hat{\bm d}}(A_t)$ and $|A_t - Z_{t-}| < \mathsf{TOL}(D)$
    \begin{center}
        $\mu_t = \mu_{t-}$ almost surely,
    \end{center}
    i.e., no information, and $dZ_t = \lambda(1-Z_{t-}).$
    \item \textbf{Noisy and asymmetric off-path.} If $\tilde{\mu}_{t-} > \psi_{LD}^{\hat{\bm d}}(A_t)$ and $Z_{t-} - A_t \geq \mathsf{TOL}(D),$
    \begin{align*}
        \tilde\mu_t = \begin{cases}
            \tilde\mu_{t-} + M \cdot \mathsf{TOL}(D) &\text{w.p. $\frac{\mathsf{DOWN}(D)}{\mathsf{DOWN}(D)+ M\cdot\mathsf{TOL}(D)}$} \\
            \tilde\mu_{t-} -  \mathsf{DOWN}(D) &\text{w.p. $\frac{M \cdot\mathsf{TOL}(D)}{\mathsf{DOWN}(D)+ M\cdot\mathsf{TOL}(D)}$},
        \end{cases}
    \end{align*}
    and reset $Z_t = A_t$.
\end{enumerate}
By applying Step 1, 
we conclude that if $\tilde\mu_{t-} > \psi_{LD}^{\hat{\bm d}}(A_t)$, then action $1$ is played under any subgame perfect equilibrium. The only subtlety is to verify that as in \eqref{ineq: conctant C}, there exists a constant $C>0$ such that 
\begin{align*}
    \min_{A_t \in [0,1]} \Ex_\tau\bigg[ \int_{s=t}^{\tau} e^{-r(s-t)} \Delta \tilde{u}(\bar{A}_s,1)  dt\bigg]  \geq C
\end{align*}
for any feasible directional vector $\hat{\bm d}.$ This is clear because $\Delta \tilde{u}(A,1) = \Delta u(A,\theta^*) > \Delta u(0,\theta^*)>0$ for every $A$ by the definition of $\theta^*$.

Since $[0,\psi_{LD}^{\hat{\bm d}}(A_t)] = \psi_{\hat{\bm d}} \big(\Delta(\Theta)_{\hat{\bm d}} \cap \Psi_{LD}(A_t)\big)$, we have $(\psi_{LD}^{\hat{\bm d}}(A_t), 1] = (\bar{\alpha}_{\hat{\bm d}}(A_t), 1]$. Hence, if $\alpha \in (\bar{\alpha}_{\hat{\bm d}}(A_t),1]$, then playing action $1$ is the unique subgame perfect equilibrium under the information policy $\bm{\mu^*}$, as desired. 

\noindent \underline{\textbf{Step 3.}} We now show sequential optimality. Step 3A handles the case when beliefs are such that $1$ is strictly dominated, while 3B handles the case when $1$ is not strictly dominated. 

\noindent \underline{\textbf{Step 3A.}} $\bm{\mu^*}$ is $\epsilon$-sequentially optimal when $\mu^*_t \in \Psi_{LD}(A_t).$

Fix any $\mu_0 \in \Psi_{LD}(A_0)$. Define $\tau^* \coloneqq \inf\{t: \mu_t \notin \Psi_{LD}(A_0)\}$ and $\bar{\tau} \coloneqq \inf\{t: \mu_t \notin \Psi_{LD}(A_t)\}$, i.e., $\tau^*$ and $\bar{\tau}$ are the first times $t$ at which the belief $\mu_t$ is not in $\Psi_{LD}(A_0)$ and $\Psi_{LD}(A_t)$, respectively. This means, at $s< \bar{\tau}$, all agents who can switch choose action $0$.   This pins down an aggregate action $A_t = \underline A (A_0, t)$ for every $t \leq \bar{\tau}$. Therefore, $A_t < A_0$ for every $t \leq \bar{\tau}$, implying $\Psi_{LD}(A_0) \subset \Psi_{LD}(A_t)$. Thus, $\tau^* \geq \bar{\tau},$ and so $A_t = \underline A (A_0, t)$ for every $t \leq \tau^*$.

Moreover, we know that $A_s \leq \bar A(A_t, s - t)$ for any $s \geq t$, so we can find an upper bound of the designer's payoff as follows:
\begin{align*}
    &\Ex^{\sigma}\Big[\phi(\bm{A})\Big]  \\
    &= \Ex_{\tau^*} \left[ \phi(\bm{A})\mathbb{1}(\tau^* = \infty) + \phi(\bm{A})\mathbb{1}(\tau^* < \infty) \right]\\
    &\leq \Ex_{\tau^*} \left[ \phi(\bm{\underline{A}})\mathbb{1}(\tau^* = \infty) + \phi(\bm{\bar{A}})\mathbb{1}(\tau^* < \infty) \right] \tag{$\bm{A} \leq \bm{\bar{A}}$}\\
    &= \phi(\bm{\underline{A}}) +   \left\{\phi(\bm{\bar{A}}) - \phi(\bm{\underline{A}})\right\}\Pr(\tau^* < \infty),
\end{align*}
where $\bm{\underline{A}}$ satisfies $\underline{A}_t = \underline{A}(A_0, t)$, and $\bm{\bar{A}}$ satisfies $\bar{A}_t = \bar{A}(A_0, t)$.

For every $t \in [0,\infty)$, the optional stopping theorem implies 
\begin{align*}
    \mu_0 &= \Ex\big[\mu_{\tau^* \wedge t} \big] \\
    &= \Ex[ \mu_{\tau^*} \mid \tau^* < t ] \Pr(\tau^* < t) + \Ex[\mu_{t} \mid \tau^* \geq t] \Pr(\tau^* \geq t) \\
    &\geq \underbrace{\Ex[ \mu_{\tau^*} \mid \tau^* < t ]}_{\eqqcolon \hat{\mu}_t } \Pr(\tau^* < t).
\end{align*}
This implies $\hat{\mu}_t \in F(\Pr(\tau^* < t),\mu_0)$ for every $t.$
By the definition of $\tau^*$ and $\mu_t$ is right-continuous, $\mu_{\tau^*} \in \overline{\Psi^c_{LD}(A_0)}$ under the event $\{\tau^* < \infty\}$. Since $\overline{\Psi^c_{LD}(A_0)}$ is convex, we also have $\hat{\mu}_t \in \overline{\Psi^c_{LD}(A_0)}.$ This means $\hat{\mu}_t \notin \text{Int } \Psi_{LD}(A_0)$, but  $\hat{\mu}_t \in F(\Pr(\tau^* < t),\mu_0)$. The definition of $p^*(\mu_0,A_0)$ implies $p^*(\mu_0,A_0) \geq \Pr(\tau^* < t)$ for every $t.$ Thus, 
\begin{align*}
    \Ex^{\sigma}\Big[\phi(\bm{A})\Big] &\leq \phi(\bm{\underline{A}}) +   \left\{\phi(\bm{\bar{A}}) - \phi(\bm{\underline{A}})\right\}p^*(\mu_0,A_0)\\
    &= (1-p^*(\mu_0,A_0)) \phi(\bm{\underline{A}}) + p^*(\mu_0,A_0)\phi(\bm{\bar{A}}).
\end{align*}
This implies
\begin{align*}
    \eqref{eqn:opt} = \sup_{\substack{\bm{\mu} \in \mathcal{M} \\
    \sigma \in \Sigma(\bm{\mu}, A_0)}}\Ex^{\sigma}\Big[\phi(\bm{A})\Big] \leq  (1-p^*(\mu_0,A_0)) \phi(\bm{\underline{A}}) + p^*(\mu_0)\phi(\bm{\bar{A}}).
\end{align*}
Under $\bm{\mu}^*$, if $\mu_{0+} \in \Psi^c_{LD}(A_0),$ then everyone takes action $1$ under any equilibrium outcome from we argued earlier. Thus,
\begin{align*}
    \inf_{\sigma \in \Sigma(\bm{\mu}^*, A_0)}
\Ex^{\sigma}\Big[\phi(\bm{A}) \Big] \geq (1-p^*(\mu_0,A_0) + \eta) \phi(\bm{\underline{A}}) + (p^*(\mu_0,A_0)- \eta)\phi(\bm{\bar{A}}).
\end{align*}
Taking limit $\eta \to 0$, we obtain
\begin{align*}
    \eqref{eqn:adv} =  \sup_{\bm{\mu} \in \mathcal{M}} \inf_{\sigma \in \Sigma(\bm{\mu}, A_0)}
\Ex^{\sigma}\Big[\phi(\bm{A})\Big] \geq  (1-p^*(\mu_0,A_0)) \phi(\bm{\underline{A}}) + p^*(\mu_0,A_0)\phi(\bm{\bar{A}}) \geq \eqref{eqn:opt}.
\end{align*}
Since $\eqref{eqn:opt} \geq  \eqref{eqn:adv}$, we obtain $\eqref{eqn:opt} =  \eqref{eqn:adv}$.

\noindent \underline{\textbf{Step 3B.}} We finally show $\bm{\mu^*}$ is sequentially optimal when $\mu^*_t \notin \Psi_{LD}(A_t) \cup \text{Bd}_{\theta^*} \Delta(\Theta).$ We proceed casewise: 
\begin{itemize}[leftmargin=*]
    \item \textbf{Case 1:} If $\mu_{t-} \notin \Psi_{LD}(A_t)$ and $|A_t - Z_{t-}| < \mathsf{TOL}(D(\mu_{t-},A_t))$. In this case, there is no information arriving, and everyone takes action 1. This will increase $A_t$, and every agent always takes action $1$ from time $t$ onwards. This is the best outcome for the designer, implying sequential optimality.
    \item \textbf{Case 2:} If $\mu_{t-} \notin \Psi_{LD}(A_t)$ and $|A_t - Z_{t-}| \geq \mathsf{TOL}(D(\mu_{t-},A_t))$. In this case, the belief moves to either $\mu_{t-} + (M \cdot \mathsf{TOL}(D))\cdot \hat{d}(\mu_{t-})$ or $\mu_{t-} - \mathsf{DOWN}(D) \cdot \hat{d}(\mu_{t-})$. Note that $\mu_{t-} - \mathsf{DOWN}(D) \cdot \hat{d}(\mu_{t-}) \notin \Psi_{LD}(A_t)$ because $\psi_{\hat{\bm d}}(\mu_t - \mathsf{DOWN}(D) \cdot \hat{d}(\mu_{t-})) = (1 + \bar{\alpha}_{\hat{\bm d}}(A_t))/2 > \bar{\alpha}_{\hat{\bm d}}(A_t)$. So no matter what information arrives, every agent takes action 1. This will increase $A_t$, and every agent always takes action $1$ after time $t$. Again, this is the best outcome for the designer, implying sequential optimality.
\end{itemize}
\end{proof}

\clearpage

\section{Examples of games nested in our environment} \label{appendix:examples}

\subsection{Dynamic regime changes games} Regime change games have been extensively studied in the literature on coordination \citep{morris1998unique,angeletos2007dynamic,basak2024panics} among many others. We show how our model can nest a variant of such games. 

\paragraph{Model} Let there be a total order over $\Theta$. $a_{it} = 0$ corresponds to attacking the regime; $a_{it} = 1$ corresponds to not. The hazard rate at which the regime fails is given by $\gamma(A,\theta) \geq 0$ which is strictly decreasing in $A$ (measure of people \emph{not} attacking) and strictly decreasing in $\theta$ (the strength of the regime). If the attack succeeds, the one-time benefit pays out a lump sum normalized to $1$, and the flow cost of attacking is normalized to $c$. Let $\rho$ denote the random time the regime has failed. Observe that---as is standard in regime change games---player $i$ cannot influence the distribution of $\rho$. If the regime has not yet failed, flow payoffs at time $t$ are:  
    \[
    u(1,A,\theta, t) = 0 \quad u(0,A,\theta, t) = \underbrace{\mathbbm{I}(a_t = 0, \rho = t)}_{\text{Attacking as regime fails}} - c, \quad c > 0
    \]
    On the other hand, if the regime has failed, payoffs are
    \[
    u(1,A,\theta, t) = 0 \quad u(0,A,\theta, t) = - c
    \]
    It suffices to consider only histories at which the regime has not yet failed. 
    For convenience, we define the ``conditional" pdf of $\rho$, the failure time, as follows:\footnote{If a nonnegative random variable $\rho$ with cdf $F$ and pdf $f$ has the hazard rate function $\gamma(t) \coloneqq \frac{f(t)}{1-F(t)}, t \geq 0$, then we can write $F$ and $f$ as functions of $\rho$ as follows:
    \[F(t) = 1-e^{-\int_0^t \gamma(s)ds}, \quad f(t) = \gamma(t)e^{-\int_0^t \gamma(s)ds}, \] for every $t \geq 0.$}
    \begin{align*}
    f_{\gamma,\theta}(\bm{A},s| t) \coloneqq \gamma (A_s,\theta) e^{-\int_t^s \gamma(A_{s'},\theta) ds'}.
    \end{align*}
    Now observe that expected payoff difference from taking $1$ (not attacking) versus $0$ (attacking) between $t$ and the next tick $t + \tau$ conditional on history $H_t$ is 
\begin{align*}
    &- \Ex_{\tau,\theta,\bm{A}} \bigg[ \underbrace{\left(1 - \int_t^{t + \tau} f_{\gamma, \theta}(\bm{A},s |t)ds\right) \int_t^{t + \tau} e^{-r(s - t)} (-c) ds}_{\text{The regime survives until $t + \tau$}} \\
    & \quad + \int_t^{t + \tau} f_{\gamma, \theta}(\bm{A},s | t) \underbrace{\left\{ e^{-r(s-t)} - \int_t^{t + \tau} e^{-r(v-t)} c dv \right\}}_{\text{The regime fails at $s \in [t, t + \tau]$}} ds \bigg\lvert H_t\bigg] \\
    &=\Ex_{\tau,\theta,\bm{A}} \bigg[ \underbrace{\int_t^{t + \tau} e^{-r(s-t)} \Big( c - f_{\gamma, \theta}(\bm{A},s | t) \Big) ds}_{=: \Delta  U(\bm{A}, \theta, \tau | t)}   \bigg\lvert H_t\bigg] 
\end{align*}
so before the regime fails, this can be handled within our framework with the caveat that $\Delta U(\bm{A},\theta,\tau | t)$ depends also on time $t$, the last time an agent's clock ticked, and the path of aggregate actions $\bm{A}$. We will assume that $\Delta U(\bm{A}, \theta, \tau| t)$ satisfies the following conditions which are close analogs of the conditions imposed on payoffs in the main text, but suitably modified to handle non-time separability: 
\begin{enumerate}
    \item[(i)] \textbf{Lipschitz Continuity and Supermodularity}: $\Delta U(\bm{A},\theta,\tau | t)$ is increasing and Lipschitz in $\bm{A}$, i.e., for every $\bm{A},\bm{A'} \in [0,1]^\mathcal{T},$
    \begin{itemize}
        \item if $A_s \geq A'_s$ for every $s \in \mathcal{T}$, then $\Delta U (\bm{A}, \theta, \tau | t) > \Delta U (\bm{A'}, \theta, \tau | t).$
        \item there exists a constant $L^* > 0$ such that $|\Delta U(\bm{A},s|t) - \Delta U(\bm{A'},s|t)| \leq L^* \|\bm{A} - \bm{A'}\|_\infty$
    \end{itemize}    
    A sufficient condition for these is $\gamma$ is continuously differentiable and strictly decreasing in $A.$ The proof is deferred to Lemma \ref{lem: technical regime change} in Online Appendix \ref{online appendix: technical proof}.
    \item[(ii)] \textbf{Dominant State}: $\Delta U (\bm{0}, \bar\theta, \tau | t) > 0$ holds for all $t<\tau \in \mathcal{T}$, where $\bar\theta$ is the maximal element of $\Theta$. A sufficient condition for this is $\gamma (0, \bar\theta) < c$.
\end{enumerate}

A further difference is that, since we have assumed that agents observe the failure of the regime (as in typical dynamic regime change games), observing that the regime has survived conveys information about the state. That is, the designer no longer has complete control over the belief martingale. Nonetheless, we have: 

\begin{claim}
    Theorem \ref{thrm:main} applies to this dynamic regime change environment. 
\end{claim}

\paragraph{Sketch of modification to proof} The argument needs to be modified in three ways. First, we verify that under the non-time separable regime change environment, the boundary of the lower dominance region decreases in $A$. Next, we argue that a time-$0$ jump to the maximum escape belief remains optimal. Finally, we argue that if beliefs are outside the lower dominance region, our information puts policy continues to fully implement action $1$. We will sketch the modification for binary states; an argument analogous to Step 2 of the proof of Theorem \ref{thrm:main} lifts the result to finite states. 

\noindent \underline{\textbf{Modification 1:}} the lower dominance region $\Psi_{LD}$ remains well-defined and $\psi_{LD}$ is still a decreasing function.

We define $\Psi_{LD}(A_t)$ analogously to the main text as follows:
\begin{align*}
\Psi_{LD}(A_t) \coloneqq \Big\{ \mu \in \Delta(\Theta): \Ex_{\theta \sim \mu} [\Delta U(\bm{\bar{A}}, \theta, \tau| t)] \leq 0 \Big\},
\end{align*}
where $\bm{\bar{A}} := (\bar A_s)_{s \geq t}$ with $\bar A_s$ solving $d \bar{A}_s = \lambda (1- \bar A_s) ds$ for $s \geq t$ and $\bar A_t = A_t$. To see that $\Psi_{LD}(A_t)$ is well-defined, observe that if agent $i$'s clock ticks at time $t$, her expected payoff difference is
\begin{align*}
    \Ex_{\tau,\theta,\bm{A}} \bigg[ \underbrace{\int_t^{t + \tau} e^{-r(s-t)} \Big( c - f_{\gamma, \theta}(\bm{A},s | t) \Big) ds}_{= \Delta U (\bm A, \theta, \tau | t)}  \bigg] = \Ex_{\tau,\theta,\bm{A}} \bigg[ \underbrace{\int_0^{\tau} e^{-rs} \Big( c - f_{\gamma, \theta}(\bm{A},s | 0) \Big) ds}_{= \Delta U (\bm{A_{-t}}, \theta, \tau | 0)}  \bigg],
\end{align*}
where $\bm A = (A_v)_{v \geq t}$ and $\bm{A_{-t}} = (A_{v-t})_{v \geq t}$. Hence, the calendar time affects the expected payoff only through future aggregate plays, which implies that the lower dominance region $\Psi_{LD}$ depends only on the time-$t$ aggregate action $A_t$. Note also that $\psi_{LD}$ is decreasing in $A$ because of the supermodularity condition.

Lemma \ref{lem: lipschitz} still applies to the regime change environment because $\frac{\partial \Delta  \bar U}{\partial A}$ is still bounded above using the fact that 1) $\Delta U(\bm{A},s|t)$ is Lipschitz continuous in $A$ and 2) $\|\bm{\bar{A}} - \bm{\bar{A'}} \|_\infty \leq |A_0-A'_0|$ because $|\bar{A}_t-\bar{A'_t}| \leq A_0-A'_0$ for every $t \in \mathcal{T}$.

\noindent \textbf{\underline{Modification 2:}} time-$0$ jump to the maximum escape belief remains optimal.

\begin{figure}[h]
\centering
\caption{Escaping $\Psi_{LD}$ in dynamic regime change games}
    \subfloat[No information]{\includegraphics[width=0.25\textwidth]{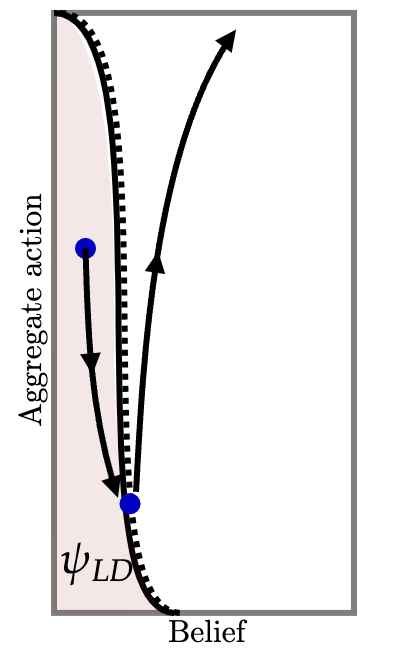}}
    \subfloat[Immediate injection]{\includegraphics[width=0.25\textwidth]{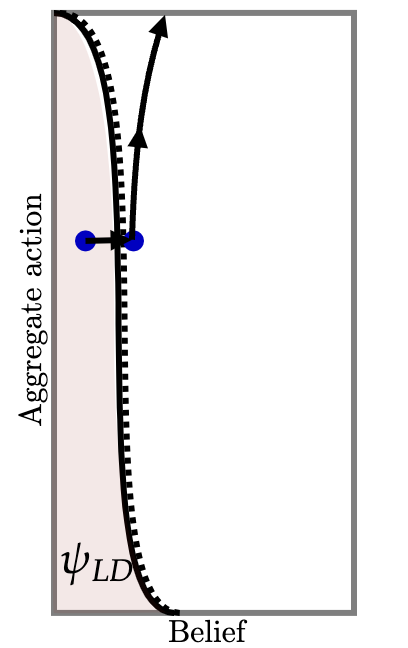}}
    \label{fig:regime_change}
\end{figure}

Step 3A in the proof of Theorem \ref{thrm:main} shows, if the designer can choose any belief martingale for agents, his optimal value is bounded by his utility from an information structure that has time-0 jump: $\mu_{0+} \in \{0,\psi_{LD}(A_0)\}$ {under a conjecture that all agents choose action 1 at all time if $\mu_{0+} = \psi_{LD}(A_0)$.} This time-0 belief jump remains implementable in this environment via the same informational puts policy in Theorem \ref{thrm:main}. It remains to verify that under informational puts, players must in equilibrium believe that all future agents will choose action 1 for all subsequent times if $\mu_{0+} = \psi_{LD}(A_0) + \eta$, as shown in the next modification:

\noindent \textbf{\underline{Modification 3:}} information puts continue to fully implement action $1$.  First, we argue that, under the conjecture all agents choose action $1$ at $t$ if $\mu_t > \psi^n(A_t)$ for some decreasing function $\psi^n$, they also prefer to do so every time after $t$ under the informational puts policy. When $\mu_t> \psi^n(A_t),$ not only $A_t$ increases as more people take action 1 but $\mu_t$ also increases as long as the regime has not yet failed. Since $\psi^n$ is a decreasing function, $\mu_s > \psi^n(A_s)$ when $s \geq t$ as long as the regime has not failed yet, so they still prefer to take action $1$ at every time after $t$, as desired.


Given these observations, the proofs in Appendix~\ref{appendix:proofs} are still valid by replacing $L = L^*\lambda$.\footnote{In the proof of Theorem~\ref{thrm:main}, we use the following bound of the payoff difference:
\begin{align*}
\Ex_{\tau,\theta}\Big[ \int_{s=t}^\tau e^{-r(s-t)} |\Delta u(A_s,\theta) - \Delta u(A'_s,\theta) | ds  \Big] \leq \frac{L}{\lambda}\|A-A'\|_\infty, 
\end{align*}
so we can replace $\frac{L}{\lambda}$ by $L^*$.} In particular, if $\mu_t \geq \psi^n(A_t) - M\cdot\mathsf{TOL}(D(\mu_t,A_t))/2$, we have
\begin{align*}
    &\Ex_{\tau,\theta,\bm{A}} \bigg[ \Delta U (\bm A, \theta, \tau | t) \bigg] \\
    &\geq \Ex_{\theta, \tau}\bigg[\Delta U (\bm{\bar{A}}, \theta, \tau | t)\bigg]  - \frac{\mathsf{TOL}(\psi^n(A_t) - \psi_{LD}(A_t))L(1 + p_n) + L(1-p_n)  }{\lambda} \\
    &> \frac{C}{2}(\psi^n(A_t) - \psi_{LD}(A_t)) - \frac{\mathsf{TOL}(\psi^n(A_t) - \psi_{LD}(A_t))L(1 + p_n) + L(1-p_n)  }{\lambda}
\end{align*}
for some $C > 0$ as shown in \eqref{eq: lower bound} and \eqref{ineq: first term} because
\[
\min_{A_0 \in [0,1]} \Ex_{\tau}[\Delta U (\bm{A},\bar{\theta},\tau) ] > 0.
\]
since  $\Delta U(\bm{A},\bar{\theta},\tau) > 0$ for every $A_0 \in [0,1]$. Then, the remainder of the proof of Theorem~\ref{thrm:main} proceeds unchanged.

\subsection{Stopping games} Stopping games in coordination environments have been extensively studied \citep{gale1995dynamic,dasgupta2007coordination,basak2024panics} among many others. We show how our model can nest a variant of such games with frictions in stopping opportunities.

\paragraph{Model} Suppose that $A_0 = 0$. Agents decide whether to irreversibly invest ($a = 1$) or not $(a = 0)$ and investment opportunities arrive at the ticks of a personal Poisson clock. Different from the main text, investment is irreversible; that is, if player $i$ has ever played action $1$, she continues to do so forever. The flow payoff from investing relative to not is, as in the main text, $\Delta(A,\theta)$ which is strictly increasing in $A$. We maintain the dominant state assumption as in the main text. Then, suppose that player $i$ has not yet invested and her clock ticks at time $t$. Her payoff from investing relative to not is:  
\[
\Ex\bigg[ \int^{+\infty}_t e^{-r(s-t)} u(1,A_s,\theta) ds\bigg] - \Ex\bigg[ \int^{t + \tau}_t e^{-r(s-t)} u(0,A_s,\theta) ds + e^{-r(t+\tau)}V(A_{t + \tau}, \theta)\bigg] 
\]
where $V(A_{t + \tau}, \theta)$ is the value function at time $t + \tau$ conditional on not investing. This can be rewritten as 
\begin{align*}
    &\Delta  U (\bm A, \mu_t) \\
    &:= \Ex\bigg[ \int^{t + \tau}_t e^{-r(s-t)} \Delta u (A_s,\theta) ds\bigg] + \underbrace{\Ex\bigg[ \int_{t + \tau}^{+\infty} e^{-r(s-t)} u(1,A_s,\theta) ds - e^{-r(t+\tau)}V(A_{t + \tau}, \theta)\bigg]}_{\leq 0}
\end{align*}
where, unlike the main text, the second term is not necessarily zero and may be strictly negative. Nonetheless, the \emph{same} policy in the main text can fully implement the equilibrium in which all players stop and irreversibly take action $1$ at the first opportunity: 

\begin{claim}
    Theorem \ref{thrm:main} applies to this stopping game. 
\end{claim}

\paragraph{Sketch of modification to proof} We require a small modification to the proof of Theorem \ref{thrm:main} which we now outline. To apply Step 1 of the proof of Theorem~\ref{thrm:main}, suppose that $\Theta = \{0, 1\}$ and $\theta^* = 1$. First, observe that if $\mu_t \in \Psi_{LD}(A_t)$, action $0$ is strictly dominant because
\[
\Delta U (\bm A, \mu_t) \leq \Ex\bigg[ \int^{t + \tau}_t e^{-r(s-t)} \Delta u (A_s,\theta) ds\bigg] \leq 0.
\]
Second, to initiate contagion (Step 1A), we define the upper dominance region as follows:
\begin{align*}
    \Psi_{UD}(A_t) \coloneqq  \Big\{ \mu \in \Delta(\Theta): \Delta U (\underline{\bm A}, \mu) \geq 0 \Big\},
\end{align*}
which is nonempty because $\theta = 1$ is the dominant state. Now we follow Step 1B to find the lower bound on $\Delta U$. At $T^*$, the designer injects new information. If a realized signal is good, the belief exceeds $\psi^n(A_{T^*})$, which implies that the gap term is zero. If a signal is bad, we may have a negative gap term. Hence, \eqref{eq: lower bound} can be modified as
\begin{align*}
    &\Delta U (\bm A, \mu_t) \\
    &\geq \Ex_\tau\bigg[\int_{s=t}^{t + \tau} e^{-r(s-t)} \Delta u(\bar{A}_s,\theta) ds\bigg]  - \frac{\mathsf{TOL}(\psi^n(A_t) - \psi_{LD}(A_t))L(1 + p_n) + (L + W)(1-p_n)  }{\lambda},
\end{align*}
where $-W$ is the lower bound of the gap term, e.g., $W = \bar\Delta/r$ with $\bar\Delta := \max_{A,\theta}|\Delta u(A,\theta)|$. By modifying $\mathsf{TOL}$ as
\[
\mathsf{TOL}(D) = \bar\delta \cdot \frac{\lambda C D}{4(L + W) (1 + M D^{-1})},
\]
we can prove Lemma~\ref{lem: contagion} as Appendix~\ref{appendix:proofs}. The remainder of the proof remains unchanged.


\clearpage 

\titleformat{\section}
		{\normalsize\bfseries\center\scshape}     
         {Online Appendix \thesection:}
        {0.5em}
        {}
        []
\renewcommand{\thesection}{\Roman{section}}

\begin{center}
    \large{\textbf{ONLINE APPENDIX TO `INFORMATIONAL PUTS'}} \\
    \small{ANDREW KOH \quad SIVAKORN SANGUANMOO \quad KEI UZUI}
\end{center}

\setcounter{page}{1}
\setcounter{section}{0}
Online Appendix \ref{appendix:finite_players} develops Theorem \ref{thrm:main} for finite players in which the designer, as in the main text, can condition information only on aggregate rather than individual play. Online Appendix \ref{appendix:private} analyzes whether the designer can do better under private information. Online Appendix \ref{online appendix: technical proof} collects technical results on the dynamic regime change games analyzed in Appendix \ref{appendix:examples}. 

\section{Finite players} \label{appendix:finite_players}
\subsection{Finite Model} We describe the following modification to our model in the main text: instead of a unit measure of agents, there are now a finite number of $N < +\infty$ agents. As before, the action space is $\{0,1\}$ and $A_t$ is the proportion of agents playing action $1$, and so on and $A_0 = \frac{N-n}{N}$ is the proportion of initial play, where $n$ is the number of agents who initially play action $0$ (a primitive of our environment). Without loss, we will index agents such that the first $n$ players have initial play $0$, and players $n+1$ through $N$ have initial play $1$. 

\noindent \textbf{Switching frictions with finite players.} As in the main text, players have personal Poisson clocks which tick at independent rate $\lambda > 0$. Hence, a crucial difference from the continuum case is that even if all agents play the strategy `play action $1$ at the next tick of my clock', the path of aggregate play is random. 

\noindent \textbf{States and payoffs.} 
For simplicity, we will specialize the main text to the case that the state is binary $\Theta = \{0,1\}$ with dominant state $\theta^* = 1$. As in the proof of Theorem \ref{thrm:main} (see Step 2 in Appendix \ref{appendix:proofs}) it is straightforward to extend the result and proof to finite states. Flow payoffs $u(a,A,\theta)$ as almost identical to those in the continuum case. We assume
\begin{itemize}
    \item[(i)] \textbf{Supermodularity.} $\Delta u(A,\theta)$ is strictly increasing in $A$ and $u(1,\cdot,\theta)$, $u(0,\cdot,\theta)$ are continuously differentiable. 
    \item[(ii)] \textbf{Dominant state.} $\Delta u(0,1) > 0$. 
\end{itemize}

\noindent  \textbf{Dynamic information.} Our definition of dynamic information is identical to the main text: information are C\`adl\`ag martingales w.r.t. the the natural filtration generated by \emph{aggregate play} $(A_t)_t$ and \emph{past beliefs} $(\mu_t)_t$.

\noindent \textbf{Sequence of economies.} We will consider a sequence of economies parameterized by $N$, the total number of agents. For an economy with $N$ agents, let $\Sigma^N(\bm{\mu},A_0)$ denote the set of subgame perfect equilibria of the stochastic game induced by a belief martingale $\bm{\mu}$ under the economy consisting of $N$ agents whenever $A_0$ can be written as $\frac{k}{N}$ for some $k \in \{0,\dots,N\}$.

\noindent \textbf{Designer's problem.} 
As before, the designer's payoff is an increasing functional $\phi(\bm{A})$. We will assume there exists a constant $L_\phi$ such that $|\phi(\bm{A}) - \phi(\bm{A}')| \leq L_\phi \|\bm{A}-\bm{A}' \|_\infty$ for every $\bm{A},\bm{A}' \in [0,1]^\infty$. The designer's problems under optimal and adversarial equilibrium selections with a finite number of agents as follows: 
\begin{align*}
\sup_{\bm{\mu} \in \mathcal{M}} \inf_{\bm{\sigma} \in \Sigma^N(\bm{\mu},A_0)} \Ex^{\sigma} \Big[\phi\big(\bm{A}\big)\Big] \tag{ADV-$N$} \label{eqn:adv-n} \\
 \sup_{\bm{\mu} \in \mathcal{M}} \sup_{\bm{\sigma} \in \Sigma^N(\bm{\mu},A_0)} \Ex^{\sigma} \Big[\phi\big(\bm{A}\big)\Big] \tag{OPT-$N$} \label{eqn:opt-n}
\end{align*}

\noindent \textbf{Optimal information.} We will use the \emph{exact} same policy constructed in the main text which was deployed in Theorem \ref{thrm:main}. Recall we parameterized that family of policies by $(\bm{\mu}^{\eta})_{\eta > 0}$ which we called \emph{informational puts}.  

\subsection{Main result}  We are ready to state the finite player analog of Theorem \ref{thrm:main} when the prior is not in the lower dominance region. Proposition~\ref{prop: finite LD} in subsection \ref{appendix:finite_in_LD} analyzes the case when the prior is within the lower dominance region.

    \begin{theorem} \label{thrm: finite} Suppose the prior is not in the lower dominance region, i.e, $\mu_0 > \psi_{LD}(A_0)$. Then, the informational puts policy $\bm{\mu}^{\eta}$ uniquely implements the equilibrium in which all players play action $1$. Asymptotically as $N \to +\infty$ the policy is sequentially optimal, and there is no multiplicity gap.
        \begin{enumerate}
        \item[(i)] There exists a constant $G$ such that, under any subgame perfect equilibrium $\sigma \in \Sigma^N(\bm{\mu}^{\eta}, A_0)$, agents play action $1$ if $\mu_t > \psi_{LD}(A_t) + (GN)^{-1/9}$ for every history $H_t$, aggregate action $A_t$, and belief $\mu_t$, 
        \item[(ii)] There exist constants $K_1$ and $K_2$ such that, for any $(\mu_0, A_0)$\footnote{We implicitly assume $A_0 \in \mathbb{Q}$ and $N \cdot A_0 $ is an integer.}, we have
        \begin{align*}
            |\eqref{eqn:adv-n} - \eqref{eqn:opt-n}| &\leq K_1N^{-1/9}, \\
            | \eqref{eqn:adv-n} - \eqref{eqn:adv}|  &\leq K_2N^{-1/9},
        \end{align*}
for sufficiently large $N$ (depending on $(\mu_0,A_0)$).

        \item[(iii)] Sequential optimality: 
        \[
      \lim_{N \to \infty}\Bigg| \inf_{\bm{\sigma} \in \Sigma^N(\bm{\mu}^{\eta},A_0)} \Ex^{\sigma} \big[\phi\big(\bm{A}\big) \big| \mathcal{F}_t\big]  - 
\sup_{\bm{\mu}' \in \mathcal{M}} \inf_{\bm{\sigma}^N \in \Sigma(\bm{\mu}',A_0)} \Ex^{\sigma} \big[\phi(\bm{A})\big| \mathcal{F}_t\big] \Bigg| = 0.
    \]
        \end{enumerate}
    \end{theorem}

\subsection{Outline of proof of Theorem \ref{thrm: finite} and preliminaries} The proof of Theorem \ref{thrm: finite} consists of the following steps, as described in Figure~\ref{fig:map_finite}.

\begin{figure}[h!]  
\centering
\captionsetup{width=1.0\linewidth}
    \caption{Roadmap for proof of Theorem \ref{thrm: finite}} \includegraphics[width=1\textwidth]{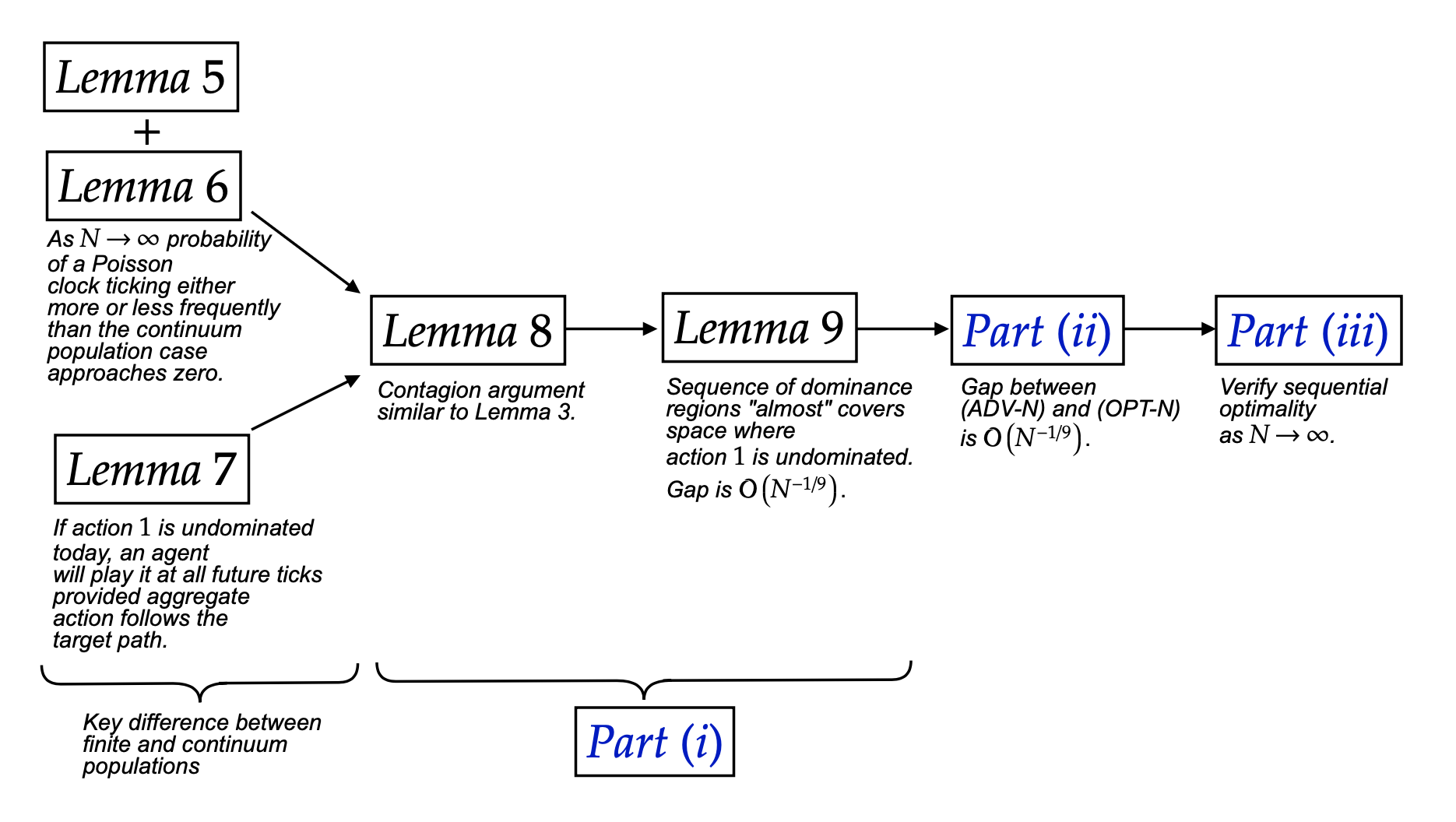}
    \label{fig:map_finite}
\end{figure}

It will be helpful to define the random path of aggregate play. For each $i \in \{1,\dots,n\}$, define $\tau_i \sim \text{Exp}(\lambda)$ as an iid exponential distribution with rate $\lambda$, that is, $\tau_i$ is agent $i$'s random waiting time for the first switching opportunity. We define paths  $(\bar A_t^N)_t$ and $(\bar{A}_t)_t$ as follows:
\begin{align*}
    \bar A_t^N &= A_0 + \frac{1}{N}\sum_{i=1}^n 1\{\tau_i \leq t\} \\
    \bar{A}_t &= 1- (1-A_0) e^{-\lambda t} 
\end{align*}
Note that $\bar A_t^N$ is the \emph{random} proportion of agents playing action $1$ at time $t$ when everyone switches to action $1$ as quickly as possible i.e., at the random ticks of their individual clocks, while $\bar{A}_t$ is non-random: it is simply the proportion of agents playing action $1$ at time $t$ when $1 - e^{-\lambda t}$ of the agents initially playing action $0$ have switched to action $1$ by time $t$. Our definition of $\bar{A}_t$ coincides exactly with the continuum case.

If the number of agents is finite, the proportion of agents playing action $1$ can deviate from the tolerated distance from the target even when no one has switched to action $0$. Lemmas~\ref{lem: unlucky bound} and \ref{lem: unlucky bound_2} provide an upper bound on the probability of such ``unlucky'' events:
\[
\Pr(\forall t, | \bar{A}_t^N - \bar{A}_t| \leq  \delta ) > 1 - (\bar c + \underline c)\delta^{-4}N^{-1}.
\]

\begin{lemma}\label{lem: unlucky bound}
    There exists an  absolute constant $\bar{c}>0$ such that, for every $\delta >0$ and $N \in \mathbb{N}$, $\Pr(\forall t, \bar{A}_t^N +  \delta \geq \bar{A}_t) > 1 - \bar{c}\delta^{-4}N^{-1}$.
\end{lemma}

\begin{proof}

Fix $\alpha$ such that $\delta = 2N^{-\alpha}$.
    We rearrange $(\tau_i)_{1,\dots,n}$ as $\tau_{(1)} < \tau_{(2)} < \cdots < \tau_{(n)}.$ For each $k \in \{0,\dots, \lceil nN^{\alpha-1} \rceil - 1\}$, define $T_k \coloneqq [ \tau_{(kN^{1-\alpha})}, \tau_{((k+1)N^{1-\alpha})})$, where $\tau_{(i)} = \tau_{(\lfloor i \rfloor )}$, $\tau_{(0)} = 0$, and $\tau_{(n+1)} = \infty$. If $t \in T_k$, We must have $\bar{A}_t^N \in [A_0 +     \frac{\lfloor kN^{1-\alpha} \rfloor}{N}, A_0 + (k+1)N^{-\alpha}].$ Therefore,
\begin{align*}
    &\Pr(\forall t \leq T, \bar{A}_t^N \geq \bar{A}_t - \delta) \\
    &= \Pr\bigg( \bigcap_{k=0}^{\lceil nN^{\alpha-1} \rceil - 1} \{\omega: \forall t \in T_k, \bar{A}_t^N \geq \bar{A}_t - \delta\} \bigg)  \\
    &\geq \Pr\bigg( \bigcap_{k=0}^{\lceil nN^{\alpha-1} \rceil - 1 }\bigg\{\omega: \forall t \in T_k, A_0 +     \frac{\lfloor kN^{1-\alpha} \rfloor}{N}  \geq \bar{A}_t - \delta \bigg\} \bigg)\tag{$\bar{A}_t^N \geq A_0 + \frac{\lfloor kN^{1-\alpha} \rfloor}{N}$}  \\
    &\geq \Pr\bigg( \bigcap_{k=0}^{\lceil nN^{\alpha-1} \rceil - 1} \bigg\{\omega:  A_0 +     \frac{kN^{1-\alpha} -1}{N}  \geq \bar{A}_{\tau_{((k+1)N^{1-\alpha})}} - \delta \bigg\} \bigg) \\
    &= \Pr\bigg( \bigcap_{k=0}^{\lceil nN^{1-\alpha} \rceil - 2} \bigg\{\omega:  A_0 +     \frac{kN^{1-\alpha} -1}{N}  \geq \bar{A}_{\tau_{((k+1)N^{1-\alpha})}} - \delta \bigg\} \bigg),
\end{align*}
    where the last equality follows from that if $k = \lceil nN^{\alpha - 1} \rceil - 1$ then 
    \begin{align*}
        A_0 + \frac{kN^{1-\alpha } -1}{N} &> \frac{N-n}{N} + \frac{n - N^{1-\alpha}-1}{N} \tag{$k > nN^{1-\alpha} - 1$}\\
        &= 1 - \frac{N^{1-\alpha}+1}{N} \\
        &> 1- \delta \tag{$N^{1-\alpha} > 1$}\\
        &> \bar{A}_{\tau_{((k+1)N^{1-\alpha})}} - \delta .
    \end{align*}   
    Note that $N^{1-\alpha} > 1$ holds because we later focus on $\delta > 18/N$, which is equivalent to $N^{1-\alpha} > 9$. We define the event $\Omega_{relax}$ as follows:
    \begin{align*}
        \Omega_{relax} = \bigcap_{k=0}^{\lceil nN^{\alpha - 1}\rceil -2}  \underbrace{\bigg\{ \tau_{((k+1)N^{1-\alpha})} - \tau_{(kN^{1-\alpha})} \leq \lambda^{-1}(1+\delta/4)\log \bigg(\frac{ \lfloor n-kN^{1-\alpha} \rfloor}{\lfloor n-(k+1)N^{1-\alpha} \rfloor} \bigg) \bigg\}}_{\Omega_k}.
    \end{align*}
    Under the event $\Omega_{relax}$, for every $k \leq \lceil nN^{\alpha-1}\rceil -2$, we have
    \begin{align*}
        \bar{A}_{\tau_{((k+1)N^{1-\alpha})}}  &= 1-\frac{n}{N}\exp(-\lambda \tau_{((k+1)N^{1-\alpha})}) \\
        &= 1-\frac{n}{N}  \exp\bigg( -\lambda \sum_{i=0}^k (\tau_{((i+1) N^{1-\alpha})}-\tau_{(i N^{1-\alpha})} ) \bigg) \\
        &\leq 1-\frac{n}{N} \exp \bigg(-(1+\delta/4)(\log n - \log \big( \lfloor n- (k+1)N^{1-\alpha}\rfloor \big) \bigg) \tag{From $\Omega_l$ for $ l=0, \dots,k$}\\
        &= 1- \frac{n}{N} \bigg(1-\frac{ \lceil(k+1)N^{1-\alpha} \rceil}{n} \bigg)^{1+\delta/4} \\
        &\leq 1 -\frac{n}{N}\bigg(1 - \frac{(1+\delta/4)\big((1+k)N^{1-\alpha}+1\big)}{n}  \bigg) \\
        &= A_0 + (1+\delta/4)((1+k) N^{-\alpha} + N^{-1}) \tag{$A_0 = \frac{N-n}{N}$}
    \end{align*}
    Note that $k < N^{\alpha}$. Thus, if $18/N < \delta < 1 ,$
    \begin{align*}
        \bar{A}_{\tau_{((k+1)N^{1-\alpha})}} - \delta &\leq A_0 + (1+\delta/4)((1+k) N^{-\alpha}+N^{-1}) - \delta \tag{From the above inequality}\\
        &= A_0+ kN^{-\alpha} + (1+\delta/4 + \delta k/4) N^{-\alpha} +(1+\delta/4)N^{-1} - \delta \\
        &\leq A_0 + kN^{-\alpha} + (1+\delta/4)(N^{-\alpha}+N^{-1}) - 3\delta/4 \tag{$kN^{-\alpha} < 1$}\\
        &= A_0 + kN^{-\alpha} + (1+\delta/4)(\delta/2+N^{-1}) - 3\delta/4 \\
        &\leq A_0 +kN^{-\alpha} - N^{-1}, \tag{$\delta \in (18/N,1)$}
    \end{align*}
    Note that, the lemma statement is trivial when $\delta > 1$. Also, if $\delta < 18/N$, we can set $\bar{c}>18^4$ so that $\bar{c} \delta^{-4}N^{-1} > \bar{c}N^3/18^4 \geq 1$, which makes the lemma statement trivial.
    
    This implies
    \begin{align*}
        \Omega_{relax} \subset \bigcap_{k=0}^{\lceil nN^{1-\alpha} \rceil - 2} \bigg\{\omega:  A_0 +     \frac{kN^{1-\alpha} -1}{N}  \geq \bar{A}_{\tau_{((k+1)N^{1-\alpha})}} - \delta \bigg\}.
    \end{align*}
    Now we compute $\Pr(\Omega_{relax})$. Note that $\tau_{((k+1)N^{1-\alpha})} - \tau_{(kN^{1-\alpha})}$ has 
    \begin{align*}
        \text{mean} &= \sum_{i=\lfloor kN^{1-\alpha} \rfloor}^{\lfloor (k+1)N^{1-\alpha} \rfloor - 1}  \frac{1}{\lambda(n-i)}  \leq \frac{1}{\lambda}\log \bigg(\frac{\lfloor n - kN^{1-\alpha} \rfloor}{\lfloor n - (k+1)N^{1-\alpha} \rfloor} \bigg)  \\
        \text{variance} &= \sum_{i=\lfloor kN^{1-\alpha} \rfloor}^{\lfloor (k+1)N^{1-\alpha} \rfloor} \frac{1}{\lambda^2(n-i)^2} \leq  \frac{1}{\lambda^2}\bigg(\frac{1}{\lfloor n-kN^{1-\alpha} \rfloor} - \frac{1}{\lfloor n-(k+1)N^{1-\alpha} \rfloor} \bigg).
    \end{align*}
Let $a_k = \lfloor n-kN^{1-\alpha} \rfloor$. Thus, by Chebyshev inequality, the probability of $\Omega_k^c$ is bounded above by
\begin{align*}
    \frac{1/a_k - 1/a_{k+1}}{(\delta/4)^2  (\log a_{k+1} - \log a_k)^2 } &= \bigg(\frac{1/a_k - 1/a_{k+1}}{\log a_{k+1} - \log a_k}\bigg)^2 \cdot \frac{16/\delta^2 }{1/a_k-1/a_{k+1}} \\
    &\leq \frac{16}{\delta^2 a_k^2} \cdot \frac{1}{1/a_k-1/a_{k+1}} \tag{$x-\log x \geq 1$ for all $x>0$}\\
    &= \frac{16}{\delta^2} \bigg(\frac{1}{a_k} + \frac{1}{a_{k+1}-a_k}\bigg) \\
    &< \frac{16}{\delta^2} \cdot 4N^{\alpha-1} \tag{$a_k,a_{k+1}-a_k > \lfloor N^{1-\alpha} \rfloor > \frac{1}{2}N^{1-\alpha}$}
\end{align*}
for every $k \leq \lceil nN^{\alpha-1} \rceil -2$. Thus, $\Pr(\Omega_{relax}^c)$ is bounded above by
    \begin{align*}
        \lceil nN^{\alpha-1}\rceil \cdot \frac{64N^{\alpha-1}}{\delta^2} \leq  64N^{2\alpha-1}\delta^{-2} = \underbrace{256}_{=: \bar c} \delta^{-4} N^{-1}
    \end{align*}
since $\delta = 2N^{-\alpha}$, as desired. 
\end{proof}

\begin{lemma}\label{lem: unlucky bound_2}
    There exists an  absolute constant $\underline{c}>0$ such that, for every $\delta >0$ and $N \in \mathbb{N}$,  $\Pr(\forall t, \bar{A}_t^N -  \delta \leq \bar{A}_t) > 1 - \underline{c}\delta^{-4}N^{-1}$.
\end{lemma}

\begin{proof} Fix $\alpha$ such that $N^{-\alpha} = 2\delta/3$. We define $(\tau_i)_{1,\dots,n}$ and $(T_k)_k$ in the same way as in the proof Lemma \ref{lem: unlucky bound}. Therefore,
\begin{align*}
    &\Pr(\forall t \leq T, \bar{A}_t^N \leq \bar{A}_t + \delta) \\
    &= \Pr\bigg( \bigcap_{k=0}^{\lceil nN^{\alpha-1} \rceil - 1} \{\omega: \forall t \in T_k, \bar{A}_t^N \leq \bar{A}_t + \delta\} \bigg) \\
    &\geq \Pr\bigg( \bigcap_{k=0}^{\lceil nN^{\alpha-1} \rceil - 1 }\bigg\{\omega: \forall t \in T_k, A_0 +     (k+1)N^{-\alpha}  \leq \bar{A}_t + \delta \bigg\} \bigg) \tag{$\bar{A}_t^N \leq A_0 + (k+1)N^{-\alpha}$} \\
    &\geq \Pr\bigg( \bigcap_{k=0}^{\lceil nN^{\alpha-1} \rceil - 1} \bigg\{\omega:  A_0 +    (k+1)N^{-\alpha}  \leq \bar{A}_{\tau_{(kN^{1-\alpha})}} + \delta \bigg\} \bigg) \tag{$\bar A_t \geq \bar{A}_{\tau_{(kN^{1-\alpha})}}$ for $t \in T_k$}
\end{align*}
    We define the event $\Omega_{relax}$ as follows:
    \begin{align*}
        \Omega_{relax} = \bigcap_{k=0}^{\lceil nN^{\alpha - 1}\rceil -2}  \underbrace{\bigg\{ \tau_{((k+1)N^{1-\alpha})} - \tau_{(kN^{1-\alpha})} \geq \lambda^{-1}(1-\delta/3)\log \bigg(\frac{\lfloor n-kN^{1-\alpha} \rfloor}{\lfloor n-(k+1)N^{1-\alpha} \rfloor} \bigg) \bigg\}}_{\Omega_k}.
    \end{align*}
    Under the event $\Omega_{relax}$, for every $k \leq \lceil nN^{\alpha-1}\rceil -1$, we have
    \begin{align*}
        \bar{A}_{\tau_{(k N^{1-\alpha})}}  &= 1-\frac{n}{N}\exp(-\lambda \tau_{(k N^{1-\alpha})}) \\
        &= 1-\frac{n}{N}  \exp\bigg( -\lambda \sum_{i=0}^{k-1} (\tau_{((i+1) N^{1-\alpha})}-\tau_{(i N^{1-\alpha})} ) \bigg) \\
        &\geq 1-\frac{n}{N} \exp \bigg((1-\delta/3)(\log n - \log \big(\lfloor n- kN^{1-\alpha}\rfloor \big) \bigg) \tag{From $\Omega_l$ for $ l=0, \dots,k$}\\
        &\geq 1- \frac{n}{N} \bigg(1-\frac{ kN^{1-\alpha}}{n} \bigg)^{1-\delta/3} \\
        &\geq 1 -\frac{n}{N}\bigg(1 - \frac{(1-\delta/3)kN^{1-\alpha}}{n}  \bigg) \\
        &= A_0 + (1-\delta/3)k N^{-\alpha} 
    \end{align*}
    Note that $k < N^{\alpha}$. Thus,
    \begin{align*}
        \bar{A}_{\tau_{(kN^{1-\alpha})}} + \delta &\geq A_0 + (1-\delta/3)k N^{-\alpha} + \delta \tag{From the above inequality}\\
        &= A_0+ (k+1)N^{-\alpha} -(1+\delta k /3) \cdot N^{-\alpha} + \delta \\
        &> A_0 + (k+1)N^{-\alpha} - N^{-\alpha} + 2\delta/3 \tag{$kN^{-\alpha} < 1$}\\
        &= A_0+ (k+1)N^{-\alpha}, \tag{$N^{-\alpha} = 2\delta/3$}
    \end{align*}
    This implies
    \begin{align*}
        \Omega_{relax} \subset \bigcap_{k=0}^{\lceil nN^{1-\alpha} \rceil - 2} \bigg\{\omega:  A_0 +     (k+1)N^{-\alpha}  \leq \bar{A}_{\tau_{(kN^{1-\alpha})}} - \delta \bigg\}.
    \end{align*}
    Now we compute $\Pr(\Omega_{relax})$. As proved in Lemma \ref{lem: unlucky bound}, by Chebyshev inequality, the probability of $\Omega_k^c$ is bounded above by
\begin{align*}
    \frac{1/a_k - 1/(a_{k+1})}{(\delta/3)^2  (\log a_{k+1} - \log a_k)^2 } &\leq \bigg(\frac{1/a_k - 1/a_{k+1}}{\log a_{k+1} - \log a_k}\bigg)^2 \cdot \frac{9/\delta^2 }{1/a_k-1/a_{k+1}} \\
    &\leq \frac{9}{\delta^2 a_k^2} \cdot \frac{1}{1/a_k-1/a_{k+1}}\\
    &= \frac{9}{\delta^2} \bigg(\frac{1}{a_k} + \frac{1}{a_{k+1}-a_k}\bigg) \\
    &< \frac{9}{\delta^2} \cdot 4N^{\alpha-1} \tag{$a_k,a_{k+1}-a_k > \lfloor N^{1-\alpha} \rfloor > \frac{1}{2}N^{1-\alpha}$}
\end{align*}
for every $k \leq \lceil nN^{\alpha-1} \rceil -2$. Thus, $\Pr(\Omega_{relax}^c)$ is bounded above by
    \begin{align*}
        \lceil nN^{\alpha-1}\rceil \cdot  \frac{36N^{\alpha-1}}{\delta^2} \leq 36N^{2\alpha-1}\delta^{-2} = \underbrace{81}_{=: \underline c} \delta^{-4} N^{-1},
    \end{align*}
since $\delta = \frac{3}{2}N^{-\alpha}$, as desired. 
\end{proof}

A subtlety with finite agents is that agent $i$'s action today affects the aggregate path of play and hence her future decision problem. Thus, she must take this into account when choosing her action. The following lemma shows that if $\mu > \psi_{LD}(\bar A_0)$, it is optimal for each agent to take action $1$ when the future path of aggregate actions follows the target $(\bar A_s)_s$.
\begin{lemma}\label{lem: finite}
    For every agent $i$, suppose $(\tau_{in})_n$ be a increasing sequence of Poisson clocks of agent $i$. Suppose $a_{in} \in \{0,1\}$ be a (random) action agent $i$ takes at $\tau_{in}$. If $\mu > \psi_{LD}(\bar{A}_0)$, then
    \begin{align*}
        \Ex_\mu \Big[ \sum_{n=0}^\infty \int_{\tau_{in}}^{\tau_{i,n+1}} e^{-rs} u(a_{in},\bar{A}_s,\theta) ds \Big] \leq  \Ex_\mu \Big[ \int_0^\infty e^{-rs} u(1,\bar{A}_s,\theta) ds \Big].
    \end{align*} 
\end{lemma}
\begin{proof}
    For every $n \in \mathbb{N}$, consider that
    \begin{align*}
        \Ex_\mu \Big[ \int^{\tau_{i,n+1}}_{\tau_{in}} e^{-rs} \Delta u(\bar{A}_s, \theta)ds\Big] &= \Ex_\mu\bigg[ e^{-r \tau_{in}} \Ex_\mu \Big[ \int_{0}^{\tau_{i,n+1}-\tau_{i,n}} e^{-rs} \Delta u(\bar{A}_{s + \tau_{in}}, \theta) ds \big\lvert \tau_{in}\Big]  \bigg] \\
        &\geq \Ex_\mu\bigg[ e^{-r \tau_{in}} \Ex_\mu \Big[ \int_{0}^{\tau_{i,n+1}-\tau_{i,n}} e^{-rs} \Delta u( \bar{A}_{s}, \theta) ds \big\lvert \tau_{in}\Big]  \bigg] \\
        &\geq 0, \tag{$\mu> \psi_{LD}(\bar{A}_0)$}
    \end{align*}
    This implies
    \begin{align*}
        \Ex_\mu \Big[ \int^{\tau_{i,n+1}}_{\tau_{in}} e^{-rs} u(a_{in},\bar{A}_s, \theta)ds\Big] \leq \Ex_\mu \Big[ \int^{\tau_{i,n+1}}_{\tau_{in}} e^{-rs} u(1,\bar{A}_s, \theta)ds\Big]
    \end{align*}
    for every $n \in \mathbb{N}$, as desired.
\end{proof}

\begin{proof}[Proof of Part (i) of Theorem \ref{thrm: finite}]
We follow similar steps as we did in Appendix~\ref{appendix:proofs}. We restate Lemma~\ref{lem: contagion} as follows:
    \begin{lemma}\label{lem: contagion, finite}
        There exists $G > 0$ such that if $\mu_t > \psi_{LD}(A_t) + (GN)^{-1/9}$, $\Psi^n \subset \Psi^{n+1}$ holds for all $n \in \mathbb{N}$.
    \end{lemma}
    \begin{proof}[Proof of Lemma~\ref{lem: contagion, finite}]
        The proof is similar to that of Lemma~\ref{lem: contagion}. Suppose that everyone plays action 1 for any histories $H'$ such that $S(H') = (\mu,A)$ is in the round-$n$ dominance region $\Psi^n.$ To obtain $\Psi^{n+1}$, we derive the lower bound on the expected payoff difference of playing $0$ and $1$ given $\Psi^n$.

        Fix any history $H$ with the current target aggregate action $Z_t$ such that $S(H) = (\mu_t,A_t) \notin \Psi^n$. From our construction of $Z_t$, we have $|Z_{t-}- A_{t-}| < \mathsf{TOL}(D(\mu_{t-},A_{t-}))$. For any path $(A_s)_{s \geq t}$ with an increment at most $\frac{1}{N}$, we define a (deterministic) hitting time $T^*((A_s))_{s \geq t}$ as follows:
        \begin{align*}
            T^* = \inf\{s \geq t: |Z_s -A_s| \geq \mathsf{TOL}(D(\mu_t,A_s)) \text{ or } (\mu_t,A_s) \in \Psi^n\}.
        \end{align*}
        Fix $T^*$, and we first determine a behavior of the path $(A_s)_{s \geq t}$ given $T^*$.
        
\noindent        \textbf{Before time $T^*$.} For any $s \in [t,T^*)$, we showed in \eqref{ineq: delta} that $\mathsf{TOL}(D(\mu_t,A_s)) \leq \mathsf{TOL}(\psi^n(A_t) - \psi_{LD}(A_t))$. By the definition of $Z$, we must have $Z_s = \bar{A}_s$ for every $s \in [t,T^*)$ because $|Z_s - A_s| < \mathsf{TOL}(D(\mu_t,A_s))$ for every $s< T^*$. Then we can write down the lower bound of $A_s$ when $s \in [t,T^*)$ as follows:
        \begin{align}\label{ineq: before T^*}
            A_s \geq \bar{A}_s - \mathsf{TOL}(D(\mu_t,A_s)) \geq \bar{A}_s - \mathsf{TOL}(\psi^n(A_t) - \psi_{LD}(A_t)),
        \end{align}
        almost surely given $T^*$.

   \noindent     \textbf{After time $T^*$.} Fix any $s > T^*$. We consider the following two cases.

\noindent        \textbf{Case 1: $|Z_{T^*} - A_{T^*}| < \mathsf{TOL}(D(\mu_{t},A_{T^*}))$.} This means $\mu_{T^*} = \mu_t$ because no information has been injected until $T^*$. Then the definition of $T^*$ and the right continuity of $Z_s$ and $A_s$ imply $(\mu_{T^*}, A_{T^*}) \in \Psi^n.$ This means every agent strictly prefers to take action 1 at $T^*$. This increases $A_s$, inducing every agent taking action 1 after time $T^*$ until time $s'$ at which information is injected, i.e., $|Z_{s'} - A_{s'}| > \mathsf{TOL}(D(\mu_t,A_{s'}))$.

        The event that no information is injected again after time $T^*$ is equivalent to the event that $|Z_s - A_s| \leq \mathsf{TOL}(D(\mu_t,A_{s}))$ for every $s>T^*$. Observe that
        \begin{align*}
            &\Pr\Big(\forall s>T^*, |Z_s - A_s| \leq \mathsf{TOL}(D(\mu_t,A_s)) \big\lvert T^*\Big) \\
            &\geq \Pr\Big(\forall s>T^*, |Z_s - A_s| \leq \mathsf{TOL}(D(\mu_t,A_t)) \big\lvert T^*\Big) \tag{$A_s \geq A_t$}\\
            &= \Pr\Big(\forall s>T^*, |\bar{A}_s - A_s| \leq \mathsf{TOL}(D(\mu_t,A_t)) \big\lvert T^*\Big) \tag{$\bar{A_{s}} = 1 - (1-A_{T^*})e^{-\lambda (s - T^*)} = Z_s$},
        \end{align*}
        From Lemma~\ref{lem: unlucky bound} and Lemma~\ref{lem: unlucky bound_2}, we know that 
        \begin{align*}
            \Pr\Big(\forall s>T^*, |\bar{A}_s - A_s| \leq \mathsf{TOL}(D(\mu_t,A_t)) \big\lvert T^*\Big) \geq 1- (\bar c + \underline c)N^{-1} \mathsf{TOL}(D(\mu_t,A_t))^{-4}. 
        \end{align*}
        Therefore, we can write down the lower bound of $A_s$ when $s >T^*$ under this case as follows:
        \begin{align*}
            \forall s>T^*, A_s \geq \bar{A}_s - \mathsf{TOL}(D(\mu_t,A_s)) \geq \bar{A}_s - \mathsf{TOL}(\psi^n(A_t) - \psi_{LD}(A_t))
        \end{align*}
        with probability at least $1- (\bar c + \underline c)N^{-1} \mathsf{TOL}(D(\mu_t,A_t))^{-4} $.

        \noindent \textbf{Case 2: $|Z_{T^*} - A_{T^*}| \geq \mathsf{TOL}(D(\mu_{t},A_{T^*}))$.} In this case, information is injected at $T^*$. Note that, if $(\mu_{T^*},A_{T^*}) \in \Psi^n,$ then everyone prefers to take action 1 at $T^*$. This increases $A_s,$ inducing every agent taking action 1 after time $T^*$ until time $s'$ at which information is injected again. Thus, the probability that $(\mu_{T^*},A_{T^*}) \in \Psi^n$ and no information is injected again is at least 
        \begin{align*}
            \Pr((\mu_{T^*},A_{T^*}) \in \Psi^n \mid T^*) \cdot \left\{1- (\bar c + \underline c)N^{-1} \mathsf{TOL}(D(\mu_t,A_t))^{-4} \right\}.
        \end{align*}
        By definition, we have
        \begin{align*}
            \Pr((\mu_{T^*},A_{T^*}) \in \Psi^n \mid T^*) = p_+(\mu_t,A_{T^*})1\{(\mu_t + M\cdot\mathsf{TOL}(D(\mu_t,A_{T^*})),A_{T^*}) \in \Psi^n\}
        \end{align*}
        
        Now we claim that
        \begin{align*}
            A_{T^*} > A_t - \mathsf{TOL}(D(\mu_t,A_t)) - \frac{1}{N}.
        \end{align*}
        To see this, suppose for a contradiction that $A_{T^*} \leq A_t - \mathsf{TOL}(D(\mu_t,A_t)) - \frac{1}{N}$, which implies $A_{T^*}<A_t$. However, since the definition of $T^*$ implies $A_{T^*-} \geq Z_{T^*-} - \mathsf{TOL}(D(\mu_t,A_{T^*-})),$ we have 
        \begin{align*}
            A_{T^*} &\geq A_{T^*-} - \frac{1}{N} \\
            &\geq  Z_{T^*-} - \mathsf{TOL}(D(\mu_t,A_{T^*-})) - \frac{1}{N} \\
            &> A_t - \mathsf{TOL}(D(\mu_t,A_t)) - \frac{1}{N},  \tag{$Z_{T^*-} = \bar A_{T^*} > A_t$ and $A_t > A_{T^*}$}
        \end{align*}
        where the first inequality follows from the increment size of $A_t$ being at most $1/N$ by assumption. This is a contradiction.

        Hence, if $\mu_t \geq \psi^n(A_t) - M\cdot\mathsf{TOL}(D(\mu_t,A_t))/2 + \frac{1}{N}L_{\psi}$,\footnote{We use this condition when we construct $\Psi^{n+1}$.} we must have
        \begin{align*}
            \mu_t + M\cdot\mathsf{TOL}(D(\mu_t,A_t)) &\geq \psi^n(A_t) + M\cdot\mathsf{TOL}(D(\mu_t,A_t))/2 + \frac{1}{N}L_{\psi} \\
            &> (\psi^n(A_{T^*}) -  L_{\psi} \mathsf{TOL}(D(\mu_t,A_t))) + M\cdot\mathsf{TOL}(D(\mu_t,A_t))/2 \tag{Lipschitz continuity of $\psi^n$}\\
            &= \psi^n(A_{T^*}),
        \end{align*}
        by setting $M = 2L_{\psi}$. Thus, $(\mu_t + M\cdot\mathsf{TOL}(D(\mu_t,A_t)),A_{T^*}) \in \Psi^n$ holds, implying
        \begin{align*}
            \Pr((\mu_{T^*},A_{T^*}) \in \Psi^n \mid T^*) &= p_+(\mu_t,A_{T^*})\\
            &= 1 - \frac{M\cdot\mathsf{TOL}(D(\mu_t,A_{T^*}))}{\mathsf{DOWN}(D(\mu_{t},A_{T^*})) + M\cdot\mathsf{TOL}(D(\mu_t,A_{T^*}))} \\
            &\geq 1 - \frac{M\cdot\mathsf{TOL}(D(\mu_t,A_{T^*}))}{\mathsf{DOWN}(D(\mu_{t},A_{T^*}))} \tag{$M\cdot\mathsf{TOL}(D(\mu_t,A_{T^*})) \geq 0$}\\
            &= 1-\frac{\bar{\delta} \lambda MC}{2L+2LM(\mu_t - \psi_{LD}(A_{T^*}))^{-1}} \\
            &\geq 1-\frac{\bar{\delta} \lambda MC}{2L+2LM(\psi^n(A_t) - \psi_{LD}(A_t))^{-1}} \tag{From \eqref{ineq: delta} and continuity of $A_s$}\\
            &\geq 1- \bar{\delta} c (\psi^n(A_t) - \psi_{LD}(A_t))
        \end{align*}
        for absolute constant $c := \frac{\lambda C}{2L}$. Then we can write down the lower bound of $A_s$ for every $s>T^*$ under this case as follows:
        \begin{align*}
            \forall s>T^*, A_s \geq \bar{A}_s - \mathsf{TOL}(D(\mu_t,A_s)) \geq \bar{A}_s - \mathsf{TOL}(\psi^n(A_t) - \psi_{LD}(A_t))
        \end{align*}
        with probability at least 
        \begin{align*}
            &(1- \bar{\delta} c (\psi^n(A_t) - \psi_{LD}(A_t)) ) \cdot (1- (\bar c + \underline c)N^{-1} \mathsf{TOL}(D(\mu_t,A_t))^{-4} ) \\
            &\geq 1- \bar{\delta} c (\psi^n(A_t) - \psi_{LD}(A_t)) - (\bar c + \underline c)N^{-1} \mathsf{TOL}(D(\mu_t,A_t))^{-4}.
        \end{align*}
    Combining \textbf{Case 1} and \textbf{Case 2}, we must have 
    \begin{align}\label{ineq: after T^*}
            \forall s>T^*, A_s  \geq \bar{A}_s - \mathsf{TOL}(\psi^n(A_t) - \psi_{LD}(A_t))
    \end{align}
        with probability at least 
             $1- \bar{\delta} c (\psi^n(A_t) - \psi_{LD}(A_t)) - (\bar c + \underline c)N^{-1} \mathsf{TOL}(D(\mu_t,A_t))^{-4}.$

\noindent     \textbf{Obtaining the lower bound.} Our next step is to compute the lower bound when agent $i$ takes action $1$ at time $t$ and the upper bound when agent $i$ takes action $0$ at time $t$. One difference from the case with a continuum of agents is that agent $i$'s action affects the entire future path of aggregate actions. Therefore, we need to account for these effects when computing the bounds. Finally, using these bounds, we show that there exists $G$ such that if $\mu_t > \psi_{LD}(A_t) + (GN)^{-1/9}$, agent $i$ strictly prefers action $1$ when $\mu_t \geq \psi^n(A_t) - M\cdot\mathsf{TOL}(D(\mu_t, A_t))/2 + \frac{1}{N}L_{\psi}$, given that all agents take action $1$ for all $(\mu_s, A_s) \in \Psi^n$.

    Suppose that agent $i$ takes action $a \in \{0,1\}$ at $(\mu_{t-},A_{t-})$ and takes a (random) action $a_{in}$ after each tick of her Poisson clock $(\tau_n)_n$. We call this strategy $\sigma_i$ and assume that it induces $A_s(\sigma_i) $.\footnote{Again, $A_s$ depends on agent $i$'s strategy $\sigma_i$ because of finiteness.} Her payoff from strategy $\sigma_i$ is given by
    \begin{align*}
        U_i(\sigma_i) = \Ex_\mu\bigg[\sum_{n=0}^\infty \int_{\tau_n }^{\tau_{n+1}} e^{-r(s-t)} u(a_{in},A_s(\sigma_i),\theta) ds\bigg].
    \end{align*}
    In \eqref{ineq: before T^*}, we showed that $A_s(\sigma_i) \geq \bar{A}_s - \mathsf{TOL}(\psi^n(A_t) - \psi_{LD}(A_t))$ for every $s<T^*.$\footnote{The increment of $(A_s(\sigma_i))_s$ is at most $1/N$ because the probability that Poisson clocks of more than one agents tick at the same time is zero. Hence we can apply the earlier arguments.}  After time $T^*$, if no information is injected again, everyone (including agent $i$) takes action $1$, implying $a_{in} = 1$ if $\tau_n > T^*$. In \eqref{ineq: after T^*}, we showed 
    \begin{align*}
            \forall s>T^*, A_s(\sigma_i)  \geq \bar{A}_s - \mathsf{TOL}(\psi^n(A_t) - \psi_{LD}(A_t))
        \end{align*}
        with probability at least 
             $1- \bar{\delta} c (\psi^n(A_t) - \psi_{LD}(A_t)) - (\bar c + \underline c)N^{-1} \mathsf{TOL}(D(\mu_t,A_t))^{-4}.$
    Combining before and after $T^*$, we have 
        \begin{align*}
            \forall s \ne T^*, A_s(\sigma_i)  \geq \bar{A}_s - \mathsf{TOL}(\psi^n(A_t) - \psi_{LD}(A_t))
        \end{align*}
        with probability at least 
             $1- \bar{\delta} c (\psi^n(A_t) - \psi_{LD}(A_t)) - (\bar c + \underline c)N^{-1} \mathsf{TOL}(D(\mu_t,A_t))^{-4}.$     
    By Lipschitz continuity of $u(a,\cdot,\theta)$ and $\Delta u(\cdot,\theta),$ we must have\footnote{Let $L_0$ and $L_1$ be Lipschitz constants of $u(0, \cdot, \theta)$ and $u(1, \cdot, \theta)$, respectively. Then, $\Delta u(\cdot, \theta)$ is Lipschitz continuous with constant $L := L_0 + L_1$.}
    \begin{align*}
        \forall s \ne T^*, \forall a\in \{0,1\}, |u(a,A_s(\sigma_i),\theta) -  u(a,\bar{A}_s,\theta)| &\leq \mathsf{TOL}(\psi^n(A_t) - \psi_{LD}(A_t))L
    \end{align*}
    with probability at least 
    \[
    P_N(\mu_t,A_t) := 1- \bar{\delta} c (\psi^n(A_t) - \psi_{LD}(A_t)) - (\bar c + \underline c)N^{-1} \mathsf{TOL}(D(\mu_t,A_t))^{-4}.
    \] Thus, for every strategy $\sigma_i$ under conjecture $\Psi^n$, this implies
    \begin{align}
        &\bigg\lvert U_i(\sigma_i) - \underbrace{\Ex_\mu\bigg[\sum_{n=0}^\infty \int_{\tau_n }^{\tau_{n+1}} e^{-r(s-t)} u(a_{in},\bar{A}_s,\theta) ds}_{\eqqcolon U^*_i(\sigma_i)}\bigg] \bigg\rvert \notag\\
        &=\left| \Ex_\mu\bigg[\sum_{n=0}^\infty \int_{\tau_n }^{\tau_{n+1}} e^{-r(s-t)} \left\{u(a_{in},A_s(\sigma_i),\theta) - u(a_{in},\bar{A}_s,\theta) \right\} ds \right| \notag\\
        &\leq \Ex\bigg[\sum_{n=0}^\infty \int_{\tau_n }^{\tau_{n+1}} e^{-r(s-t)}   (P_N(\mu_t,A_t)\mathsf{TOL}(\psi^n(A_t) - \psi_{LD}(A_t))L  + (1-P_N(\mu_t,A_t)) L \Big) \bigg]  ds  \tag{From the above inequality}\\
        &= \left\{P_N(\mu_t,A_t)\mathsf{TOL}(\psi^n(A_t) - \psi_{LD}(A_t))L  + (1-P_N(\mu_t,A_t)) L \right\} \cdot  \int_{t}^{\infty} e^{-r(s-t)} ds \notag\\
        &=\frac{P_N(\mu_t,A_t)\mathsf{TOL}(\psi^n(A_t) - \psi_{LD}(A_t))L  + (1-P_N(\mu_t,A_t)) L }{r}. \label{inequality: U-U^*}
    \end{align}
    
    Now define $\sigma^1_i$ to be a strategy that agent $i$ always takes action 1. Suppose that agent $i$ takes action 0 at the beginning for $\sigma_i.$ Consider that, since $\mu > \psi_{LD}(A_t),$ if $\mu_t \geq \psi^n(A_t) - M\cdot\mathsf{TOL}(D(\mu_t,A_t))/2  +\frac{1}{N}L_{\psi}$,
    \begin{align}
         U^*_i(\sigma_i) &= \Ex_{\mu}\bigg[\sum_{n=0}^\infty \int_{\tau_n}^{\tau_{n+1}} e^{-r(s-t)} u(a_{in}, \bar{A}_s,\theta) ds \bigg] \\
         &\leq \Ex_{\mu}\bigg[ \int_{t}^{\tau_1} e^{-r(s-t)} u(0, \bar{A}_s,\theta) ds \bigg] + \Ex_{\mu}\bigg[\sum_{n=1}^\infty \int_{\tau_n}^{\tau_{n+1}} e^{-r(s-t)} u(1, \bar{A}_s,\theta) ds \bigg] \tag{From Lemma~\ref{lem: finite}}\\
         &= U^*_i(\sigma^1_i) - \Ex_{\mu}\bigg[ \int_{t}^{\tau_1} e^{-r(s-t)}  \Delta u(\bar{A}_s,\theta) ds \bigg] \\
         &\leq U^*_i(\sigma^1_i) - \frac{C}{2}(\psi^n(A_t) - \psi_{LD}(A_t)).\quad\quad\quad\quad\quad\quad\quad\quad\text{(From \eqref{ineq: first term})} \label{inequality: sigma-sigma^1}
    \end{align}

    Therefore, we have
    \begin{align*}
        &U_i(\sigma^1_i) - U_i(\sigma_i) \\
        &= (U_i(\sigma^1_i) - U^*_i(\sigma^1_i) ) + (U^*_i(\sigma^1_i) - U_i^*(\sigma_i)) + (U_i^*(\sigma_i) - U_i(\sigma_i)) \\
        &\geq \frac{C}{2}(\psi^n(A_t) - \psi_{LD}(A_t)) - 2 \cdot \frac{P_N(\mu_t,A_t)\mathsf{TOL}(\psi^n(A_t) - \psi_{LD}(A_t))L  + (1-P_N(\mu_t,A_t)) L }{r} \tag{From \eqref{inequality: U-U^*} and \eqref{inequality: sigma-sigma^1}}\\
        &\geq \frac{C}{2}(\psi^n(A_t) - \psi_{LD}(A_t)) - 2 \cdot \frac{\mathsf{TOL}(\psi^n(A_t) - \psi_{LD}(A_t))L  + (1-P_N(\mu_t,A_t)) L }{r} \tag{$P_N(\mu_t,A_t) \leq 1$}\\
    \end{align*}
    Recall that 
    \begin{align*}
        P_N(\mu_t,A_t) &= 1- \bar{\delta} c (\psi^n(A_t) - \psi_{LD}(A_t)) - (\bar c + \underline c)N^{-1} \mathsf{TOL}(D(\mu_t,A_t))^{-4}.
    \end{align*}
    Since $D(\mu_t,A_t) = \mu_t - \psi_{LD}(A_t) > (d N)^{-1/9}$ holds by assumption, we must have
    \begin{align*}
        P_N(\mu_t,A_t)  > 1- \bar{\delta} c (\psi^n(A_t) - \psi_{LD}(A_t)) - (\bar c + \underline c) (\underline e\bar{\delta})^{-4} d^{8/9} N^{-1/9} 
    \end{align*}
    for some constant $\bar e$ and $\underline e$ such that $\bar e\bar{\delta}D^2 \geq \mathsf{TOL}(D) \geq \underline e\bar{\delta}D^2.$\footnote{By the definition of $\mathsf{TOL}$, any $\bar e \geq \lambda C/4L$ and $\underline e \leq \lambda C/\{4L(1 + M)\}$ works.} Define $\phi_n := \psi^n(A_t) -\psi_{LD}(A_t)$. Since $\mu_t \leq \psi^n(A_t),$ we have $\phi_n \geq D(\mu_t,A_t) > (G N)^{-1/9}.$
    Thus, 
    \begin{align*}
        &U_i(\sigma^1_i) - U_i(\sigma_i) \\
        &\geq \frac{C\phi_n}{2} - 2\cdot \frac{\bar e\bar{\delta}\phi_n^2L + (\bar{\delta}c\phi_n + (\bar c + \underline c)(\underline e\bar{\delta})^{-4}G^{8/9} N^{-1/9} )}{r} \tag{From the above inequality}\\
        &\geq \bigg(\frac{C}{2} - \frac{2\bar\delta(\bar eL + c)}{r} \bigg) \phi_n - \frac{2(\bar c + \underline c)(\underline e\bar{\delta})^{-4}G^{8/9} N^{-1/9}}{r} \tag{$\phi_n \leq 1$}\\
        &\geq \bigg(\frac{C}{2} - \frac{2\bar\delta(\bar eL + c)}{r} \bigg) (dN)^{-1/9} - \frac{2(\bar c + \underline c)(\underline e\bar{\delta})^{-4}G^{8/9} N^{-1/9}}{r} \tag{$\phi_n > (GN)^{-1/9}$}\\
        &>0,
    \end{align*}
    where the last inequality is true if we choose $\bar\delta$ and $G$ such that
    \begin{align*}
        \bar{\delta} &< \frac{Cr}{4(\bar eL + c)} \\
        G &< \frac{r (\underline e \bar\delta)^4}{2(\bar c + \underline c)} \bigg( \frac{C}{2} - \frac{2\bar\delta(\bar eL + c)}{r}\bigg).
    \end{align*}
    In conclusion, we have shown that there exists a constant $d$ such that if $\mu_t > \psi_{LD}(A_t) + (dN)^{-1/9}$, agent $i$ strictly prefers action $1$ when $\mu_t \geq \psi^n(A_t) - M\cdot\mathsf{TOL}(D(\mu_t, A_t))/2 + \frac{1}{N}L_{\psi}$, given that all agents take action $1$ for all $(\mu_s, A_s) \in \Psi^n$.

    \noindent \textbf{Characterizing $\Psi^{n+1}$.} Note that $\mathsf{TOL}$ is increasing. Thus, $\mu_t + M\cdot\mathsf{TOL}(D(\mu_t,A_t))/2$ is increasing and continuous in $\mu_t$. Therefore, for each $A_t$, there exists $\mu'(A_t) < \psi^n(A_t)$ such that 
\begin{align*}
    \mu'(A_t) + \frac{M\cdot\mathsf{TOL}(D(\mu'(A_t),A_t))}{2} =  \psi^n(A_t) + \frac{L_{\psi}}{N}
\end{align*}
if
\begin{align*}
    \frac{M\cdot\mathsf{TOL}(D(\psi^n(A_t),A_t))}{2} > \frac{L_{\psi}}{N}.
\end{align*}
A sufficient condition for this is
\begin{align*}
    \frac{M}{2}\bar\delta\underline e \left( GN \right)^{-2/9} > \frac{L_{\psi}}{N} \Leftrightarrow G < \left( \frac{M\bar\delta\underline e}{2L_{\psi}} \right)^{\frac{9}{2}}N^{\frac{7}{2}}.
\end{align*}
Hence, taking $G$ such that
\begin{align}\label{ineq: sufficient condition}
G < \min \left\{ \left( \frac{M\bar\delta\underline e}{2L_{\psi}} \right)^{\frac{9}{2}}, \frac{r (\underline e \bar\delta)^4}{2(\bar c + \underline c)} \bigg( \frac{C}{2} - \frac{2\bar\delta(\bar eL + c)}{r}\bigg) \right\}
\end{align}
is sufficient. Note that we choose $C$ 

Then we define
\[\Psi^{n+1} = \{(\mu_t,A_t) : \mu_t \geq \mu'(A_t)\}\]
From the argument above, we must have an agent always choosing action $1$ whenever $(\mu_t,A_t) \in \Psi^{n+1}.$ Moreover, we can rewrite the above equation as follows:
\begin{align*}
    (\mu'(A_t)-\psi_{LD}(A_t)) + \frac{M\cdot\mathsf{TOL}(\mu'(A_t)-\psi_{LD}(A_t))}{2} = \psi^n(A_t) - \psi_{LD}(A_t) + \frac{L_{\psi}}{N},
\end{align*}
where the RHS is constant in $A_t$ by the property of $\psi^n$. Thus, $\mu'(A_t) - \psi_{LD}(A_t)$ must be also constant in $A_t.$ This concludes that round-$(n+1)$ dominance region $\Psi^{n+1}$ satisfies $\Psi^n \subset \Psi^{n+1}$ because $c_n = \psi^n(A_t) - \psi_{LD}(A_t) > \mu'(A_t) - \psi_{LD}(A_t) =: c_{n+1}$ when $\eqref{ineq: sufficient condition}$ is satisfied.
\end{proof}

To conclude the proof of part (i) of Theorem~\ref{thrm: finite}, we show the following lemma.

\begin{lemma}
    \label{lem: induction finite}
    \[
    \bigcup_{n \in \mathbb{N}} \Psi^n \supseteq \Big\{(\mu,A) \in \Delta(\Theta) \times [0,1]: \mu > \psi_{LD}(A) + (GN)^{-1/9}\Big\}.
    \]
\end{lemma}

\begin{proof}[Proof of Lemma~\ref{lem: induction finite}]
    Recall $\psi^n(A_t) = \sup\{\mu \in \Delta(\Theta): (\mu,A_t) \notin \Psi^n\}.$ By Lemma~\ref{lem: contagion, finite}, $\psi^n(A_t)$ is decreasing in $n$. Define $\psi^*(A_t) = \lim_{n\to\infty} \psi^n(A_t)$. In limit, we must have
\begin{align*}
    &\psi^*(A_t) + M\cdot\mathsf{TOL}(D(\psi^*(A_t),A_t))/2 = \psi^*(A_t) + L_{\psi}/N\\
    &\Rightarrow \mathsf{TOL}(D(\psi^*(A_t),A_t)) = 2L_{\psi}/(MN).
\end{align*}
Since our choice of $G$ by \eqref{ineq: sufficient condition} ensures
\[
\frac{2L_{\psi}}{MN} \leq \mathsf{TOL}\left((GN)^{-1/9}\right),
\]
we have 
\[
D(\psi^*(A_t),A_t) \leq \mu_t - \psi_{LD}(A_t) \Leftrightarrow \psi^*(A_t) \leq \mu_t
\]
for any $\mu_t > \psi_{LD}(A_t) + (GN)^{-1/9}$, as desired.
\end{proof}
This concludes the proof of part (i) of Theorem \ref{thrm: finite}. \end{proof}

\begin{proof}[Proof of Part (ii) of Theorem \ref{thrm: finite}]
Consider $N$ large enough so that $\mu_0 > \psi_{LD}(A_0) + (dN)^{-1/9}.$ Under $\bm{\mu}^\eta$ and the environment of $N$ agents, Part (i) implies everyone takes action 1 under any equilibrium outcome until new information is injected. 

Without loss of generality, we assume $\phi(\bm A) \geq 0.$ Let $\bm\tau := (\tau_i)_{i = 1}^N$. We have
\begin{align*}
&\inf_{\sigma \in \Sigma^N(\bm{\mu}^\eta, A_0)}
\Ex^{\sigma}\Big[\phi(\bm{A}) \Big] \\
& \geq \Ex_{\bm\tau} \left[1 \left\{\forall t, |\bar A_t - \bar A_t^N| \leq \min\left\{\mathsf{TOL}(D(\mu_0, A_t)), N^{-2/9} \right\} \right\}\phi(\bm{\bar{A}}^N)\right]\\
& \geq \Ex_{\bm\tau} \left[1 \left\{\forall t, |\bar A_t - \bar A_t^N| \leq \min\left\{\mathsf{TOL}(D(\mu_0, A_0)), N^{-2/9} \right\} \right\}\phi(\bm{\bar{A}}^N)\right] \tag{$A_t \geq A_0$}\\
&\geq \Ex_{\bm\tau} \left[1 \left\{\forall t, |\bar A_t - \bar A_t^N| \leq \min\left\{\mathsf{TOL}(D(\mu_0, A_0)), N^{-2/9} \right\} \right\} \left\{\phi(\bm{\bar{A}}) - L_\phi \|\bm{\bar{A}}-\bm{\bar{A}}^N \|_\infty \right\}\right]  \tag{Lipschitz continuity of $\phi$}\\
& \geq \left\{ 1- (\bar c + \underline c)N^{-1} \min\left\{\mathsf{TOL}(D(\mu_0, A_0)), N^{-2/9} \right\}^{-4} \right\}\left\{\phi(\bm{\bar{A}}) - L_\phi N^{-2/9} \right\} \tag{Lemmas \ref{lem: unlucky bound} and \ref{lem: unlucky bound_2}}\\
&\geq \left\{ 1- (\bar c + \underline c)N^{-1/9} \max\left\{(\underline e\bar\delta)^{-4} G^{8/9}, 1 \right\} \right\}\left\{\phi(\bm{\bar{A}}) - L_\phi N^{-2/9}\right\} \tag{$\mathsf{TOL}(D) \geq \underline e\bar\delta D^2$ and $D \geq (GN)^{-1/9}$}\\
&\geq \phi(\bm{\bar{A}}) - K_2N^{-1/9}
\end{align*}
for some constant $K_2$, where $\bm{\bar{A}}$ satisfies $\bar{A}_t = \bar{A}(A_0, t) = 1 - (1 - A_0)e^{-\lambda t}$, and $\bm{\bar{A}}^N$ satisfies
\begin{align*}
    \bar{A}_t^N &= A_0 + \frac{1}{N}\sum_{i = 1}^{n} 1 \{ \tau_i \leq t \}
\end{align*}
with $n$ being the number of agents playing $0$ at time $0$. The proof of Theorem \ref{thrm:main} implies $\eqref{eqn:adv} = \phi(\bm{\bar{A}}).$ Thus,
\begin{align} \label{eq: adv-n>adv}
    \eqref{eqn:adv-n} + K_2N^{-1/9} \geq \eqref{eqn:adv}
\end{align}
when $N$ is large enough, as desired. Note that
\begin{align*}
    \big\lvert \mathbb{E}[\phi(\bar{\bm{A}}^N) - \phi(\bar{\bm{A}})] \big\rvert &\leq \mathbb{E}[|\phi(\bar{\bm{A}}^N) - \phi(\bar{\bm{A}})|] \\
    &\leq L_\phi \mathbb{E} \big[ \| {\bm{\bar{A}}}^N - {\bm{\bar{A}}}\|_\infty \big] \\
    &\leq L_\phi \Big( N^{-1/5} + \mathbb{P}\big(\| {\bm{\bar{A}}}^N - {\bm{\bar{A}}}\|_\infty \geq N^{-1/5}\big) \Big) \\
    &\leq \bar{K}N^{-1/5}, \tag{Lemmas \ref{lem: unlucky bound} and \ref{lem: unlucky bound_2}}
\end{align*}
where $\bar{K} =L_\phi(1+ \bar{c}+\underline{c}).$ This implies
\begin{align} \label{eq: adv-n<adv}
    \eqref{eqn:adv-n} \leq \eqref{eqn:opt-n} \leq \mathbb{E}[\phi(\bar{\bm{A}}^N)] \leq \phi(\bar{\bm{A}}) + \bar{K}N^{-1/5} \leq \eqref{eqn:adv} + \bar{K}N^{-1/5}.
\end{align}
Inequalities \eqref{eq: adv-n>adv} and \eqref{eq: adv-n<adv} imply
\begin{align*}
    \eqref{eqn:opt-n} - \eqref{eqn:adv-n} &\leq \eqref{eqn:adv} - \eqref{eqn:adv-n} + \bar K N^{-1/9} \\
    &\leq \underbrace{(\bar K + K_2)}_{=: K_1}N^{-1/9}
\end{align*}
\end{proof}

\begin{proof}[Proof of Part (iii) of Theorem \ref{thrm: finite}] Since $\mu_0 > \psi_{LD}(A_0)$, we must have $\mu_{t-} > \psi_{LD}(A_t)$ for every $t$ for every history $H_t$. Let $N$ be large enough such that $\mu_{t-} > \psi_{LD}(A_t) + 2(GN)^{-1/9}$. We consider the following two cases for : 
\begin{itemize}[leftmargin=*]
    \item \textbf{Case 1:} If $\mu_{t-} > \psi_{LD}(A_t) + 2(GN)^{-1/9}$ and $|A_t - Z_{t-}| < \mathsf{TOL}(D(\mu_t, A_t))$. In this case, there is no information arriving, and everyone takes action 1. This will increase $A_t$, and every agent always takes action $1$ from time $t$ onwards as long as $|\bar A_s - \bar A_s^N| \leq \mathsf{TOL}(D(\mu_s, A_s))$ for all $s \geq t$. Since Lemmas~\ref{lem: unlucky bound} and \ref{lem: unlucky bound_2} imply that such probability converges to $1$ as $N\to\infty$, the designer's payoff converges to the best case, implying sequential optimality as $N \to \infty$.
    
    \item \textbf{Case 2:} If  $\mu_{t-} > \psi_{LD}(A_t) + 2(GN)^{-1/9}$ and $|A_t - Z_{t-}| \geq \mathsf{TOL}(D(\mu_t, A_t))$. In this case, the belief moves to either $\mu_{t-} + M \cdot \mathsf{TOL}(D)$ or $\mu_t - \mathsf{DOWN}(D)$. Note that $\mu_{t-} - \mathsf{DOWN}(D) = (\mu_t + \psi_{LD}(A_t))/2 > \psi_{LD}(A_t) + (GN)^{-1/9}$. So no matter what information arrives, every agent takes action 1. This will increase $A_t$, and every agent always takes action $1$ after time $t$ as long as $|\bar A_s - \bar A_s^N| \leq \mathsf{TOL}(D(\mu_s, A_s))$ for all $s \geq t$. Again, since such probability converges to $1$ as $N\to\infty$, we have sequential optimality as $N \to \infty$.
\end{itemize}

\end{proof}

\subsection{Optimal implementation when prior is in lower dominance region} \label{appendix:finite_in_LD}
Theorem \ref{thrm: finite} established an analog of our result in the main text the prior is outside the lower dominance region. We now analyze informational puts when the prior lies within the lower dominance region. 

\begin{proposition} 
    \label{prop: finite LD}
    \phantom{} 
    \begin{enumerate}
    \item [(i)] If $\mu_0 \in \Psi_{LD}(A_0)$, then there exists a constant $K$ such that 
    \[\eqref{eqn:adv-n} \geq \eqref{eqn:adv} - KN^{-1/9}, \]
    for sufficiently large $N$.
    \end{enumerate}
    Suppose $u(1,A,\theta)$ and $u(0,A,\theta)$ are increasing and decreasing in $A$, respectively. The following statements hold:
    \begin{enumerate}
    \item [(ii)] The gap between \eqref{eqn:opt-n} and \eqref{eqn:opt} vanishes for sufficiently large $N$ and large $r$:
     \[\text{lim sup}_{r \to \infty} |\eqref{eqn:opt-n} - \eqref{eqn:opt}| \leq H N^{-1/9}\]
    holds for some constant $H$ when $N$ is sufficiently large.

    \item[(iii)] The multiplicity gap vanishes for sufficiently large $N$ and large $r$:
    \[\lim_{r \to \infty}\lim_{N \to \infty} |\eqref{eqn:opt-n} - \eqref{eqn:adv-n}| = 0. \]
    \end{enumerate}
\end{proposition}
\begin{proof}[Proof of Part (i) of Proposition \ref{prop: finite LD}] The proof of Theorem~\ref{thrm:main} implies
\begin{align*}
    \eqref{eqn:adv} = \sup_{\substack{\bm{\mu} \in \mathcal{M} \\
    \sigma \in \Sigma(\bm{\mu}, A_0)}}\Ex^{\sigma}\Big[\phi(\bm{A})\Big] \leq  (1-p^*(\mu_0)) \phi(\bm{\underline{A}}) + p^*(\mu_0)\phi(\bm{\bar{A}}),
\end{align*}
where $p^*(\mu_0) := \mu_0/\psi_{LD}(A_0)$, and $\bm{\underline{A}}$ satisfies $\underline{A}_t = \underline{A}(A_0, t) = A_0 e^{-\lambda t}$.

 Consider $\bm{\mu}^\eta$ such that $\eta > 2(GN)^{-1/9}/(2(GN)^{-1/9} + \psi_{LD}(A_0))$, where $G$ is defined in Theorem \ref{thrm: finite}. We have
 \[
 \mu_0^+ = \frac{\mu_0}{p^*(\mu_0) - \eta} > \psi_{LD}(A_0) + 2(GN)^{-1/9},
 \]
where $\mu_0^+$ is the maximal escaping belief defined in the main text. Under $\bm{\mu}^\eta$ and the environment of $N$ agents, if $\mu_{0+} > \psi_{LD}(A_0) + (GN)^{-1/9},$ then everyone takes action $1$ until new information is injected under any equilibrium outcome by part (i) of Theorem \ref{thrm: finite}. Thus, we have
\begin{align*}
&\inf_{\sigma \in \Sigma^N(\bm{\mu}^\eta, A_0)}
\Ex^{\sigma}\Big[\phi(\bm{A}) \Big]\\
&\geq (1-p^*(\mu_0) + \eta) \Ex_{\bm\tau}\left[1\left\{ \forall t, |\underline A_t - \underline A_t^N| \leq \min\left\{\mathsf{TOL}(D(\mu_t, A_t)), N^{-2/9} \right\} \right\} \phi(\bm{\underline{A}}^N) \right] \\
&\quad + (p^*(\mu_0)- \eta)\Ex_{\bm\tau}\left[1\left\{ \forall t, |\bar A_t - \bar A_t^N| \leq \min\left\{\mathsf{TOL}(D(\mu_t, A_t)), N^{-2/9} \right\}  \right\}\phi(\bm{\bar{A}}^N)\right] \tag{$\phi(\bm A) \geq 0$}\\
&\geq (1-p^*(\mu_0) + \eta)\left\{ 1- (\bar c + \underline c)N^{-1/9} \max\left\{(\underline e\bar\delta)^{-4} G^{8/9}, 1 \right\} \right\} \left\{ \phi(\bm{\underline{A}}) - L_\phi N^{-2/9} \right\} \\
&\quad + (p^*(\mu_0)- \eta)\left\{ 1- (\bar c + \underline c)N^{-1/9} \max\left\{(\underline e\bar\delta)^{-4} G^{8/9}, 1 \right\} \right\}\left\{ \phi(\bm{\bar{A}}) - L_\phi N^{-2/9}\right\} \tag{From the same argument as Theorem \ref{thrm: finite} part (ii)},
\end{align*}
where $\bm{\underline{A}}^N$ satisfies
\begin{align*}
    \underline A_t^N &= A_0 - \frac{1}{N}\sum_{i = 1}^{N - n} 1 \{ \tau_i \leq t \}.
\end{align*}
This implies
\begin{align*}
    \eqref{eqn:adv-n} &\geq (1-p^*(\mu_0)) \left\{ 1- \mathcal{O}(N^{-1/9}) \right\} \left\{ \phi(\bm{\underline{A}}) - L_\phi N^{-2/9}\right\} \\
&\quad + (p^*(\mu_0))\left\{ 1- \mathcal{O}(N^{-1/9}) \right\}\left\{ \phi(\bm{\bar{A}}) - L_\phi  N^{-2/9}\right\}
\end{align*}
Hence, we have
\begin{align*}
     \eqref{eqn:adv} - \eqref{eqn:adv-n}  & \leq L_\phi N^{-2/9} \\
    &\quad + \mathcal{O}(N^{-1/9}) p^*(\mu_0)  \left\{ \phi(\bm{\bar{A}}) - L_\phi N^{-2/9} \right\} \\
    &\quad + \mathcal{O}(N^{-1/9}) (1 - p^*(\mu_0)) \left\{ \phi(\bm{\underline{A}}) - L_\phi N^{-2/9} \right\}\\
    &= \mathcal{O}(N^{-1/9}).
\end{align*}
For sufficiently large $N$, there exists a constant $K$ such that $ \eqref{eqn:adv} - \eqref{eqn:adv-n}  \leq KN^{-1/9},$ as desired.
\end{proof}

\begin{proof}[Proof of Part (ii) of Proposition \ref{prop: finite LD}]
 Suppose that $u(a,A,\theta) \in [\underline{u},\bar{u}]$. 
Consider any information structure $\bm{\mu} \in \mathcal{M}$ and equilibrium $\bm{\sigma} \in \Sigma^N(\bm{\mu},A_0).$ Suppose that agent $i$'s Poisson clock ticks at time $t$ and then she chooses a new action. Let $A_{t}^1$ and $A_t^0$ be the aggregate play at time $t$ if she chooses action 1 and 0 at time $t$, respectively. Define paths $(\bar{A}_s^{1,N})_{s \geq t}$ and $(\bar{A}_s^{0,N})_{s \geq t}$ such that
\begin{align*}
\bar{A}^{1,N}_s &= A_t^1 + \frac{1}{N} \sum_{i=1}^n 1\{\tau_i \leq t\}   \\
\bar{A}^{0,N}_t &= A_t^0 + \frac{1}{N} \sum_{i=1}^n 1\{\tau_i \leq t\}.
\end{align*}

Suppose that $(A^1_s)_{s \geq t}$ and $(A^0_s)_{s \geq t}$ are the realized aggregate action paths after the agent takes actions $1$ and $0$ at $t$, respectively. It must be that $A^1_s \leq \bar{A}^{1,N}_s$ and $A^0_s \leq \bar{A}^{0,N}_s$ a.s. for every $s \geq t$. Thus, the expected payoff of the agent choosing action 1 at $t$ followed by a path of (random) actions after $\tau$ $(a_{1s})_{s \geq \tau}$  is
\begin{align*}
&\Ex_{\tau, \theta \sim \mu_t}\bigg[ \int_t^\tau e^{-r(s-t)}u(1,A^1_s,\theta) ds + \int_\tau^\infty e^{-r(s-t)} u(a_{1s},A^1_s,\theta) ds  \bigg] \\
&\leq \Ex_{\tau, \theta \sim \mu_t}\bigg[ \int_t^\tau e^{-r(s-t)}u(1,\bar{A}^{1,N}_s,\theta) ds + \int_\tau^\infty e^{-r(s-t)} \bar{u} ds  \bigg] \tag{$u(a, \theta, A) \leq \bar u$}\\
& \leq \Ex_{\tau, \theta \sim \mu_t}\bigg[ \int_t^\tau e^{-r(s-t)}u(1,\bar{A}_s,\theta) ds \bigg] +\Ex\bigg[ \int_t^\tau e^{-r(s-t)} L|\bar{A}_s - \bar{A}^{1,N}_s| ds \bigg] + \frac{\lambda \bar{u}}{r(\lambda + r)} \tag{Lipschitz continuity of $u$}
\end{align*}
since $u(1,A,\theta)$ is increasing in $A.$
Note that $|\bar{A}_t^N - \bar{A}^{1,N}_t| \leq \frac{1}{N}$. Thus,
\begin{align*}
    \Ex|\bar{A}_s - \bar{A}^{1,N}_s| \leq \frac{1}{N} + \frac{1}{N} \Ex \Big\lvert n (1-e^{-\lambda t}) - \sum_{i=1}^n 1\{\tau_i \leq t\} \Big\rvert \tag{$n = N(1 - A_0)$} 
\end{align*}
Consider that 
\begin{align*}
    \Ex \Big[ \sum_{i=1}^n 1\{\tau_i \leq t\}\Big] &= n \Pr(\tau_i \leq t) = n(1-e^{-\lambda t}) \\
    Var \Big( \sum_{i=1}^n 1\{\tau_ i \leq t\}\Big) &= n\Pr(\tau_i \leq t)\big( 1-\Pr(\tau_i \leq t)\big) = ne^{-\lambda t}(1-e^{-\lambda t}) < n.
\end{align*}
Thus, by the Cauchy-Schwarz inequality,
\begin{align*}
    \Ex \Big\lvert n (1-e^{-\lambda t}) - \sum_{i=1}^n 1\{\tau_i \leq t\} \Big\rvert   \leq  \sqrt{Var\Big( \sum_{i=1}^n 1\{\tau_i \leq t \}\Big)} < \sqrt{n}.
\end{align*}
Therefore, $\Ex|\bar{A}_s - \bar{A}^{1,N}_s| \leq (1+\sqrt{n})N^{-1} \leq 2N^{-1/2}$. This implies the expected payoff of the agent choosing action 1 at $t$ is bounded above by
\[\Ex_{\tau,\theta \sim \mu_t}\bigg[ \int_t^\tau e^{-r(s-t)}u(1,\bar{A}_s,\theta) ds \bigg] + \frac{2LN^{-1/2}}{\lambda(\lambda+r)} + \frac{\lambda \bar{u}}{r(\lambda+r)}\eqqcolon \overline{U}^N_1 \]
Similarly, the expected payoff of the agent choosing action $0$ at $t$ followed by a path of (random) actions after $\tau$ $(a_{0s})_{s \geq \tau}$  is 
\begin{align*}
&\Ex_{\tau,\theta \sim \mu_t}\bigg[ \int_t^\tau e^{-r(s-t)}u(0,A^0_s,\theta) ds + \int_\tau^\infty e^{-r(s-t)} u(a_{0s},A^0_s,\theta) ds  \bigg] \\
&\geq \Ex_{\tau,\theta \sim \mu_t}\bigg[ \int_t^\tau e^{-r(s-t)}u(0,\bar{A}^{0,N}_s,\theta) ds + \int_\tau^\infty e^{-r(s-t)} \underline{u} ds  \bigg] \\
&\geq \Ex_{\tau,\theta \sim \mu_t}\bigg[ \int_t^\tau e^{-r(s-t)}u(0,\bar{A}_s,\theta) ds \bigg] - \frac{2LN^{-1/2}}{\lambda(\lambda+r)} + \frac{\lambda \underline{u}}{r(\lambda + r)} \eqqcolon \underline{U}^N_0.
\end{align*}
Thus, the agent always chooses action $0$ if
\begin{align*}
    &\underline{U}^N_0 \geq \overline{U}^N_1 \iff \Ex_{\tau,\theta \sim \mu_t}\bigg[ \int_t^\tau e^{-r(s-t)} \Delta u(\bar{A}_s,\theta) ds \bigg] \leq -\frac{4LN^{-1/2}}{\lambda+r} - \frac{\lambda(\bar{u} - \underline{u})}{r(\lambda+r)}
\end{align*}
From (\ref{ineq: conctant C}), if $\psi_{LD}(A_t) > 0$, then for every $\mu_t < \psi_{LD}(A_t)$
\[\Ex_{\tau,\theta \sim \mu_t}\bigg[ \int_t^\tau e^{-r(s-t)} \Delta u(\bar{A}_s,\theta) ds \bigg] < -C(\psi_{LD}(A_t) - \mu_t),\]
where
\[C = \min_{A_t \in [0,1]} \Ex_\tau\Big[ \int_{s=t}^\tau e^{-r(s-t)} \Delta u(\bar{A}_s,1) ds \Big] > \frac{ \Delta u(0,1)}{\lambda+r}.\]
Thus, if $\psi_{LD}(A_t) - \mu_t > \frac{4LN^{-1/2} + (\bar{u} - \underline{u})\lambda/r }{\Delta u(0,1)}$, then 
\[\Ex_{\tau,\theta \sim \mu_t}\bigg[ \int_t^\tau e^{-r(s-t)} \Delta u(\bar{A}_s,\theta) ds \bigg] \leq -\frac{4LN^{-1/2}}{\lambda+r} - \frac{\lambda(\bar{u} - \underline{u})}{r(\lambda+r)} \Longrightarrow \underline{U}^N_0 \geq \overline{U}^N_1,\]
which implies the agent must choose action $0$ at $t$. We define a finite-population lower dominance region $\underline{\psi}^N_{LD}:[0,1] \to [0,1]$ as follows:
\[\underline{\psi}^N_{LD}(A_t) = \max \bigg\{ 0 ,\psi_{LD}(A_t) - \frac{4LN^{-1/2} + (\bar{u}-\underline{u}) \lambda/r}{\Delta u(0,1)} \bigg\}. \]
We showed above that, if $\mu_t < \underline{\psi}^N_{LD}(A_t)$, then the agent must choose action $0.$ This implies 
\begin{align*}\Big \lvert\max\{\mu_t,\underline{\psi}^N_{LD}(A_t)\} - \max\{\mu_t,\psi_{LD}(A_t)\}\Big\rvert &\leq \Big\lvert \underline{\psi}^N_{LD}(A_t) - \psi_{LD}(A_t) \Big\rvert\\
&\leq \frac{4LN^{-1/2} + (\bar{u}-\underline{u}) \lambda/r}{\Delta u(0,1)}  \\
&\eqqcolon \Delta\psi(N,\lambda/r) \to 0
\end{align*}
in any order of limits of $\lambda/r \to 0$ and $N\to \infty.$ 

From Step 3A of the proof of Theorem \ref{thrm:main}, we can apply a similar argument to show that
\[\eqref{eqn:opt-n} \leq (1-p^*_N(\mu_0,A_0))\Ex[\phi(\underline{\bm{A}}^N)] + p^*_N(\mu_0,A_0)\Ex[\phi(\bar{\bm{A}}^N)],\]
where $p^*_N(\mu_0,A_0) = \frac{\mu_0}{\max\{\mu_0,\underline{\psi}^N_{LD}(A_0)\}}$. Recall that, if $\mu_0 < \psi_{LD}(A_0)$, then 
\[\eqref{eqn:opt} = (1-p^*(\mu_0,A_0))\phi(\underline{\bm{A}}) + p^*(\mu_0,A_0)\phi(\bar{\bm{A}}),\]
where $p^*(\mu_0,A_0) = \frac{\mu_0}{\max\{\mu_0,\psi_{LD}(A_0)\}} $.

We showed in inequality \eqref{eq: adv-n<adv} that $|\mathbb{E}[\phi(\bar{\bm{A}}^N) - \phi(\bar{\bm{A}})]| \leq \bar{K}N^{-1/5}.$ Similarly, we can show $\big\lvert \mathbb{E}[\phi(\underline{\bm{A}}^N) - \phi(\underline{\bm{A}})] \big\rvert \leq \bar{K}N^{-1/5}.$ Moreover, if we define $Q := 4L/\Delta u(0,1)$, we have
\begin{align*}
    p^*_N(\mu_0,A_0) - p^*(\mu_0,A_0) &= \frac{\mu_0}{\max\{\psi_{LD}(A_0),\mu_0\}} \cdot \bigg(\frac{\max\{\psi_{LD}(A_0),\mu_0\} }{\max\{  \underline{\psi}_{LD}^N(A_0),\mu_0\}} - 1\bigg) \\
    &\leq \frac{\Delta \psi(N,\lambda/r) }{ \mu_0 }  \\
    &\to \frac{QN^{-1/2}}{\mu_0}
\end{align*}
as $r \to \infty$ since $\Delta \psi(N,\lambda/r) \to QN^{-1/2}$ as $r \to \infty.$  Hence, we have
\begin{align*}
     &\eqref{eqn:opt-n} -\eqref{eqn:opt}\\  & \leq (1 - p_N^*(\mu_0, A_0)) \left| \mathbb{E}[\phi(\underline{\bm{A}}^N) - \phi(\underline{\bm{A}})] \right| + p_N^*(\mu_0, A_0) \left|\mathbb{E}[\phi(\bar{\bm{A}}^N) - \phi(\bar{\bm{A}})]\right| \\
     &\quad + (p_N^*(\mu_0, A_0) - p^*(\mu_0, A_0)) \left\{\phi(\bar{\bm{A}}) - \phi(\underline{\bm{A}}) \right\} \\
     &\leq \bar K N^{-1/5} +  \frac{\Delta\psi(N,\lambda/r)}{\mu_0}\left\{\phi(\bar{\bm{A}}) - \phi(\underline{\bm{A}}) \right\} \\
     &\to \bar K N^{-1/5} +  \frac{QN^{-1/2}}{\mu_0}\left\{\phi(\bar{\bm{A}}) - \phi(\underline{\bm{A}}) \right\} \quad \text{as $r\to\infty$},
\end{align*} 
and
\begin{align*}
    &\eqref{eqn:opt} - \eqref{eqn:opt-n} \\
    \leq & \eqref{eqn:opt} - \eqref{eqn:adv-n} \\
    \leq & (1-p^*(\mu_0,A_0)) \mathcal{O}(N^{-1/9}) \phi(\underline{\bm{A}}) + p^*(\mu_0,A_0) \mathcal{O}(N^{-1/9}) \phi(\bar{\bm{A}})\\
    &\quad + \left\{ 1- \mathcal{O}(N^{-1/9}) \right\} L_\phi N^{-2/9}\\
    =& \mathcal{O}(N^{-1/9}).
\end{align*}
These together imply
\begin{align*}
\text{lim sup}_{r \to \infty} |\eqref{eqn:opt-n} -\eqref{eqn:opt}| \leq HN^{-1/9}
\end{align*}
for some constant $H$, as desired.
\end{proof}

\begin{proof}[Proof of Part (iii) of Proposition \ref{prop: finite LD}]
Finally, we have 
\begin{align*}
    &\lim_{r \to \infty}\lim_{N \to \infty} |\eqref{eqn:opt-n} -\eqref{eqn:adv-n}|\\
    &\leq \lim_{r \to \infty}\lim_{N \to \infty} \underbrace{|\eqref{eqn:opt-n} -\eqref{eqn:opt}|}_{\text{Proposition~\ref{prop: finite LD} Part (ii)}} + \underbrace{|\eqref{eqn:opt} -\eqref{eqn:adv}|}_{\text{Theorem~\ref{thrm:main} Part (i)}} + \underbrace{|\eqref{eqn:adv} -\eqref{eqn:adv-n}|}_{\text{Proposition~\ref{prop: finite LD} Part (i)}} \\
    &= 0,
\end{align*}
as desired.

\end{proof}
\clearpage 

\section{Designing private information} \label{appendix:private}
In this appendix we discuss whether the designer can do better by designing private information.

\paragraph{Relaxed feasibility for joint belief processes.} We consider the relaxed problem under which each agent's belief can be `separately controlled' i.e., any joint distribution over agents' beliefs under which the marginal distribution is a martingale is feasible under the relaxed problem. 
There is a common prior $\mu_0$ and a private belief process $\bm{\mu}_i := (\mu_{it})_t$, where $\mu_{it} := \mathbb{P}(\theta = 1|\mathcal{F}_{it})$ with $\mathcal{F}_{it}$ being agent $i$'s time-$t$ filtration generated by $(A_{s},\mu_{is})_{s \leq t}$. 

The belief process for agent $i \in [0,1]$, $\bm{\mu}_i := (\mu_{it})_{t}$ is R-feasible if it is an $(\mathcal{F}_{it})_t$-martingale. The set of joint R-feasible belief process is 
\[
\mathcal{M}^P := \Big\{(\mu_{it})_{t}: i \in [0,1], \text{ $(\mu_{it})_{t}$ is R-feasible } \Big\}.
\]
We emphasize that this is a necessary condition on beliefs, but is not sufficient (see, e.g., \cite*{arieli2021feasible,morris2020no} for a discussion of the static case). Let the set of feasible joint belief processes be $\mathcal{M}^F$. Although it is still an open question of how to characterize this set, we know $\mathcal{M}^F \subseteq \mathcal{M}^P$.

\paragraph{The problem under private information.}
\[
\sup_{\bm{\mu} \in \mathcal{M}^F} \inf_{\bm{\sigma} \in PBE(\bm{\mu},A_0)} \Ex^{\sigma} \Big[\phi\big(\bm{A}\big)\Big] \tag{ADV-P} \label{eqn:private_ADV}.
\]
noting that we have moved from subgame perfection to Perfect-Bayesian Equilibria since there is now private information among players. However, observe that $\mathcal{M} \subseteq \mathcal{M}^F$ and, furthermore, that BNE coincides with SPE under public information so $\eqref{eqn:private_ADV} \geq \eqref{eqn:adv}$.

\begin{manualtheorem}{1B} Suppose that $\mu_0 \notin \Psi_{LD}(A_0) \cup \text{Bd}_{\theta^*}$, then  
\[
\eqref{eqn:private_ADV} = \eqref{eqn:adv}. 
\]
If $\mu_0 \in \Psi_{LD}(A_0)$ and further supposing $\phi$ is a convex functional, then  
\[\eqref{eqn:private_ADV} - \eqref{eqn:adv} \leq \Big(p^*(\mu_0,A_0) - p^*(\mu_0,1)\Big)\Big( \phi\big(\bm{\overline{A}}^\lambda \Big) - \phi\big(\bm{\underline{A}}^\lambda\big)\Big),\]
where $p^*(\mu,A)$ is defined in Definition \ref{defn:insideLD}.
\end{manualtheorem}
\begin{proof} The case in which $\mu_0 \notin \Psi_{LD}(A_0) \cup \text{Bd}_{\theta^*}$ follows directly from Theorem \ref{thrm:main} since it already attains the upper bound on the time-path of aggregate play. We prove the second part in several steps. 
    
\paragraph{Step 1A.} Constructing a relaxed problem. Some care is required: by moving from $\mathcal{M}^F$ to $\mathcal{M}^P$, equilibria of the resultant game might not be well-defined. We will deal with this in two ways. First, we will weaken PBE to what we call non-dominance which requires that players play action $1$ whenever it is not strictly dominated. Notice that this is not an equilibrium concept and is well-defined even with hetrogeneous beliefs. Second, we will replace the inner $\inf$ with $\sup$ to obtain the relaxed problem 
\[
\sup_{\bm{\mu} \in \mathcal{M}^P} \sup_{\bm{\sigma} \in ND(\bm{\mu},A_0)} \Ex^{\sigma} \Big[\phi\big(\bm{A}\big)\Big] \tag{ADV-P-R} \label{eqn:private_ADV_R}.
\]
It is easy to see that this is indeed a relaxed problem i.e., $\eqref{eqn:private_ADV_R} \geq \eqref{eqn:private_ADV}$ since (i) $\mathcal{M}^P \supseteq \mathcal{M}^F$ and furthermore, for each $\bm{\mu} \in \mathcal{M}^F$, $PBE(\bm{\mu},A_0) \subseteq ND(\bm{\mu},A_0)$. 

\paragraph{Step 1B.} Solving the relaxed problem. First observe that for each player $i \in [0,1]$, a necessary condition for action $1$ to not be strictly dominated is 
\[
\mu_{it} > \psi_{LD}(A = 1) 
\]
Hence, consider the strategy $\overline \sigma$ in which each player $i$ plays $1$ if $\mu_{it} > \psi_{LD}(A = 1)$ and $0$ otherwise. Clearly, 
\[
\sup_{\bm{\mu} \in \mathcal{M}^P} \Ex^{\overline \sigma} \Big[\phi\big(\bm{A}\big)\Big] \geq \eqref{eqn:private_ADV_R}.
\]
Let $(\mu_{it})_t$ be any Cadlag martingale and let $\tau_i := \inf\{t \in \mathcal{T}: \mu_{it} \notin \Psi_{LD}(1) \}$. Clearly this Cadlag martingale is improvable if it continues to deliver information after $\tau_i$, so it is without loss to consider $(\mu_{it})_t$ which are constant a.s. after $\tau_i$. But observe that since $(\mu_{it})_t$ is a martingale, the probability of exiting the region $\Psi_{LD}(1)$ is upper-bounded with the same calculation : 
\[
\mathbb{P}(\tau_i < +\infty) \leq p^*(\mu_0,1).
\]

 We define the (random) number of agents whose beliefs eventually cross $\psi_{LD}(1)$ as follows:
\begin{align*}
    F &= \int_{i \in I} 1\{\tau_i < \infty\} di. 
\end{align*}
Consider that
\begin{align*}
    \Ex_\mu [F] &= \Ex_\mu\bigg[ \int_{i \in I} 1\{\tau_i < \infty \}  di \bigg]
    = \int_{i \in I} \Pr_\mu(\tau_i<\infty) di \leq p^*(\mu_0,1).
\end{align*}
Now we will derive the upper bound of $A_t$ for each realization of $(\mu_{it})_{i,t}.$ Agent $i \in I$ takes action $1$ at time $t$ only if either
\begin{enumerate}
    \item[(I)] agent $i$'s Poisson clock ticked before $t$, and his belief eventually crosses $\psi_{LD}(1)$, or
    \item[(II)] agent $i$ took action 1 initially, and his Poisson clock has not ticked yet.
\end{enumerate}
The measures of agents in (I) and (II) are $F (1-\exp(-\lambda t))$ and $A_0 \exp(-\lambda t)$, respectively. Thus,
\begin{align*}
    A_t &\leq  F(1-\exp(-\lambda t)) + A_0 \exp(-\lambda t)
\end{align*}
almost surely. Define $\overline{\bm{A}}^\lambda := (\overline A_t^\lambda)_{t}$ as the solution to the ODE $d\overline{A}_t^\lambda = \lambda(1-\overline{A}_t^\lambda) dt $ with boundary $\overline A_0^\lambda = A_0$, and $\underline{\bm{A}}^\lambda := (\underline A_t^\lambda)_{t}$ as the solution to the ODE $d\underline{A}^\lambda_t = -\lambda \underline{A}^\lambda_t dt $ with boundary $\underline A_0^\lambda = A_0$. We have
\[\overline{A}_t^\lambda = 1 - (1-A_0)\exp(-\lambda t), \quad \quad \underline{A}_t^\lambda = A_0\exp(-\lambda t),  \]
so we can rewrite the upper bound of $A_t$ as follows:
\[A_t \leq F \overline{A}_t^\lambda + (1-F) \underline{A}_t^\lambda \quad \forall t \quad \Rightarrow \quad \bm{A} \leq F \bm{\overline{A}}^\lambda + (1-F) \bm{\underline{A}}^\lambda\]
almost surely.  
Since $\phi$ is a convex and increasing functional, we must have
\[\phi(\bm{A}) \leq F \phi(\bm{\overline{A}}^\lambda) + (1-F) \phi(\bm{\underline{A}}^\lambda)\]
almost surely. This implies
\begin{align*}
    \Ex_\mu\Big[\phi(\bm{A})\Big] &\leq \Ex_\mu [F] \phi \big(\bm{\overline{A}}^\lambda \big) + (1- \Ex_\mu [F]) \phi\big(\bm{\underline{A}}^\lambda\big)  \\
    &\leq p^*(\mu_0,1)\phi\big(\bm{\overline{A}}^\lambda \big) + \big(1- p^*(\mu_0,1)\big) \phi\big(\bm{\underline{A}}^\lambda\big)
\end{align*}
for every $\bm{\mu}.$ Thus,
\begin{align*}
\eqref{eqn:private_ADV} \leq \eqref{eqn:private_ADV_R}  \leq p^*(\mu_0,1)\phi\big(\bm{\overline{A}}^\lambda \big) + \big(1- p^*(\mu_0,1)\big) \phi\big(\bm{\underline{A}}^\lambda\big).
\end{align*}
This implies
\begin{align*}
    \eqref{eqn:private_ADV} - \eqref{eqn:adv} \leq (p^*(\mu_0,A_0) - p^*(\mu_0,1))\big( \phi\big(\bm{\overline{A}}^\lambda \big) - \phi\big(\bm{\underline{A}}^\lambda\big)\big),
\end{align*}
as desired.
\end{proof}
\section{Technical Results for Dynamic Regime Change} \label{online appendix: technical proof}
\begin{lemma} \label{lem: technical regime change}
    In the setting of dynamic regime change games (Appendix \ref{appendix:examples}), if $\gamma$ is continuously differentiable and strictly decreasing in $A$, then $\Delta U (\bm{A},\theta,\tau|t)$ is increasing and Lipschitz in $\bm{A}.$
\end{lemma}
\begin{proof}
    To show that $\Delta U (\bm A,\theta, \tau | t)$ is increasing $\bm{A}$, consider that
\begin{align*}
&\Delta U (\bm{A},\theta,\tau | t)- c\int_t^{t+\tau} e^{-r(s-t)} ds  \\
&=  \int_{s=t}^{t+\tau} e^{-r(s-t)} de^{-\int_t^s \gamma(A_{s'},\theta) ds'} \\
&=\bigg(\exp\Big({-r\tau-\int_t^{t+\tau} \gamma(A_{s'},\theta) ds'}\Big)-1\bigg) +r\int_{s=t}^{t+\tau} \exp\Big({-r(s-t)-\int_t^s \gamma(A_{s'},\theta) ds'} \Big) ds,
\end{align*}
which is increasing in $\bm{A}$ if $\gamma$ is strictly decreasing in $A$. 

Next, we show $\Delta U(\bm{A},\theta,\tau|t)$ is Lipschitz in $\bm{A}.$ Since $\gamma$ is continuously differentiable with respect to $A$ and its domain is compact, $\gamma$ is bounded and Lipschitz continuous. Suppose $L_\gamma$ is a constant such that $|\gamma(A,\theta) - \gamma(A',\theta)| \leq L_\gamma|A-A'|$ and $\gamma(A,\theta) \leq L_\gamma$ for every $\theta$ and $A,A'\in [0,1]$. This implies
\[\big| e^{-\int_t^s \gamma (A_{s'},\theta) ds'} - e^{-\int_t^s \gamma (A'_{s'},\theta) ds'}\big| \leq \Big| \int_t^s (\gamma(A_{s'},\theta) - \gamma(A'_{s'},\theta)) ds' \Big| \leq (s-t)L_\gamma \|\bm{A}-\bm{A'}\|_\infty,\]
where we used the inequality $|e^x-e^y| \leq |x-y|$ for every $x,y<0$. Thus,
\begin{align*}
&|f_{\gamma,\theta}(\bm{A},s|t) - f_{\gamma,\theta}(\bm{A'},s|t)| \\
&\leq \gamma(A_s,\theta)(s-t)L_\gamma \|\bm{A} - \bm{A'} \|_\infty + |\gamma(A_s,\theta) - \gamma(A'_s,\theta)|e^{-\int_t^s \gamma(A'_{s'},\theta) ds'}\\
&\leq (L_\gamma^2(s-t)+L_\gamma) \|\bm{A} - \bm{A'}\|_\infty,
\end{align*}
which implies
\begin{align*}
&|\Delta U (\bm{A},\theta,\tau | t) -\Delta U (\bm{A'},\theta,\tau | t)| \\
&= \Big|\int_t^{t+\tau} e^{-r(s-t)}(f_{\gamma,\theta}(\bm{A},s|t) - f_{\gamma,\theta}(\bm{A'},s|t)) \Big| \\
&\leq \|\bm{A}-\bm{A'}\|_\infty\int_t^{t+\tau}e^{-r(s-t)}(L_\gamma^2(s-t)+L_\gamma)ds \\
&\leq L^*\|\bm{A}-\bm{A'}\|_\infty,
\end{align*}
where $L^* \coloneqq \frac{L_\gamma^2}{r^2} + \frac{L_\gamma}{r}$, as desired.
\end{proof}
\end{document}